\documentclass[a4paper,12pt,twoside]{article}

\usepackage{amsfonts}
\usepackage{amsmath}
\usepackage{amsopn}
\usepackage{amsthm}
\usepackage{amssymb}
\usepackage{graphicx}
\usepackage{rotating}
\usepackage[]{natbib}
\usepackage[english]{babel}
\usepackage[utf8]{inputenc}
\usepackage[colorlinks=true,citecolor=blue,linkcolor=blue]{hyperref}
\usepackage{xspace}
\usepackage{setspace}
\usepackage{epstopdf}
\usepackage{dsfont}
\usepackage{bbm}
\usepackage{booktabs}
\usepackage[T1]{fontenc}
\usepackage{anysize}
\usepackage{multirow}
\usepackage[labelfont=bf,labelsep=period]{caption}
\usepackage{floatrow}
\newfloatcommand{capbtabbox}{table}[][\FBwidth]
\usepackage{soul}
\marginsize{1.87cm}{1.87cm}{1.87cm}{1.87cm} 
\usepackage{fancyhdr}
\RequirePackage{txfonts}
\usepackage{times}

\newtheorem{Proposition}{Proposition}

\newtheorem{Corollary}{Corollary}
\newtheorem{Theorem}{Theorem}
\newtheorem{Lemma}{Lemma}

\renewcommand{\d}{\mathrm{d}}
\DeclareMathOperator{\R}{\mathbb{R}}

\newcommand{\e}{\mathrm{e}}
\renewcommand{\P}{\mathbb{P}}
\newcommand{\bx}{\mathbf{x}}
\newcommand{\by}{\mathbf{y}}
\newcommand{\var}{\hbox{var}}
\newcommand{\bbeta}{\boldsymbol\beta}
\newcommand{\zl}{z_{\text{l}}}
\newcommand{\zr}{z_{\text{r}}}
\newcommand{\lambdar}{\lambda_{\text{r}}}
\newcommand{\lambdal}{\lambda_{\text{l}}}

\allowdisplaybreaks

\begin{document}

\def\figureautorefname{Figure}
\def\tableautorefname{Table}
\def\sectionautorefname{Section}
\def\subsectionautorefname{Section}
\def\Propositionautorefname{Proposition}
\def\Theoremautorefname{Theorem}
\def\Lemmaautorefname{Lemma}
\def\Assumptionautorefname{Assumption}
\def\Corollaryautorefname{Corollary}
\renewcommand*\footnoterule{}

\title{Robust heavy-tailed versions of generalized linear models with applications in actuarial science}

\author{Philippe Gagnon$^{1}$ and Yuxi Wang$^1$}

\maketitle

\thispagestyle{empty}

\noindent $^{1}$Department of Mathematics and Statistics, Universit\'{e} de Montr\'{e}al, Canada.

\begin{abstract}
 Generalized linear models (GLMs) form one of the most popular classes of models in statistics. The gamma variant is used, for instance, in actuarial science for the modelling of claim amounts in insurance. A flaw of GLMs is that they are not robust against outliers (i.e., against erroneous or extreme data points). A difference in trends in the bulk of the data and the outliers thus yields skewed inference and predictions. To address this problem, robust methods have been introduced. The most commonly applied robust method is frequentist and consists in an estimator which is derived from a modification of the derivative of the log-likelihood. We propose an alternative approach which is modelling-based and thus fundamentally different. It allows for an understanding and interpretation of the modelling, and it can be applied for both frequentist and Bayesian statistical analyses. The approach possesses appealing theoretical and empirical properties.
\end{abstract}


\section{Introduction}

\subsection{Generalized linear models}

Generalized linear models (GLMs) form a class of regression models introduced by \cite{nelder1972generalized} which encompasses among the most widely used statistical models, with applications ranging from actuarial science \citep{goldburd2016generalized} to medicine \citep{lindsey1998choosing}. This class generalizes normal linear regression, that is linear regression with normally distributed errors, by assuming that the distribution of the dependent variable defines an exponential family with parameters that depend on explanatory variables, and an expectation that is linear in the explanatory variables, up to a transformation. As a result, GLMs can handle both discrete and continuous dependent variables, with distribution shapes that offer flexibility regarding in particular the skewness. GLMs cover models such as, as mentioned, the classical linear regression for normally distributed responses, logistic regression for binary ones, Poisson regression for count data, gamma regression for right-skewed positive data, plus many other statistical models obtained through its general model formulation.

The use of GLMs in actuarial science can be traced back to the early 1980s with examples of the fitting of GLMs to motor-insurance data \citep{mccullagh1983generalized}. In the insurance industry, the use of gamma GLM is popular for the modelling of claim severity due to the similar characteristics of the probability density function (PDF) with its empirical version. The latter indeed often exhibits skewness and is strictly increasing up to a positive real value, after which it strictly decreases. Gamma GLM is used to determine the factors that contribute the most to the claim size and how they influence the latter, and to predict claims based on a given set of explanatory variables, ultimately leading to the pricing of insurance products.

\subsection{Robustness problems}\label{sec:robustness_problems}

In this paper, we study the robustness problems of GLMs against outliers and propose a robust version. In \autoref{fig:estimates_yn}, we show how these problems translate in practice with gamma GLM by illustrating the evolution of parameter estimates as a data point moves away from the bulk of the data. We simulated a data set of size $n = 20$ with one explanatory variable whose data points $x_{i2}$ are a standardized version of $1, \ldots, n$ ($x_{i1} = 1$ for all $i$ to introduce an intercept in the model). First, each observation of the dependent variable $y_i$ was sampled from a gamma distribution with a mean parameter of $\mu_i = \exp(\bx_i^T \bbeta)$, where $\bbeta = (\beta_1 = 0, \beta_2 = 1)^T$, and a parameter $\nu = 40$ corresponding to the inverse of the dispersion parameter. Next, we gradually increased the value of $y_n$ from $6$ (a non-outlying value) to $15$ (a clearly outlying value). For each data set associated with a different value of $y_n$, we estimated the parameters $\nu$, $\beta_1$ and $\beta_2$ of gamma GLM based on maximum likelihood method. In \autoref{fig:estimates_yn}, we also show the estimates based on the proposed method, and those based on the method of \cite{cantoni2001robust}, the latter being the most commonly applied robust method. In the following, we will present the details of both robust methods. In \autoref{fig:estimates_yn}, we observe that outliers can arbitrarily affect the maximum likelihood estimation of gamma GLM. Note that $n = 20$ and one explanatory variable (thus three parameters) together represent a situation where the sample size is moderate comparatively to the number of parameters.

  \begin{figure}[ht]
  \centering\small
  $\begin{array}{ccc}
 \vspace{-2mm}\hspace{-2mm}\includegraphics[width=0.34\textwidth]{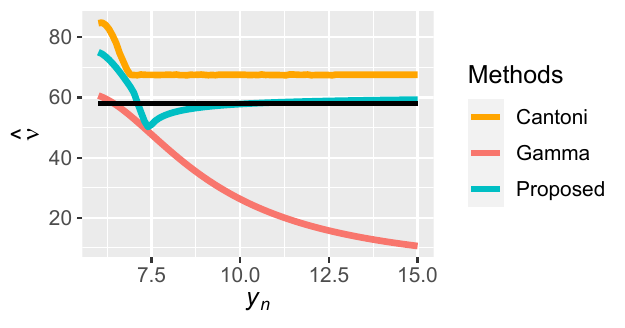} &  \hspace{-5mm} \includegraphics[width=0.34\textwidth]{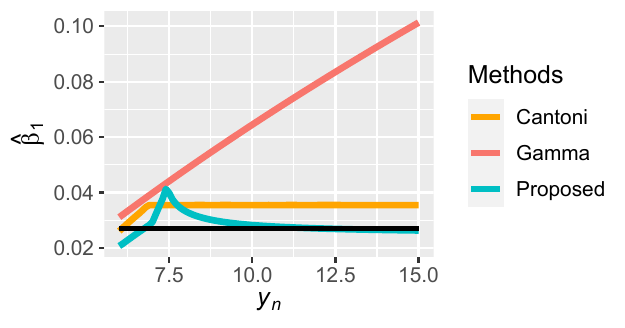} &  \hspace{-5mm} \includegraphics[width=0.34\textwidth]{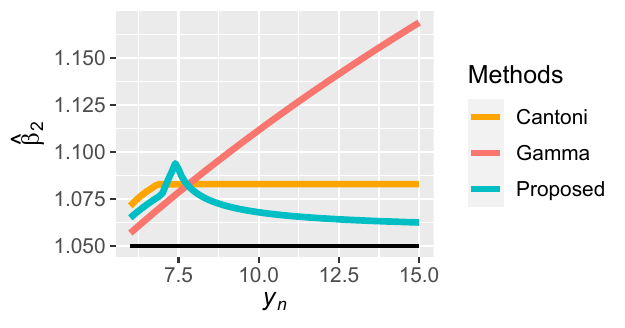} \cr
   \hspace{-5mm} \textbf{(a) $\hat{\nu}$ as a function of $y_n$} & \hspace{-7mm} \textbf{(b) $\hat{\beta}_1$ as a function of $y_n$} & \hspace{-5mm} \textbf{(c) $\hat{\beta}_2$ as a function of $y_n$}
  \end{array}$\vspace{-2mm}
  \caption{\small Estimates of $\nu$, $\beta_1$ and $\beta_2$ as a function of $y_n$ based on the method of \cite{cantoni2001robust}, gamma GLM and the proposed method (with $c = 1.6$); the black horizontal lines represent the maximum likelihood estimates of gamma GLM based on the data set excluding the outlier.}\label{fig:estimates_yn}
 \end{figure}
\normalsize

To have a concrete idea of how and why those robustness problems arise, let us discuss the example of the data bases of insurance companies and the potential impact of their use of gamma GLM. Insurance companies are not shielded from issues such as data quality, and thus the presence of erroneous data points in their data bases. Also, extreme data points are often present in their data bases. Both have a negative impact on the conclusions drawn and predictions made from statistical analyses. This is due to the non-robustness of the regression models typically employed, such as gamma GLM, against outliers, and therefore against data with gross errors and extreme data points. When using gamma GLM, the non-robustness is a consequence of the light tails of the PDF, combined with a difference in the trends in the bulk of the data and the outliers. When the likelihood function is evaluated at parameter values reflecting the trend in the bulk of the data, the light tails penalize heavily those values for the outliers, diminishing significantly the likelihood-function value. The analogous phenomenon arises when the likelihood function is evaluated at parameter values reflecting the trend in the outliers: those values are heavily penalized for the bulk of the data. All that makes values between those aforementioned more plausible, representing an undesirable compromise. The resulting maximum likelihood estimate (MLE) thus reflects neither the bulk of the data nor the outliers.

\subsection{Motivation for robust approaches, existing approaches and our proposal}\label{sec:motivation}

Another undesirable consequence of those robustness problems is that identifying outliers using standard measures based on GLM maximum likelihood estimation such as Pearson residuals (which will be defined in the following) may be ineffective. This is a result of the \textit{masking effect} \citep{hadi1993procedures}, as outliers may mask one another due to an adjustment of the model. Moreover, univariate analyses of extreme values may not allow to deal with the problem, because the notion of outlier here is with respect to the model employed, that is a combination of $\bx_i$ and $y_i$ which is unlikely under this model when using parameter values reflecting the trend in the bulk of the data. A data point can thus be an outlier with respect to the model without having any extreme values in the explanatory or dependent variables. There is an additional difficulty in situations where there are many explanatory variables and one wants to perform variable selection, because going through univariate and Pearson-residual analyses (assuming that these are effective) is simply infeasible.

All that motivates the use of robust GLMs which can automatically and effectively deal with outliers and thus offer a protection in case the data set to be analysed contains outliers. The non-robustness properties of GLMs have been studied by several authors which have proposed robust alternatives. On the frequentist side, with a focus on logistic regression, \cite{pregibon1982resistant} proposed a resistant fitting method, and \cite{stefanski1986optimally} and \cite{kunsch1989conditionally} studied optimally bounded score functions. \cite{cantoni2001robust} proposed robust estimators for all GLMs based on the notion of quasi-likelihood; their approach is to adapt the robust M-estimator for linear regression of \cite{huber1973robust}. More recently, \cite{beath2018mixture} proposed a robust alternative based on a mixture where the main component is a standard GLM, while the other, dedicated to outliers, is an over-dispersed GLM achieved by including a random-effect term in the linear predictor. This approach has the advantage of being generally applicable and to be modelling-based, implying that it can be used for both frequentist and Bayesian statistical analyses. Random-effect models are however known for being difficult to estimate; the estimation procedure is sometimes ineffective or quite slow. This is for instance the case in the simulation study conducted in this paper. The estimation based on a single data set is too long and makes the inclusion of the method in the simulation study infeasible.

On the Bayesian side, except the approach of \cite{beath2018mixture} that can be used for Bayesian analyses, we did not find any robust approach specifically for GLMs. A general approach is that of \cite{bissiri2016general} who introduced another statistical paradigm based on the premise that the model assumed is incorrect, but we consider it as another type of approaches to those considered here and will not focus on such an approach in the current document. Bayesian robust approaches typically consist in adapting the original model to the potential presence of outliers by replacing the distribution by one that is similar, but with heavier tails.

Frequentist and Bayesian typical robust methods are often seen as being fundamentally different. In former methods, the likelihood function is modified (yielding M-estimators) for the purpose of diminishing the impact of outliers, whereas in the latter, the original PDF is directly replaced by another density which, while being as similar as possible to the original one, has heavier tails. We highlight in \autoref{sec:connection_M_heavy} that a connection exists in some situations by viewing M-estimation as maximum likelihood estimation of a heavy-tailed model. Establishing such a connection provides another perspective on M-estimators. Firstly, it allows to associate a model to the estimator, which is important from a modelling point of view. Secondly, it allows the method to be used for Bayesian analyses. In \autoref{sec:connection_M_heavy}, we also highlight that, with the approach of \cite{cantoni2001robust}, it is not possible to establish a clear correspondence with a model because the function that they modify to gain in robustness is the derivative of the log-likelihood, instead of the log-likelihood. In \autoref{sec:proposed}, we present our approach in the form of a heavy-tailed distribution whose maximum likelihood estimation can be seen as M-estimation of the gamma GLM parameters. We focus on the case of gamma GLM throughout the document, but our approach is valid for any GLM based on a distribution with tails, whether it is continuous or discrete, such as inverse Gaussian and Poisson GLMs. The case of inverse Gaussian is treated in detail in \autoref{sec:robust_inv_GLM}. In \autoref{sec:proposed}, we also present theoretical results which allow to characterize the proposed model. We present sufficient conditions under which the posterior distribution for a Bayesian analysis is proper. We study the robustness properties by characterizing the behaviour of the likelihood function and posterior distribution as outliers move away from the bulk of the data. In \autoref{sec:numerical_experiments}, we evaluate the performance of the proposed approach through a simulation study and present a Bayesian case study based on a detailed analysis of a real data set. The paper finishes in \autoref{sec:discussion} with a discussion and retrospective comments. All proofs of theoretical results are deferred to \autoref{sec:proofs}. The code to reproduce all numerical results is available online (see ancillary files on \url{https://arxiv.org/abs/2305.13462}).

\section{Connection between robust estimators and heavy-tailed distributions}\label{sec:connection_M_heavy}

In \autoref{sec:connection_huber}, we take the Huber M-estimator in the context of linear regression \citep{huber1973robust} as an example to highlight that a clear connection between an M-estimator and a heavy-tailed distribution exists in some situations. In \autoref{sec:connection_cantoni}, we return to the context of gamma GLM, and explain that several PDFs yield the modified estimating equation proposed in \cite{cantoni2001robust}.

\subsection{One-to-one correspondence: Huber M-estimator in linear regression}\label{sec:connection_huber}

Consider that we have access to a data set of the form $\{\bx_i, y_i\}_{i = 1}^n$, where $\bx_1, \ldots, \bx_n \in \R^p$ are vectors of explanatory-variable data points and $y_1, \ldots, y_n \in \R$ are observations of the dependent variable. In linear regression, the random variables $Y_1, \ldots, Y_n$ are assumed to be independent (or conditionally independent under a Bayesian framework) and modelled as
\[
 Y_i = \bx_i^T \bbeta + \sigma \epsilon_i, \quad i = 1, \ldots, n,
\]
where $\bbeta = (\beta_1, \ldots, \beta_p)^T \in \R^p$ is the vector of regression coefficients, $\sigma > 0$ is a scale parameter, and $\epsilon_1, \ldots, \epsilon_n \in \R$ are standardized errors, which are assumed to be independent and identically distributed with $\epsilon_i \sim f$. In the normal linear regression model, $f = \mathcal{N}(0, 1)$. To find the MLE, we maximize the log-likelihood function, denoted by $\ell$, which is such that
\begin{align}\label{eq:logL}
 \ell(\bbeta,  \sigma) = -n \log(m) - n \log(\sigma) +\sum_{i = 1}^n \log\left\{g\left(\frac{y_i - \bx_i^T \bbeta}{\sigma}\right)\right\},
\end{align}
where we wrote $f(\epsilon) = g(\epsilon) / m$ with $g(\epsilon)= \exp(-\epsilon^2 / 2)$ and $m = \sqrt{2 \pi}$, the normalizing constant.

Maximizing this function is equivalent to minimizing
\[
 n \log(\sigma) +\frac{1}{2}\sum_{i = 1}^n \left(\frac{y_i - \bx_i^T \bbeta}{\sigma}\right)^2,
\]
if we omit the constant term. The quadratic term above produces extreme values when some residuals $y_i - \bx_i^T \bbeta$ are extreme, which is the case for outliers when the log-likelihood function is evaluated at parameter values reflecting the trend in the bulk of the data. The idea of \cite{huber1973robust} was to modify the quadratic term to deal with this problem. He proposed to instead minimize
\begin{align}\label{eq:Huber}
 n \log(\sigma) +\sum_{i = 1}^n \varrho\left(\frac{y_i - \bx_i^T \bbeta}{\sigma}\right),
\end{align}
with $\varrho$ being the Huber loss function \citep{huber1964robust}, defined as
\begin{align}
 \varrho(\epsilon) = \begin{cases}
    \epsilon^2 / 2 \quad \text{if} \quad |\epsilon| \leq k, \cr
    k|\epsilon| - k^2 / 2 \quad \text{otherwise},
 \end{cases} \label{eq:varrho}
\end{align}
where $k>0$ is a tuning parameter chosen by the user to reach a compromise between efficiency and robustness.\footnote{The value commonly used of $k = 1.345$ allows the estimator to produce 95-percent efficiency.} The penalization by the Huber loss function is quadratic, as before, between $-k$ and $k$; otherwise, the penalization is linear, which is more moderate. Note that the term $-k^2 / 2$ in \eqref{eq:varrho} is to make $\varrho$ continuous.

Replacing $f$ in the log-likelihood function yields an estimator called a \textit{maximum likelihood type estimator} (M-estimator). To establish a connection with a heavy-tailed distribution, we instead view it as the MLE of another model. In our example above, it is like viewing the minimization of \eqref{eq:Huber} as being equivalent as the maximization of \eqref{eq:logL}, but with
\[
 g(\epsilon) = \begin{cases}
  \exp(- \epsilon^2 / 2) \quad \text{if} \quad |\epsilon| \leq k, \cr
  \exp(-k|\epsilon| + k^2 / 2) \quad \text{otherwise}.
 \end{cases}
\]
The function to minimize with the Huber M-estimator can thus be viewed as being associated to a likelihood function, where the PDF $f$ involved in it has a central part which is proportional to a standard normal density and tails that behave like $\exp(-k|\epsilon|)$. The tails are thus similar to those of a Laplace density and are thus heavier than the tails of a normal density.

The Huber M-estimator is a proper example for which a clear correspondence with a heavy-tailed distribution exists. However, that is not the case for all M-estimators. For instance, it is not possible to establish a connection with a model for the Tukey's biweight M-estimator \citep{beaton1974fitting} as the loss function is constant beyond a threshold, thus yielding an improper distribution.

\subsection{The case of \cite{cantoni2001robust}}\label{sec:connection_cantoni}

In the GLM context, we will use, for this section and the rest of the paper, the same notation as in \autoref{sec:connection_huber} for the explanatory-variable data points (i.e., $\bx_1, \ldots, \bx_n$), the observations of the dependent variable (i.e., $y_1, \ldots, y_n$),  and the regression coefficients (i.e., $\bbeta = (\beta_1, \ldots, \beta_p)^T$). In the case of gamma GLM, there is an additional parameter, $\nu > 0$, that corresponds to the shape parameter, in addition to the inverse of the dispersion parameter. Note that $y_i > 0$ here as the dependent variable is positive.

The robust estimator of $\bbeta$ proposed by \cite{cantoni2001robust} corresponds to the solution of the following estimating equation:
\[
 \sum_{i = 1}^n \Psi(y_i, \bx_i, \bbeta, \nu) = \mathbf{0},
\]
where
\begin{align}
 \Psi(y_i, \bx_i, \bbeta, \nu) = \sqrt{\nu} \, \Psi_c(r_i(\bbeta, \nu)) \, \bx_i,   \label{eq:Psi}
\end{align}
with
\begin{align*}
 \Psi_c (r)= \begin{cases}
 r \quad \text{if} \quad |r| \leq c, \cr
 c \, \text{sgn}(r)  \quad \text{if} \quad |r| > c,
 \end{cases}
\end{align*}
and $r_i(\bbeta, \nu) = (y_i - \mu_i) / \sqrt{\var[Y_i]}$, $\text{sgn}(\, \cdot \,)$ being the sign function and $c>0$ a tuning parameter chosen by the user to reach a compromise between efficiency and robustness.\footnote{The value commonly used for $c$ is $c = 1.50$.} In \eqref{eq:Psi}, we set the weight function $w$ included in $\Psi$ in \cite{cantoni2001robust} and applied to each $\bx_i$ to 1 and omitted a Fisher consistency term to simplify. Note that $\mu_i$ corresponds to the mean parameter in gamma GLM and is thus such that $\mu_i = \exp(\bx_i^T \bbeta)$ when using the $\log$ link, which will be used throughout. Note also that $\var[Y_i]$ corresponds to the variance with gamma GLM and is thus such that $\var[Y_i] = \mu_i^2 / \nu$. Therefore, $r_i(\bbeta, \nu)$ becomes what is referred to as the \textit{Pearson residual} when evaluated at $\hat\bbeta$ and $\hat\nu$.

Placed in a context of M-estimator, where one wants to minimize $\sum_{i = 1}^n -\rho(y_i, \bx_i, \bbeta, \nu)$ with $-\rho$ being a loss function different from minus the log of the density, $\Psi$ can be viewed as the partial derivative of $\rho$ with respect to $\bbeta$. If we set $\rho$ as follows:
\begin{align}
 \rho(y_i, \bx_i, \bbeta, \nu) = \begin{cases}
   \ell_i(\bbeta, \nu) \quad \text{if} \quad |r_i(\bbeta, \nu)| \leq c, \cr
  -c\sqrt{\nu} \, (h(y_i) - \bx_i^T \bbeta) + h(y_i) + a_1(\nu) \quad \text{if} \quad r_i(\bbeta, \nu) > c, \cr
  c\sqrt{\nu} \, (h(y_i) - \bx_i^T \bbeta) + h(y_i) + a_2(\nu) \quad \text{if} \quad r_i(\bbeta, \nu) < -c, \cr
 \end{cases} \label{eq:rho}
\end{align}
with possibly an omitted additive term common to all cases independent of $\bbeta$, $\ell_i$ being the contribution of the $i$-th data point to the log-likelihood in gamma GLM, that is
\[
 \ell_i(\bbeta, \nu) = -\nu(y_i / \mu_i + \log \mu_i) + (\nu - 1)\log y_i + \nu \log \nu - \log \Gamma(\nu),
\]
it can be readily verified that the derivative of  $\rho(y_i, \bx_i, \bbeta, \nu)$  \eqref{eq:rho} with respect to $\bbeta$ is equal to  $\Psi(y_i, \bx_i, \bbeta, \nu)$ \eqref{eq:Psi}. The terms $a_1(\nu)$ and $a_2(\nu)$ can be used to make a corresponding PDF continuous.

There are thus many possible loss functions \eqref{eq:rho} and corresponding PDFs that can result in the estimating equation \eqref{eq:Psi}, because $h(y_i)$ which does not depend on $\bbeta$ has no influence on the derivative of $\rho$ with respect to $\bbeta$. In an attempt to establish a clear connection between the loss function in \eqref{eq:rho} and a specific heavy-tailed distribution, and to take the analysis one step further, we consider a natural choice for $h$, as we now explain.

Let us have a look at the behaviour of the right tail of the gamma PDF, as it is the one that creates the most serious robustness problems because of its exponential decay. When $y_i \rightarrow \infty$, with $\bbeta$ and $\nu$ fixed, the dominant term of $\ell_i(\bbeta, \nu)$ in gamma GLM is
\[
 -\nu(y_i / \mu_i) = -\nu \exp\left\{\log(y_i) - \bx_i^T\bbeta\right\}.
\]
To retrieve a similar form in the function \eqref{eq:rho} when $r_i(\bbeta, \nu) > c$ (which is the part of the function that is activated when $y_i$ is large, and $\bbeta$ and $\nu$ are fixed), $h$ should be set to the $\log$ function. With this function, the PDF $f_{\bbeta, \nu, c}$ of the dependent variable based on $\rho$ in \eqref{eq:rho} is such that (when the estimator is viewed as the MLE of a different model instead of an M-estimator)
\[
 f_{\bbeta, \nu, c}(y_i) = \exp(\rho(y_i, \bx_i, \bbeta, \nu)) = \frac{1}{\mu_i} f_{\nu, c}\left(\frac{y_i}{\mu_i}\right) \propto  \frac{1}{\mu_i} g_{\nu, c}\left(\frac{y_i}{\mu_i}\right),
\]
where $f_{\nu, c}$ is the PDF of $Y_i / \mu_i$ which does not depend on $\bbeta$ and $g_{\nu, c}$ is the unnormalized version of the latter defined as
\begin{align}\label{eqn:g_cantoni}
 g_{\nu, c}(z) := \begin{cases}
  \exp(-\nu z) \, z^{\nu - 1} \nu^\nu / \Gamma(\nu) \quad \text{if} \quad |\sqrt{\nu} \, (z - 1)| \leq c, \cr
  z^{-c\sqrt{\nu} - 1} \exp(-a_1(\nu)) \quad \text{if} \quad \sqrt{\nu} \, (z - 1) > c, \cr
  z^{c\sqrt{\nu} - 1} \exp(-a_2(\nu)) \quad \text{if} \quad \sqrt{\nu} \, (z - 1) < -c.
 \end{cases}
\end{align}

Let us now discuss the characteristics of the function $g_{\nu, c}$. The left part (the third case in \eqref{eqn:g_cantoni}) may not exist (in the sense of never be activated): given that $z > 0$, the left part exists when $-c / \sqrt{\nu} + 1 > 0$, which is equivalent to $c < \sqrt{\nu}$. This means that the left part may exist even when $\nu \leq 1$ which is counterproductive given that in this case the original gamma PDF does not converge to $0$ as $z \rightarrow 0$ (it converges to a constant when $\nu = 1$ and goes to infinity when $\nu < 1$); the gamma PDF has, in a sense, no left tail in this case. Note that the analogous fact is true regarding the estimator of \cite{cantoni2001robust} (recall \eqref{eq:Psi}), meaning that it may be the case that $\Psi(y_i, \bx_i, \bbeta, \nu) = \sqrt{\nu} \, c \, \text{sgn}(r_i(\bbeta, \nu)) \, \bx_i$ because $r_i(\bbeta, \nu) < -c$, even when $\nu \leq 1$.

In order to understand clearly the difference with gamma GLM in terms of tail behaviour, let us have a close look at the two tails of $g_{\nu, c}$ separately. We compare the latter with a gamma PDF with a mean parameter of $1$, corresponding to the central part of $g_{\nu, c}$. On the right side, when $z \rightarrow \infty$, with $\nu$ fixed, the dominant term of the gamma PDF is $\exp(-\nu z)$; the decrease is thus exponential and faster than the polynomial decrease of $g_{\nu, c}$. On the left side, when $z \rightarrow 0$, both PDFs have essentially the same behaviour. When $0 < \nu < 1$, the dominant term of the gamma PDF, which is $z^{\nu - 1}$, increases polynomially, and it is the same for $g_{\nu, c}$ (we can show that when $0 < \nu < 1$, $c\sqrt{\nu} - 1 < 0$ when the left part exists). When $\nu > 1$, the gamma PDF decays polynomially when $z \rightarrow 0$ and it is the same for $g_{\nu, c}$ (at least when $c \geq 1$).

There are three flaws with the model $f_{\bbeta, \nu, c}$ presented above and the associated M-estimator: i) its central part does not match that of a gamma PDF, but is proportional to it, which may negatively affect the efficiency in the absence of outliers, ii) the left part may exist when not useful, and iii) the tail decay is arguably not slow enough (we return to this point in \autoref{sec:proposed}), which may provide an explanation for the bounded, but not redescending, influence of outliers that was observed in \autoref{fig:estimates_yn}. In \autoref{sec:proposed}, we propose a robust alternative that is similar in essence to $f_{\bbeta, \nu, c}$ but does not have those three flaws.

\section{Robust heavy-tailed versions of GLMs}\label{sec:proposed}

In this section, we present our proposal to gain in robustness in statistical analyses based on GLMs. We start in \autoref{sec:model} with an alternative model definition and next discuss theoretical properties characterizing the approach in \autoref{sec:properties}.

\subsection{Model definition}\label{sec:model}

Our proposal is rooted in a line of research called \textit{resolution of conflict} that studies how conflicting sources of information are dealt with by Bayesian models. In this line of research, an outlier is seen as a source of information that is in conflict with others. The sources with which it is in conflict represent, among others, the non-outliers. Here, we consider that the prior distribution (in a Bayesian analysis) is not in conflict with the non-outliers to simplify. That line of research was started by \cite{de1961bayesian} with a first analysis in \cite{lindley1968choice}, followed by an introduction of a formal theory in \cite{dawid1973posterior}, \cite{hill1974coherence} and \cite{o1979outlier}. For a review of Bayesian heavy-tailed
models and conflict resolution, see \cite{o2012bayesian}. In the latter paper, it is noted that
there exists a gap between the models formally covered by the theory of conflict resolution and models
commonly used in practice. The present paper contributes to the expansion of the theory of conflict
resolution by covering models used in practice, namely GLMs.

The reason why that gap exists is because it is notoriously difficult to study models from a point of view of conflict resolution, even simple location--scale models; see, e.g., \cite{o1979outlier}, \cite{angers2007conflicting}, \cite{andrade2011bayesian}, \cite{desgagne2013full} and \cite{desgagne2015robustness}. The work of \cite{desgagne2015robustness} introduced an analysis technique and paved the way to the studying of more complex models, like linear regressions \citep{DesGag2019, gagnon2020, gagnon2020PCR, 10.1214/22-BA1330, gagnon2023theoretical, hamura2020log, hamura2023posterior}, and Poisson and negative binomial regressions \citep{hamura2021robust}. The work of \cite{desgagne2015robustness} also showed that polynomial tails are not heavy enough to yield a desirable property called \textit{whole robustness} (which is defined precisely in \autoref{sec:properties}), at least for the location--scale model; the same was shown to be true in linear regression in \cite{gagnon2023theoretical}. \cite{desgagne2015robustness} proved that, for a location--scale model, it is sufficient to assume that the PDF has tails which are log-regularly varying, a concept introduced in that paper. The author proposed a PDF which satisfies this condition; it is called the \textit{log-Pareto-tailed normal} (LPTN) distribution as the central part of this continuous PDF coincides with that of the standard normal and the tails are log-Pareto, meaning that they behave like $(1 / |z|)(1 / \log|z|)^\lambda$ with $\lambda >1$. This approach was subsequently adapted to the context of linear regression by \cite{gagnon2020}, where the error distribution is assumed to be LPTN instead of normal, and whole robustness was shown to hold.

With this work, we take one step further by adapting the approach to GLMs: the distribution of the dependent variable is a modified version where the central part is kept as is, while the extremities are replaced by log-Pareto tails. Focusing on gamma GLM, we assume that $Y_i \sim f_{\bbeta, \nu, c}$ with $Y_i / \mu_i \sim f_{\nu, c}$ (we use the same notation as in \autoref{sec:connection_cantoni} to simplify), where the proposed PDF  $f_{\nu, c}$ is defined as
\begin{align}\label{eq:proposed}
 f_{\nu, c}(z) := \begin{cases}
    f_{\text{mid}}(z) := \exp(-\nu z) z^{\nu - 1} \nu^\nu / \Gamma(\nu) \quad \text{if} \quad \zl \leq z \leq \zr, \cr
    f_{\text{right}}(z) := f_{\text{mid}}(\zr) \frac{\zr}{z} \left(\frac{\log \zr}{\log z}\right)^{\lambdar} \quad \text{if} \quad z > \zr, \cr
    f_{\text{left}}(z) := f_{\text{mid}}(\zl) \frac{\zl}{z} \left(\frac{\log \zl}{\log z}\right)^{\lambdal} \quad \text{if} \quad 0 < z < \zl,
 \end{cases}
\end{align}
where $\zr, \lambdar, \zl$ and $\lambdal$ are functions of $\nu > 0$ and $c > 0$ given by
\begin{align*}
 & \zr := 1 + c / \sqrt{\nu}, \quad \zl := \begin{cases}
                                                                0 \quad \text{if} \quad \nu \leq 1, \cr
                                                                \max\{0, 1 - c / \sqrt{\nu}\} \quad \text{if} \quad \nu > 1,
                                                            \end{cases} \cr
 & \lambdar := 1 + \frac{f_{\text{mid}}(\zr) \log(\zr) \, \zr}{\P(Z_\nu > \zr)}, \quad \text{and} \quad \lambdal := 1 - \frac{f_{\text{mid}}(\zl) \log(\zl) \, \zl}{\P(Z_\nu < \zl)} = 1 + \frac{f_{\text{mid}}(\zl) \log(1 / \zl) \, \zl}{\P(Z_\nu < \zl)},
\end{align*}
with $Z_\nu$ being a random variable following a gamma distribution whose mean and shape parameters are 1 and $\nu$, respectively.

We now make a few remarks about the model. First, $\zr > 1$ and thus the log terms in $f_{\text{right}}$ are positive. Also, $f_{\text{left}}$ is activated for some value of $z$ when $\zl > 0$, that is when $c < \sqrt{\nu}$ and $\nu > 1$, and $\zl$ is upper bounded by 1. This implies that both log terms in $f_{\text{left}}$ are negative and thus that $f_{\text{left}}(z) > 0$ when $0 < z < \zl$. The constraint that $\zl = 0$ if $\nu \leq 1$ is to ensure that $f_{\text{left}}$ is never activated when the original gamma PDF does not have a left tail.

The terms $\zl$ and $\zr$, depending on $\nu$ and $c$, control which part of the function is activated. The terms  $f_{\text{mid}}(\zr)$, $\zr$ and $\log \zr$ in $f_{\text{right}}$, as well as $f_{\text{mid}}(\zl)$, $\zl$ and $\log \zl$ in $f_{\text{left}}$ ensure that the PDF is continuous. The function $f_{\nu, c}$ is integrable for all $c, \nu > 0$. It goes to $+\infty$ as $z \rightarrow 0$, when $f_{\text{left}}$ exists. This behaviour close to 0 allows to have integrals that are similar to those on the right tails, and that are to be contrasted with those under the original gamma PDF given that the latter function goes to 0 as $z \rightarrow 0$ (when it has a left tail). Indeed, an integral from 0 to a small value $a$ can be rewritten as
\[
 \int_0^a f_{\text{mid}}(\zl) \frac{\zl}{z} \left(\frac{\log \zl}{\log z}\right)^{\lambdal} \d z = \int_{1 / a}^\infty  f_{\text{mid}}(\zl) \frac{\zl}{u} \left(\frac{\log(1 / \zl)}{\log(u)}\right)^{\lambdal} \d u.
\]
After the change of variables $u = 1 / z$, the mass associated to the left tail can be viewed as an integral from $1 / a$ to $\infty$, with respect to a function which is similar to $f_{\text{right}}$, but with a different normalizing constant and a different power term. In other words, the behaviour of $f_{\text{left}}$ is analogous to that of $f_{\text{right}}$.

Comparisons between gamma PDFs (with mean and shape parameters of 1 and $\nu$, respectively) and $f_{\nu, c}$ with $c = 1.6$ are shown for different values of $\nu$ in \autoref{fig:gamma_prop}. We observe that both PDFs are globally quite similar, but beyond the threshold at which they start to be defined differently, $f_{\nu, c}$ first decreases slightly faster for a short interval (a consequence of the continuity of the function with a constraint of integrating to 1), after which $f_{\nu, c}$ goes above the gamma PDF. The length of that interval shortens as $\nu$ increases. Note that in \autoref{fig:gamma_prop} (b)-(c), we do not see that $f_{\nu, c}(z) \rightarrow \infty$ as $z \rightarrow 0$ because this explosive behaviour happens too close to 0 to be observed.

  \begin{figure}[ht]
  \centering\small
  $\begin{array}{ccc}
 \vspace{-2mm}\hspace{-2mm}\includegraphics[width=0.34\textwidth]{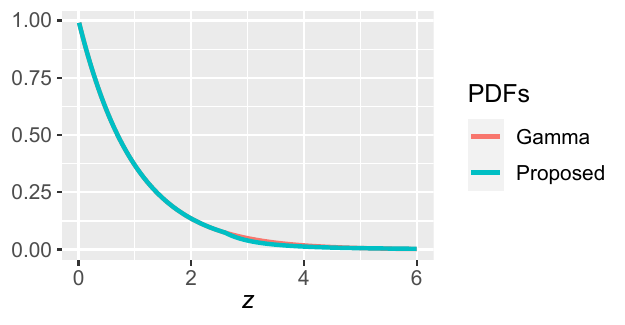} &  \hspace{-5mm} \includegraphics[width=0.34\textwidth]{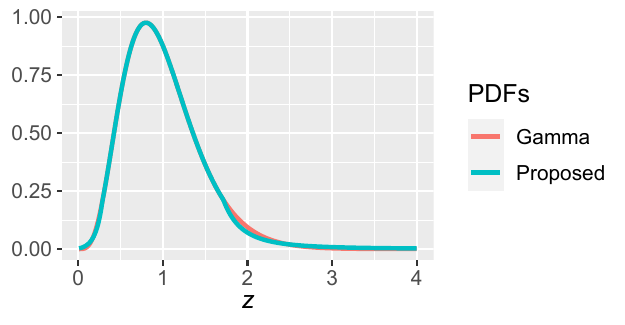} &  \hspace{-5mm} \includegraphics[width=0.34\textwidth]{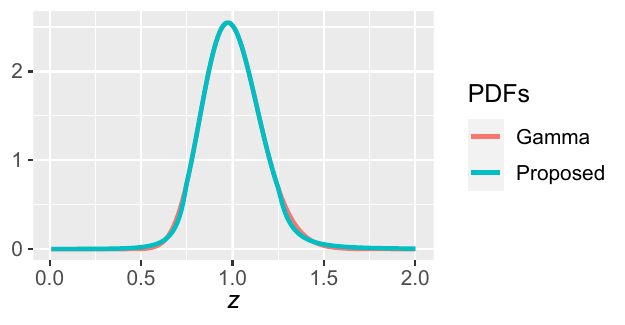} \cr
   \hspace{-5mm} \textbf{(a) $\nu = 1$} & \hspace{-7mm} \textbf{(b) $\nu = 5$} & \hspace{-5mm} \textbf{(c) $\nu = 40$} \cr
   \vspace{-2mm}\hspace{-2mm}\includegraphics[width=0.34\textwidth]{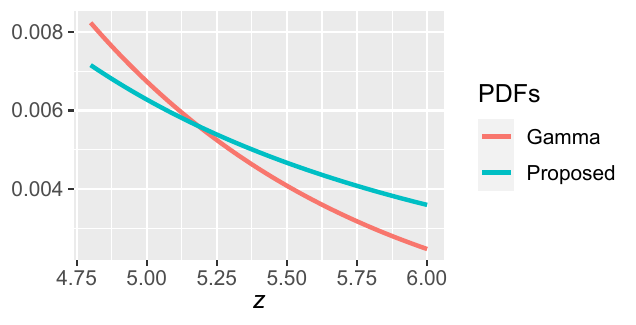} &  \hspace{-5mm} \includegraphics[width=0.34\textwidth]{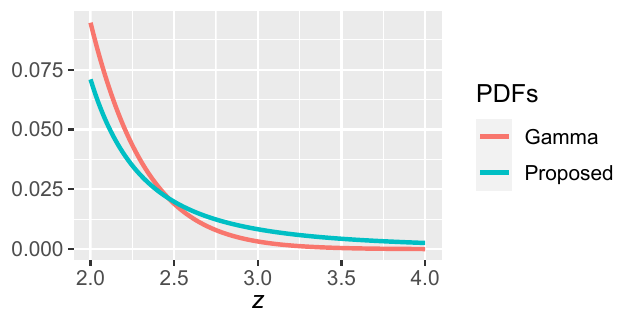} &  \hspace{-5mm} \includegraphics[width=0.34\textwidth]{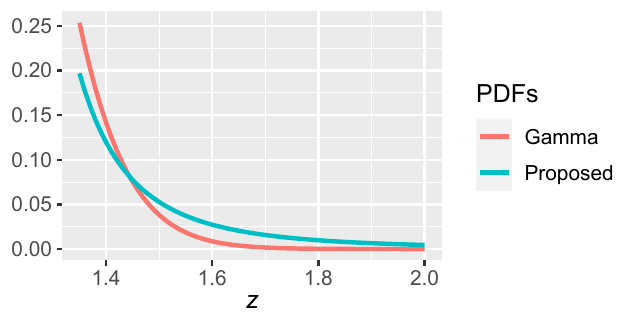} \cr
   \hspace{-5mm} \textbf{(d) $\nu = 1$ (zoom in right tail)} & \hspace{-7mm} \textbf{(e) $\nu = 5$ (zoom in right tail)} & \hspace{-5mm} \textbf{(f) $\nu = 40$ (zoom in right tail)} \cr
  \end{array}$\vspace{-2mm}
  \caption{\small Comparisons between gamma PDFs and $f_{\nu, c}$ with $c = 1.6$, for different values of $\nu$.}\label{fig:gamma_prop}
 \end{figure}
\normalsize

The exponents $\lambdar$ and $ \lambdal$ in $f_{\nu, c}$ play an important role: they make the function $f_{\nu, c}$ a PDF. In particular, $\int_{\zr}^\infty f_{\nu, c}(z) \, \d z = \P(Z_\nu > \zr)$, and, when $\zl > 0$, $\int_{0}^{\zl} f_{\nu, c}(z) \, \d z = \P(Z_\nu < \zl)$. In \autoref{fig:lambdas}, we highlight, with a log-scale on the $y$-axis, that $\lambdar$ and $ \lambdal$ are well defined and can be computed for any $\nu$ and $c$ (provided that $f_{\text{left}}$ exists in the case of $ \lambdal$). \autoref{fig:lambdas} also allows to show that $\lambdar$ and $ \lambdal$ have an interesting asymptotic behaviour as $\nu \rightarrow \infty$, as indicated by \autoref{prop:asymptotic_lambdas}. In the latter, we use $\Phi$ to denote the cumulative distribution function of a standard normal distribution.

\begin{Proposition}\label{prop:asymptotic_lambdas}
 Viewed as functions of $\nu$, both $\lambdal$ and $\lambdar$ converge, as $\nu \rightarrow \infty$ for any fixed $c$, towards
\[
 1 + \frac{c \, \e^{-c^2 / 2}}{\sqrt{2 \pi} \, (1 - \Phi(c))}.
\]
\end{Proposition}

  \begin{figure}[ht]
  \centering
\includegraphics[width=0.45\textwidth]{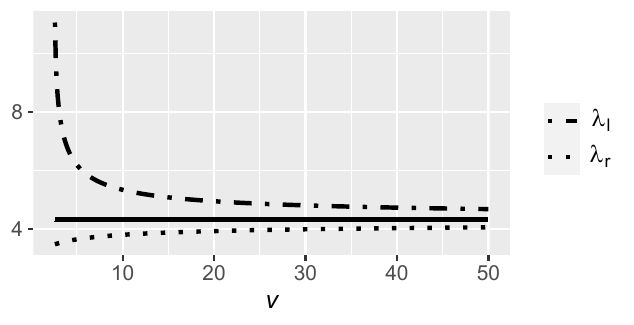}
\vspace{-2mm}
  \caption{\small $\lambdar$ and $ \lambdal$ as a function of $\nu$ when $c = 1.6$; the scale on the $y$-axis is
logarithmic; the black horizontal line represents the asymptotic value as $\nu \rightarrow \infty$.}\label{fig:lambdas}
 \end{figure}
\normalsize

In terms of estimation, the proposed model can be estimated by the maximum likelihood method (the results shown in \autoref{fig:estimates_yn} have been produced using this method). The MLE can be viewed as a robust M-estimator of gamma GLM, with an estimating equation having the same form as that proposed by \cite{cantoni2001robust} (we return to this point in \autoref{sec:properties}). Estimating the proposed model for a frequentist statistical analysis has thus the same conceptual complexity; the computational complexity is similar. From that perspective of M-estimation of gamma GLM, the estimating equation corresponding to the derivative of the log-likelihood of the proposed model can be modified to include a correction term to ensure Fisher consistency and a weight function can be applied to the vectors $\bx_i$ to decrease the influence of high-leverage points, in the same spirit as the method of \cite{cantoni2001robust}. As mentioned, one of the advantages of our approach is that it can also be applied to perform robust Bayesian analyses. Markov chain Monte Carlo methods can be employed to obtain posterior means, medians, credible intervals, and so on. We discuss Bayesian estimation in detail in \autoref{sec:case_study}.

We finish this section with two remarks about the tuning parameter $c$. Firstly, it plays the same role as the parameter with the same notation $c$ in the method of \cite{cantoni2001robust} presented in \autoref{sec:connection_cantoni}: it allows the user to reach a compromise between efficiency and robustness, and the conditions in \eqref{eq:proposed} to determine which part of the function is activated can be rewritten like those in \eqref{eqn:g_cantoni}. Secondly, there is a correspondence between the value of $c$ and the mass under $f_{\nu, c}$ assigned to the part where the density exactly matches the gamma PDF. For example, when $c = 1.6$ and $\nu = 36.3$ (the value used for $c$ and estimated for $\nu$ in the real-data example in \autoref{sec:case_study}), the mass of the central part is $\P(\zl \leq Z_\nu \leq \zr) \approx 0.89 $. This correspondence can be exploited to guide the choice of $c$, if one has prior belief about $\nu$. If, for instance, one believes that $\nu$ should take values around 40, and one wants 90\% of the mass to be assigned to the central part, one should set $c$ to $1.65$. In order to recommend an objective and effective choice of value for $c$ in case one does not have prior belief about $\nu$ or wants to use an automated approach for selecting a value for the tuning parameter $c$, we evaluate the estimation performance for several values of $c$ in our simulation study in \autoref{sec:simulation}. We identify that $c = 1.6$ offers a good balance between efficiency and robustness, at least in the scenarios evaluated. The choice of value for this parameter can also be completely data driven; it can be included as a parameter like $\bbeta$ and $\nu$ and estimated using the maximum likelihood method in frequentist analyses. A fully Bayesian approach can also be applied where $c$ would be considered as unknown and a random variable like the other parameters. In our numerical experiments, we consider it as fixed and as a tuning parameter to simplify.

\subsection{Theoretical properties}\label{sec:properties}

The theoretical results presented in this section assume that all explanatory variables are continuous to simplify. The first result that we present is crucial for Bayesian analyses. In our Bayesian framework, we consider that the explanatory-variable data points $\bx_1, \ldots, \bx_n$ are fixed and known, that is not realizations of random variables, contrarily to $y_1, \ldots, y_n$. The posterior distribution is thus conditional on the latter only. To use the proposed model for a Bayesian analysis, we need to select a prior distribution for $\bbeta$ and $\nu$, denoted by $\pi(\, \cdot \,, \cdot \,)$. Importantly, we have to make sure that the resulting posterior distribution is proper given that any Bayesian analysis assumes this. We will present a proposition providing sufficient conditions. Beforehand, we introduce notation. Let $\pi(\, \cdot \,, \cdot \mid \by)$ be the posterior distribution, where $\by := (y_1, \ldots, y_n)^T$. It is such that
\[
 \pi(\bbeta, \nu \mid \by) = \pi(\bbeta, \nu) \left[\prod_{i = 1}^n \frac{1}{\mu_i} f_{\nu, c}\left(\frac{y_i}{\mu_i}\right) \right] \Bigg/ m(\by), \quad \bbeta \in \R^p, \nu > 0,
\]
where
\[
 m(\by) := \int_{\R^p}\int_0^\infty \pi(\bbeta, \nu) \left[\prod_{i = 1}^n \frac{1}{\mu_i} f_{\nu, c}\left(\frac{y_i}{\mu_i}\right) \right] \d\nu \, \d\bbeta,
\]
if $m(\by) < \infty$, a situation where the posterior distribution is proper and thus well defined. We use $\pi(\, \cdot \mid \nu)$ to denote the conditional prior density of $\bbeta$ given $\nu$ and $\pi(\, \cdot \,)$ to denote the marginal prior density of $\nu$.

\begin{Proposition}\label{prop:proper}
 Assume that $\pi(\bbeta \mid \nu) \leq B$, for any $\bbeta$ and $\nu$, $B$ being a positive constant. Assume that $n \geq p \geq 1$. If $\pi(\, \cdot \,)$ is a proper PDF such that $\int_0^\infty \nu^{(n - p)/ 2} \, \pi(\nu) \, \d\nu < \infty$, then the posterior distribution is proper.
\end{Proposition}

The assumptions on the prior are weak, which explains why we require $n \geq p$, a condition similar to that for frequentist inference. The condition on $\pi(\, \cdot \mid \nu)$ is satisfied by any continuous PDF and by Jeffreys prior $\pi(\, \cdot \mid \nu) \propto 1$. The condition on $\pi(\, \cdot \,)$ is satisfied if, for instance, the prior on $\nu$ is a gamma distribution with any shape and scale parameters.

We now turn to the characterization of the robustness of the proposed approach against outliers. We characterize the robustness in an asymptotic regime where outliers are considered to be further and further away from the bulk of the data. As mentioned in \autoref{sec:motivation}, an outlier is defined as a couple $(\bx_i, y_i)$ whose components are incompatible with the trend in the bulk of the data. We can use $r_i(\bbeta, \nu) = \sqrt{\nu}(y_i / \mu_i - 1)$ to evaluate this incompatibility. It can be extreme because, for a given $\bx_i$ (yielding $\mu_i = \exp(\bx_i^T \bbeta)$), the value of $y_i$ makes it extreme or because, for a given $y_i$, the value of $\bx_i$ makes it extreme. We mathematically represent such extreme situations by considering an asymptotic scenario where the outliers move away from the bulk of the data along particular paths (see \autoref{fig:paths}). More precisely, we consider that the outliers $(\bx_i, y_i)$ are such that $y_i \rightarrow \infty$ or $y_i \rightarrow 0$ with $\bx_i$ being kept fixed (but perhaps extreme). Our results are asymptotic, meaning here that, for the outlying data points with fixed $\bx_i$ (but perhaps extreme), there exist $y_i$ values such that the results hold approximately.

  \begin{figure}[ht]
  \centering
\includegraphics[width=0.45\textwidth]{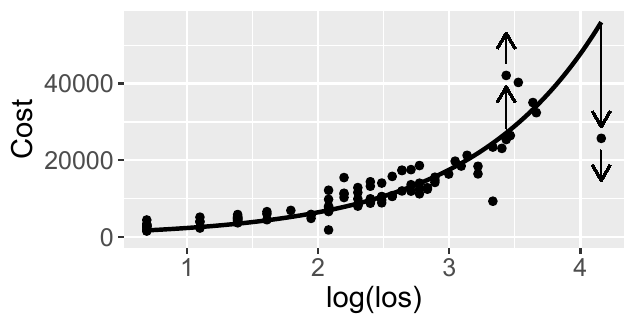}
\vspace{-2mm}
  \caption{\small Couples formed of data points of the dependent variable (cost of stay) and of an explanatory variable (log of length of stay) in the real-data example of \autoref{sec:case_study}; the black line represents an estimated exponential trend.}\label{fig:paths}
 \end{figure}
\normalsize

We refer to a couple $(\bx_i, y_i)$ with $y_i \rightarrow \infty$ as a \textit{large} outlier, and a couple with $y_i \rightarrow 0$ as a \textit{small} outlier. The $y_i$ component is referred to as a \textit{large/small} outlying observation. We consider that each outlying observation goes to $\infty$ or 0 as its own specific rate. More specifically, for a large outlying observation, we consider that $y_i = b_i \omega$, and that $y_i = 1 / b_i \omega$ for a small outlying observation, with $b_i \geq 1$ a constant, and we let $\omega \rightarrow \infty$. For each non-outlying observation, we assume that $y_i = a_i$, where $a_i > 0$ is a constant.  Among the $n$ observations $y_1, \ldots, y_n$, we assume that $k$ of them form a group of non-outlying observations, $s$ of them form a group of small outlying observations, and $l$ of them form a group of large outlying observations. We denote the set of non-outlying observations, small outlying observations, and large outlying observations as $\by_k, \by_s$ and $\by_l$, respectively. For $i=1, \ldots, n$, we define the binary functions $k_i$, $s_i$ and $l_i$ as follows: $k_i=1$ if $y_i$ is a non-outlying observation, $s_i=1$ if it is a small outlying observation, and $l_i=1$ if it is a large outlying observation. These functions take the value of 0 otherwise. Therefore, we have $k_i+s_i+l_i=1$ for $i=1, \ldots, n$, with $\sum_{i=1}^n k_i=k$, $\sum_{i=1}^n s_i=s$, and $\sum_{i=1}^n l_i = l$.

Central to the characterization of the robustness of the proposed approach is the limiting behaviour of the PDF evaluated at an outlying data point. \autoref{prop:limit_PDF} below is about this limiting behaviour.

\begin{Proposition}\label{prop:limit_PDF}
For any $i$ with $l_i = 1$, and $c$, $\nu$ and $\mu_i$ fixed, we have that
\[
 \lim_{\omega \rightarrow \infty} \frac{ f_{\nu,c}(y_i/\mu_i)/\mu_i}{f_{\nu,c}(y_i)} = 1.
 \]
 If $\nu>1$ and $c<\sqrt{\nu}$ (the condition under which $f_{\text{left}}$ exists), the same result holds for any $i$ with $s_i = 1$.
\end{Proposition}

\autoref{prop:limit_PDF} suggests that the PDF term of an outlier in the likelihood function or the posterior density behaves in the limit like $f_{\nu,c}(y_i)$. This is made precise in results below. The term $f_{\nu,c}(y_i)$ is independent of $\bbeta$ but depends on $\nu$. It is thus treated as a constant in the likelihood function or posterior density when varying $\bbeta$ with $\nu$ fixed, but not when varying $\nu$. We thus say that the conflicting information (the outlier) is \textit{partially rejected}. Our approach is thus said to be \textit{partially robust}. Ideally, conflicting information is \textit{wholly rejected} as its source becomes increasingly remote \citep{1984west431}. Note that the tail thickness of $f_{\nu,c}$ is already extreme (with a density not integrable if we omit the log terms in $f_{\text{left}}$ and $f_{\text{right}}$), and thus, it does not seem possible to remedy the situation by considering a density with heavier tails without exceedingly increasing the complexity of the model. Note also that with a polynomial tail, such as that of the density identified from the estimator of \cite{cantoni2001robust} in \autoref{sec:connection_cantoni}, it is not possible to get rid of $\bbeta$ in the limiting regime, implying a weaker robustness property. This provides an explanation for the difference in behaviour between the estimators as observed in \autoref{fig:estimates_yn}.

A corollary of \autoref{prop:limit_PDF} is the characterization of the limiting behaviour of the likelihood function.
\begin{Corollary}\label{cor:likelihood}
 The likelihood function $\prod_{i = 1}^n f_{\nu,c}(y_i/\mu_i)/\mu_i$, when evaluated at $(\bbeta, \nu)$ such that $\nu>1$ and $c<\sqrt{\nu}$, asymptotically behave like
 \begin{align}\label{eq:limit_likelihood}
  \prod_{i = 1}^n \left[\frac{1}{\mu_i} f_{\nu, c}\left(\frac{y_i}{\mu_i}\right)\right]^{k_i} \left[f_{\nu,c}(y_i)\right]^{s_i + l_i},
 \end{align}
 as $\omega \rightarrow \infty$, implying that, if the MLE belongs to a compact set with $\nu>1$ and $c<\sqrt{\nu}$, then it corresponds asymptotically to the mode of \eqref{eq:limit_likelihood}, provided that the latter belongs to a compact set with $\nu>1$ and $c<\sqrt{\nu}$ as well.
\end{Corollary}

The function in \eqref{eq:limit_likelihood} can be seen as the likelihood function based on the non-outliers only that is adjusted by a factor of $\prod_{i = 1}^n \left[f_{\nu,c}(y_i)\right]^{s_i + l_i}$ coming from the outliers. To have an idea of the impact of this additional factor, we can return to our data set simulated and discussed in \autoref{sec:robustness_problems}, and consider $y_n = 15$ for the outlying-observation value, the latter being the maximum value for which the parameter estimates are computed in \autoref{fig:estimates_yn}. We can compare $f_{\nu,c}(y_n)$ (viewed as a function of $\nu$) with the likelihood function based on non-outlying data points $(\bx_1, y_1), \ldots, (\bx_{n - 1}, y_{n - 1})$, that is $\prod_{i = 1}^{n - 1} f_{\nu,c}(y_i/\mu_i)/\mu_i$  (evaluated at $\bbeta = \hat\bbeta$). We observe in \autoref{fig:f_nu} (left panel) that the function $f_{\nu,c}(y_n)$ decreases, but less quickly than $\prod_{i = 1}^{n - 1} f_{\nu,c}(y_i/\mu_i)/\mu_i$ increases (right panel), explaining why the resulting MLE is not so influenced by the outlier.

  \begin{figure}[ht]
    \centering\scriptsize
  $\begin{array}{cc}
 \includegraphics[width=0.45\textwidth]{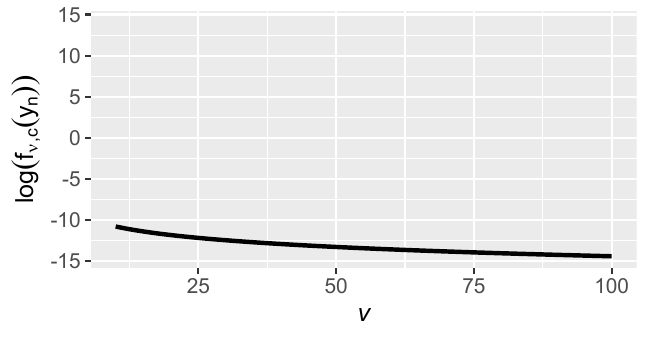} &  \includegraphics[width=0.45\textwidth]{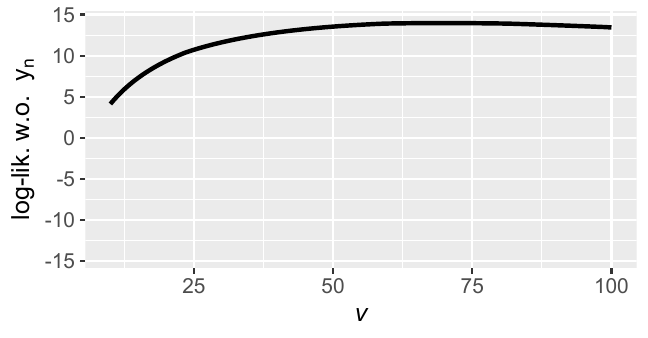}
  \end{array}$\vspace{-2mm}
  \caption{\small Log of $f_{\nu,c}(y_n)$ and log-likelihood function based on data points $(\bx_1, y_1), \ldots, (\bx_{n - 1}, y_{n - 1})$ simulated in \autoref{sec:robustness_problems} evaluated at $\bbeta = \hat\bbeta$, viewed as a function of $\nu$, with $c = 1.6$ and $y_n = 15$.}\label{fig:f_nu}
 \end{figure}
\normalsize

With results such as \autoref{prop:limit_PDF} and \autoref{cor:likelihood}, one may wonder if the influence function of $\bbeta$ is redescending, in the sense of being asymptotically null. We now explain that it is the case by writing the estimating equation of $\bbeta$ in the same way as that of \cite{cantoni2001robust} (recall \eqref{eq:Psi}) and by showing that the function applied to the Pearson residuals is redescending. A difference is that we allow this function to depend on $\nu$.
\begin{Proposition}\label{prop:estimating}
With the proposed model, the estimating equation regarding $\bbeta$ is the following:
\begin{align}\label{eqn:estimating_proposed}
 \frac{\partial}{\partial \bbeta} \sum_{i=1}^n \log f_{\bbeta, \nu, c}(y_i) = \sum_{i=1}^n \sqrt{\nu} \, \Psi_{\nu, c}(r_i(\bbeta, \nu)) \, \bx_i = \mathbf{0},
\end{align}
where,
\[
 \Psi_{\nu, c}(r) := \begin{cases}
     r \quad \text{if} \quad |r| \leq c, \cr
  \frac{\lambdar / \sqrt{\nu}}{\log(r / \sqrt{\nu} + 1)}  \quad \text{if} \quad r > c, \cr
  \frac{\lambdal / \sqrt{\nu}}{\log(r / \sqrt{\nu} + 1)}  \quad \text{if} \quad r < -c,
 \end{cases}
 \]
 if $\nu>1$ and $c<\sqrt{\nu}$; otherwise, the third case is omitted in the definition of $\Psi_{\nu, c}$.
\end{Proposition}
For any $i$ with $l_i = 1$, and $c$, $\nu$ and $\mu_i$ fixed, we have that $r_i(\bbeta, \nu) \rightarrow \infty$ as $\omega \rightarrow \infty$, implying that $\Psi_{\nu, c}(r_i(\bbeta, \nu)) \rightarrow 0$. For any $i$ with $s_i = 1$, we have that $r_i(\bbeta, \nu) \rightarrow -\sqrt{\nu}$ as $\omega \rightarrow \infty$, implying that $\Psi_{\nu, c}(r_i(\bbeta, \nu)) \rightarrow 0$ as well if $\nu>1$ and $c<\sqrt{\nu}$ (the condition under which $f_{\text{left}}$ exists). Note that the function $\Psi_{\nu, c}$ is in general not continuous, which a difference with the functions typically used in the robust frequentist literature. The reason is that, to simplify, we here focus on designing a PDF which is continuous, but not necessarily continuously differentiable. As a consequence, there is not necessarily a solution to the equation in \eqref{eqn:estimating_proposed}.

We finish this section with a theoretical result about the asymptotic behaviour of the posterior distribution. We derive the result in a simplifying situation where the parameter $\nu$ is considered fixed, like $c$; the unknown parameter is thus considered to be only $\bbeta$ for the rest of the section. The prior and posterior are thus about this parameter only, and they will be denoted by $\pi$ and $\pi(\, \cdot \mid \by)$, respectively. We further simplify by considering that $\nu$ is such that $\nu > 1$ and $c < \sqrt{\nu}$ to ensure the existence of both tails in gamma GLM (and of $f_{\text{left}}$ in our model), which corresponds to the gamma PDF shape that often is sought for and supported by the data in actuarial science. The simplifying situation can be seen as an approximation of that where $\nu$ is considered as unknown and random (as previously), but with a posterior mass that concentrates strongly around a specific value. The result that is derived suggests that the posterior density (when both $\bbeta$ and $\nu$ are considered unknown) asymptotically behaves like one where the PDF terms of the outlying data points in the original density are each replaced by $f_{\nu,c}(y_i)$.

A conclusion of our theoretical result is a convergence of the posterior distribution towards $\pi(\, \cdot \mid \by_k)$, which has a density defined as follows:
\[
 \pi(\bbeta \mid \by_k) = \pi(\bbeta) \prod_{i = 1}^n \left[\frac{1}{\mu_i} f_{\nu, c}\left(\frac{y_i}{\mu_i}\right) \right]^{k_i} \Bigg/ m(\by_k), \quad \bbeta \in \R^p,
\]
where
\[
 m(\by_k) := \int_{\R^p} \pi(\bbeta) \prod_{i = 1}^n \left[\frac{1}{\mu_i} f_{\nu, c}\left(\frac{y_i}{\mu_i}\right) \right]^{k_i} \d\bbeta.
\]
In \autoref{thm:robustness} below, we use $\lceil \, \cdot \,\rceil$ to denote the ceiling function.

\begin{Theorem}\label{thm:robustness}
 Assume that $\nu$ is fixed and such that $\nu > 1$ and $c < \sqrt{\nu}$. Assume that $\pi$ is bounded. If $k \geq \lceil \lambdal / \lambdar \rceil(l + s) + 2p - 1$, that is $n \geq (\lceil\lambdal / \lambdar\rceil + 1)(l + s) + 2p - 1$, then both $\pi(\, \cdot \mid \by)$ and $\pi(\, \cdot \mid \by_k)$ are proper and as $\omega \rightarrow \infty$ (with $y_i = b_i\omega$ or $y_i = 1 / b_i \omega$ for outlying observations),
 \begin{enumerate}

  \item[(a)] the asymptotic behaviour of the marginal distribution is: $m(\by) / m(\by_k) \prod_{i = 1}^n \left[f_{\nu,c}(y_i)\right]^{s_i + l_i}\rightarrow 1$;

  \item[(b)] the posterior density converges pointwise: for any $\bbeta \in \R^p, \pi(\bbeta \mid \by) \rightarrow \pi(\bbeta \mid \by_k)$;

  \item[(c)] the posterior distribution converges: $\pi(\, \cdot \mid \by) \rightarrow \pi(\, \cdot \mid \by_k)$.
 \end{enumerate}
\end{Theorem}

In the simplifying situation considered, the conclusions of \autoref{thm:robustness} about our approach hold once the prior on $\bbeta$ is set to be bounded, as long as the number of non-outliers is large enough. A sufficient number is $\lceil\lambdal / \lambdar\rceil(l + s) + 2p - 1$ which is equivalent to having an upper bound on the number of outliers of $l + s \leq (n - 2p + 1) / (1 + \lceil\lambdal / \lambdar\rceil)$, where $\lambdal / \lambdar$ varies from about 1 to 4 when $c = 1.6$, as seen in \autoref{fig:lambdas}. This condition suggests that the breakdown point, generally defined as the proportion of outliers $(l + s) / n$ that an estimator can handle, is $(n - 2p + 1) / n(1 + \lceil\lambdal / \lambdar\rceil)$, which is close to $1/(1 + \lceil\lambdal / \lambdar\rceil)$ if $n$ is large relatively to $p$.

In \autoref{thm:robustness}, Result (a) represents the centrepiece; it leads relatively easily to the other results, but its demonstration requires considerable work. Result (a) together with \autoref{prop:limit_PDF} lead to Result (b), which in turn leads to Result (c) using Scheffé's theorem \citep{scheffe1947useful}. Also, Result (a) together with \autoref{prop:limit_PDF} suggest that the posterior density with both $\bbeta$ and $\nu$ unknown asymptotically behaves like one where the PDF terms of the outlying data points in the original density are each replaced by $f_{\nu,c}(y_i)$. Indeed, this posterior density essentially corresponds to $\pi(\bbeta \mid \by)$ if we multiply the numerator by $\pi(\nu)$ and integrate the denominator with respect to $\pi(\nu)$. We understood that the reason why we cannot prove a result in the situation where both $\bbeta$ and $\nu$ are considered unknown (at least using our proof technique) is that we cannot write $f_{\nu,c}(y_i)$ as a product of two terms with one depending on $\nu$ but not on $y_i$ and the other one depending on $y_i$ but not $\nu$.

The convergence of the posterior density in Result (b) indicates that the maximum a posteriori probability estimate is robust in the simplifying situation. Result (c) indicates that any estimation based on posterior quantiles (e.g., using posterior medians or Bayesian credible intervals) is robust. It is possible to obtain a result about the convergence of the posterior expectations under additional technical conditions. All these results characterize the limiting behaviour of a variety of Bayes estimators.

\section{Numerical experiments}\label{sec:numerical_experiments}

The results shown in \autoref{fig:estimates_yn} are interesting in that they allow to qualitatively evaluate the quality of the proposed model and the associated estimator. In \autoref{sec:simulation}, we provide a quantitative evaluation through a simulation study in which the parameters are estimated using maximum likelihood method. We next turn to Bayesian estimation in \autoref{sec:case_study} where we present a detailed case study.

\subsection{Simulation study}\label{sec:simulation}

 In this simulation study, we evaluate the estimation performance of: gamma GLM, the method of \cite{cantoni2001robust}, and the proposed approach. The gamma GLM is estimated using maximum likelihood method, as well as the proposed model. The goal of this simulation study is also to identify good values of $c$ for the proposed model; accordingly, several values are considered: $1.2, 1.3, \ldots, 2$. Values outside of this range yield non-effective approaches, at least based on our numerical experiments. The estimates based on the method of \cite{cantoni2001robust} are computed using the \textsf{robustbase R} package with the default options \citep{maechler2022package}. Note that, in order to make maximum likelihood estimation of the proposed model and the estimation method of \cite{cantoni2001robust} comparable, we do not include a weight function applied to each $\bx_i$ in the latter. As mentioned, a weight function can be used to decrease the weight of high-leverage points, and such a function can be included in our approach when viewed as an M-estimator. Here, we simplify by using vanilla versions.

Performance will be measured under several scenarios. In each scenario, base data sets are first simulated using gamma distributions based on the same mechanism as for \autoref{fig:estimates_yn}. A scenario where performance evaluation is based on such base data sets allows to measure the efficiency of an estimator in the absence of outliers when gamma GLM is the gold standard. In other scenarios, data points of the base data sets are modified to introduce outliers for robustness evaluation. In these scenarios, the location of a data point is shifted as follows: given a location shift of $\vartheta>0$, we modify $r_i(\bbeta, \nu) = \sqrt{\nu}(y_i - \mu_i) / \mu_i$ (computed using the true parameter values) by adding $\vartheta$, and obtain $\tilde{r}_i(\bbeta, \nu) = r_i(\bbeta, \nu) + \vartheta$ and the shifted data point $(\bx_i, \tilde{y}_i)$ with $\tilde{y}_i = \tilde{r}_i(\bbeta, \nu) \, \mu_i / \sqrt{\nu} + \mu_i$. In a subset of these scenarios, we also change $\bx_i$ (of the modified data points $(\bx_i, \tilde{y}_i)$) to make the data points high-leverage points; we do this by setting $\tilde{\bx}_i = (1, 1.5 \, \max_j x_{j2})^T$. The modified data points thus considered in this case are $(\tilde{\bx}_i, \tilde{y}_i)$.

The scenarios are now enumerated and described in more detail.
\begin{itemize}
 \item \textbf{Scenario 0}: simulation of base data sets without modification.
 \item \textbf{Scenario 1}: simulation of base data sets with modification of $5\%$ of data points chosen uniformly at random using $\vartheta = 7$.
 \item \textbf{Scenario 2}: simulation of base data sets with modification of $10\%$ of data points chosen uniformly at random using $\vartheta = 7$.
 \item \textbf{Scenario 3}: simulation of base data sets with modification of $5\%$ of data points chosen uniformly at random using $\vartheta = 3$ and with modification of $\bx_i$ as well.
 \item \textbf{Scenario 4}: simulation of base data sets with modification of $10\%$ of data points chosen uniformly at random using $\vartheta = 3$ and with modification of $\bx_i$ as well.
\end{itemize}

The choice of location shifts produces challenging and interesting situations where modified data points are often in a gray area where there is uncertainty regarding whether they really are outliers or not. The location shifts have been chosen analogously as in \cite{gagnon2020} who study a robust linear regression approach. Regarding the choice of scenarios, Scenarios 1 and 3 can be seen as scenarios with relatively few outliers, and Scenarios 2 and 4 allow to show how performance varies when the number of outliers is doubled. For each scenario, we consider two sample sizes: $n = 20$ and $n = 40$; this is to evaluate the impact of doubling the sample size. Note that similar results can be obtained with larger samples if the number of covariates (and thus of parameters) is increased accordingly.

The performance of each model/estimator is evaluated through the premium-versus-protection approach of \cite{anscombe1960rejection}. This approach consists in computing the premium to pay for using a robust alternative $\mathcal{R}$ to gamma GLM when there are no outliers (Scenario 0), and the protection provided by this alternative when the data sets are contaminated (other scenarios). The premium and protection associated with a robust alternative are evaluated through the following:
\begin{align*}
 &\text{Premium}(\mathcal{R}, \hat{\bbeta}) := \frac{\mathcal{M}_{\mathcal{R}}( \hat{\bbeta}) - \mathcal{M}_{\text{gamma}}( \hat{\bbeta})}{\mathcal{M}_{\text{gamma}}( \hat{\bbeta})}, \cr
 &\text{Protection}(\mathcal{R}, \hat{\bbeta} \mid \mathcal{S}) := \frac{\mathcal{M}_{\text{gamma}}( \hat{\bbeta}\mid \mathcal{S}) - \mathcal{M}_{\mathcal{R}}( \hat{\bbeta}\mid \mathcal{S})}{\mathcal{M}_{\text{gamma}}( \hat{\bbeta}\mid \mathcal{S})},
\end{align*}
where $\mathcal{S}$ is the scenario under which the protection is evaluated (1, 2, 3 or 4), and $\mathcal{M}_{\text{gamma}}( \hat{\bbeta}\mid \mathcal{S})$, for instance, denotes a measure $\mathcal{M}$ of the estimation error of the (true) regression coefficient using $\hat{\bbeta}$ with gamma GLM, in Scenario $\mathcal{S}$. The scenario is not specified for the premium because it does not vary; it is Scenario 0. The premium and protection for $\hat{\nu}$ have analogous definitions. We do not combine the estimation errors of all parameters, but instead measure the error of $\hat{\bbeta}$ and $\hat{\nu}$ separately to highlight a difference in estimation behaviour. For $\hat{\nu}$, $\mathcal{M}$ is the square root of the mean squared error; for $\hat{\bbeta}$, it is the square root of the expected (squared) Euclidean norm. The expectations are approximated through the simulation of 10,000 data sets. Note that premiums and protections are only evaluated for robust alternatives to gamma GLM as they are relative measures with respect to gamma GLM.

The results are graphically presented by plotting the couples $(\text{Premium}(\mathcal{R}, \hat{\bbeta}), \text{Protection}(\mathcal{R}, \hat{\bbeta} \mid \mathcal{S}))$ and $(\text{Premium}(\mathcal{R}, \hat{\nu}), \text{Protection}(\mathcal{R}, \hat{\nu} \mid \mathcal{S}))$. The results for Scenarios 1 and 2 are shown in \autoref{fig:scenarios_1_2}, and those for Scenarios 3 and 4 in \autoref{fig:scenarios_3_4}. From this \textit{premium-versus-protection} perspective, a robust alternative dominates another if its premium is smaller and protection larger. This means that in Figures \ref{fig:scenarios_1_2} and \ref{fig:scenarios_3_4}, we want to pay attention to the points in the upper left corner of the plots. The proposed models associated with the different values of $c$ studied are all excellent, as well as the method of \cite{cantoni2001robust}, as they offer better protections than their premiums in most cases. The proposed approach with $c = 1.6$ essentially dominates that of \cite{cantoni2001robust} in all cases, sometimes significantly, and offers an appealing premium-versus-protection profile. Based on our numerical experiments, we thus recommend using this value, when no prior belief about $\nu$ is available for choosing the value of $c$ and when $c$ is not estimated using statistical approaches (recall the discussion at the end of \autoref{sec:model}). Note that for a given percentage of outliers (and therefore of non-outliers), a larger sample size translates into enhanced protection for all approaches in all scenarios, which is a consequence of a larger number of non-outliers for estimation (for a fixed number of parameters).

 \begin{figure}[ht]
  \centering\small
  $\begin{array}{ccc}
  &  \hat{\bbeta} & \hat{\nu} \cr
 \vspace{-2mm}\hspace{-2mm}\rotatebox{90}{\hspace{5mm}\textbf{Scenario 1 -  $n = 20$}} &  \hspace{-0mm} \includegraphics[width=0.41\textwidth]{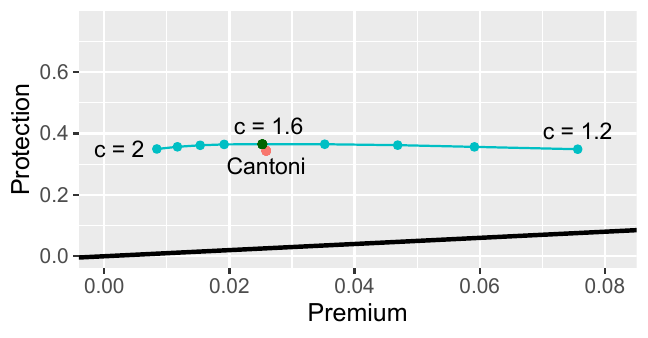} &  \hspace{-5mm} \includegraphics[width=0.41\textwidth]{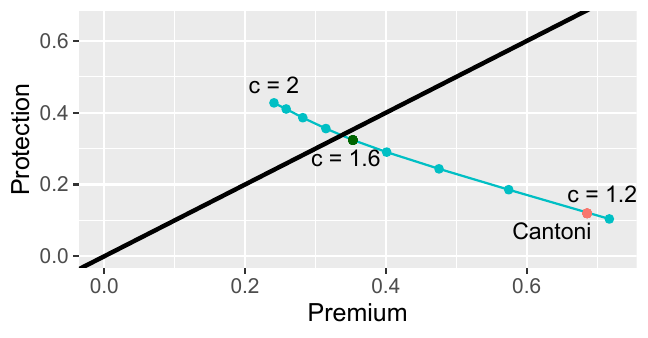} \cr
 \vspace{-2mm}\hspace{-2mm}\rotatebox{90}{\hspace{5mm}\textbf{Scenario 1 -  $n = 40$}} &  \hspace{-0mm} \includegraphics[width=0.41\textwidth]{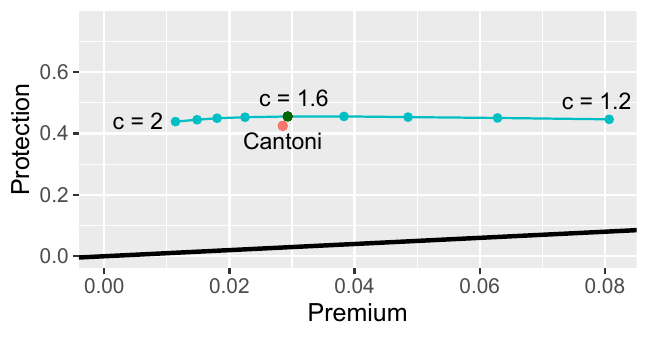} &  \hspace{-5mm} \includegraphics[width=0.41\textwidth]{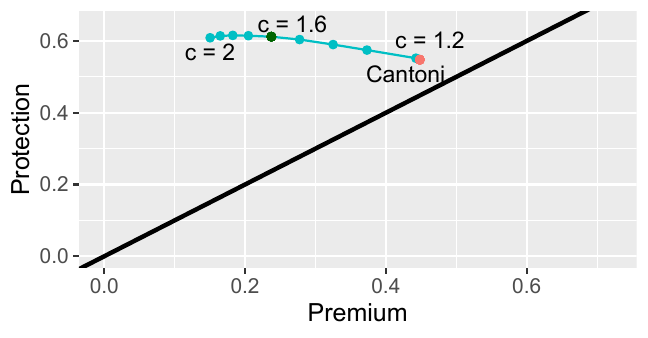} \cr
 \vspace{-2mm}\hspace{-2mm}\rotatebox{90}{\hspace{5mm}\textbf{Scenario 2 -  $n = 20$}} &  \hspace{-0mm} \includegraphics[width=0.41\textwidth]{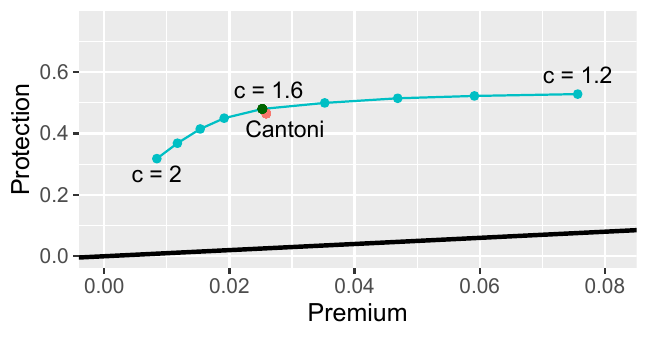} &  \hspace{-5mm} \includegraphics[width=0.41\textwidth]{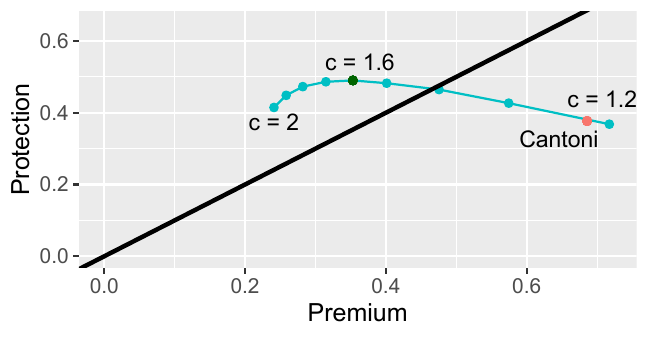} \cr
 \vspace{-2mm}\hspace{-2mm}\rotatebox{90}{\hspace{5mm}\textbf{Scenario 2 -  $n = 40$}} &  \hspace{-0mm} \includegraphics[width=0.41\textwidth]{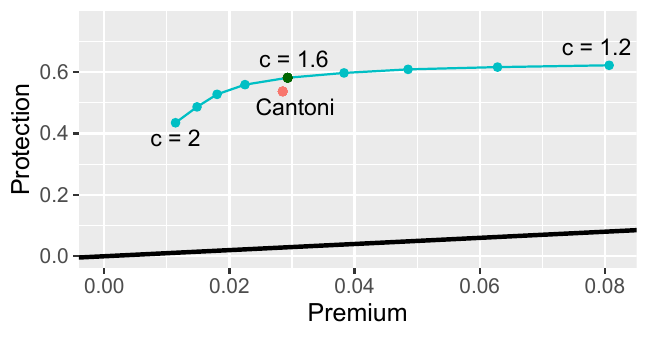} &  \hspace{-5mm} \includegraphics[width=0.41\textwidth]{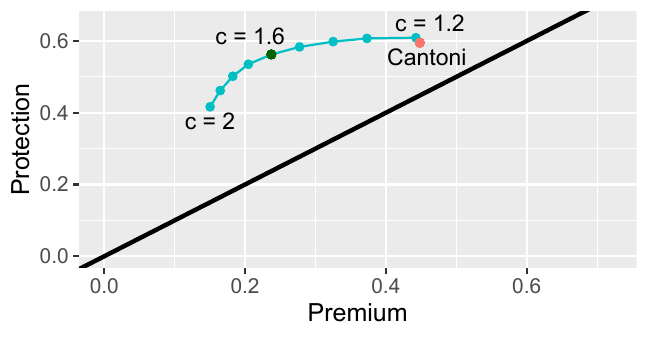}
  \end{array}$\vspace{-2mm}
  \caption{\small Premiums versus protections in Scenarios 1 and 2, with lines premium $=$ protection to identify the robust alternatives that offer better protections than their premiums.}\label{fig:scenarios_1_2}
 \end{figure}
\normalsize

 \begin{figure}[ht]
  \centering\small
  $\begin{array}{ccc}
  &  \hat{\bbeta} & \hat{\nu} \cr
 \vspace{-2mm}\hspace{-2mm}\rotatebox{90}{\hspace{5mm}\textbf{Scenario 3 -  $n = 20$}} &  \hspace{-0mm} \includegraphics[width=0.41\textwidth]{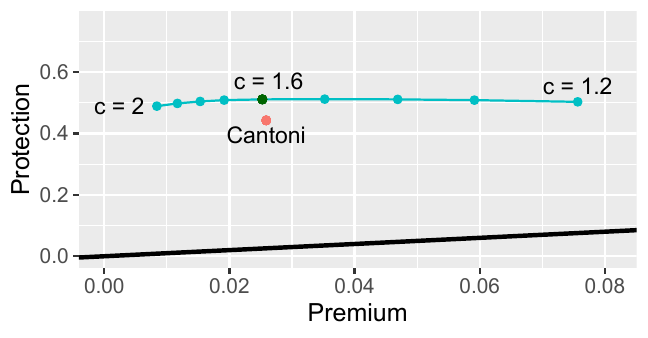} &  \hspace{-5mm} \includegraphics[width=0.41\textwidth]{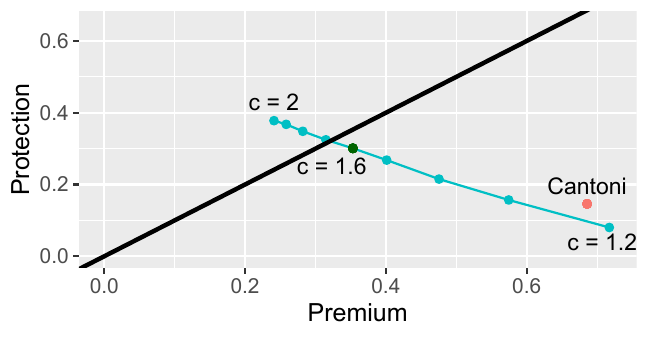} \cr
 \vspace{-2mm}\hspace{-2mm}\rotatebox{90}{\hspace{5mm}\textbf{Scenario 3 -  $n = 40$}} &  \hspace{-0mm} \includegraphics[width=0.41\textwidth]{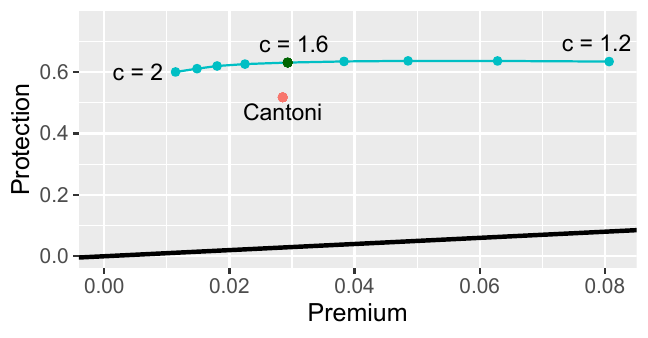} &  \hspace{-5mm} \includegraphics[width=0.41\textwidth]{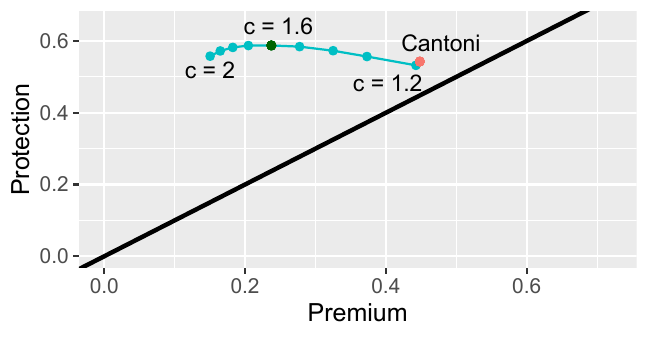} \cr
 \vspace{-2mm}\hspace{-2mm}\rotatebox{90}{\hspace{5mm}\textbf{Scenario 4 -  $n = 20$}} &  \hspace{-0mm} \includegraphics[width=0.41\textwidth]{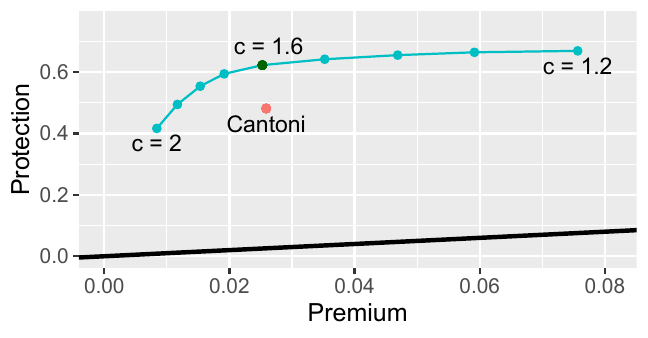} &  \hspace{-5mm} \includegraphics[width=0.41\textwidth]{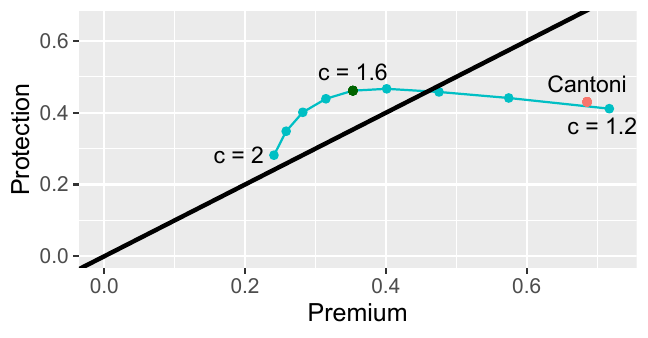} \cr
 \vspace{-2mm}\hspace{-2mm}\rotatebox{90}{\hspace{5mm}\textbf{Scenario 4 -  $n = 40$}} &  \hspace{-0mm} \includegraphics[width=0.41\textwidth]{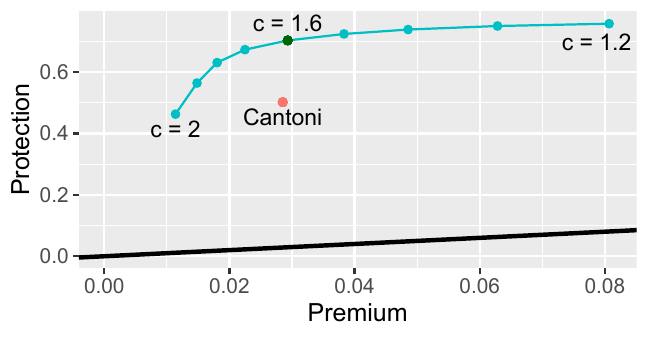} &  \hspace{-5mm} \includegraphics[width=0.41\textwidth]{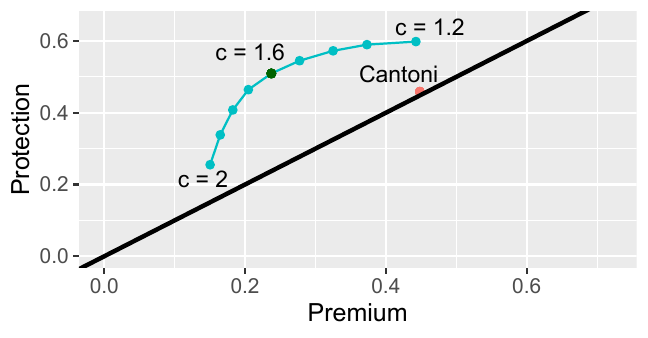}
  \end{array}$\vspace{-2mm}
  \caption{\small Premiums versus protections in Scenarios 3 and 4, with lines premium $=$ protection to identify the robust alternatives that offer better protections than their premiums.}\label{fig:scenarios_3_4}
 \end{figure}
\normalsize

\subsection{Health care expenditures: A Bayesian case study}\label{sec:case_study}

In this section, we provide a Bayesian statistical analysis of health-care expenditures. The data set is about 100 patients hospitalized at the \textit{Centre Hospitalier Universitaire Vaudois} in Lausanne, Switzerland, for medical back problems during 1999. The data set is available in the \textsf{robmixglm R} package \citep{beath2018robmixglm}. It is known for containing outliers, and has been analysed by \cite{marazzi2004adaptively}, \cite{cantoni2006robust} and \cite{beath2018mixture} to highlight the benefits of using robust statistical methods. The objective of a statistical analysis of that data set is to model the cost of stay in that hospital using six explanatory variables such as the length of stay and the admission
type (planned versus emergency). The empirical density of the dependent variable is highly right-skewed. This characteristic of the data, together with the fact the dependent variable is positive, motivate the use of gamma GLM. Analysing such a data set is of interest for actuaries. It can help them understand the main contributing factors, in this case, to the health cost, in order to provide a basis for accurate insurance pricing.

For the analysis, we assign a prior distribution on the parameters which is such that: $\pi(\bbeta \mid \nu) \propto 1$ and $\nu$ has a gamma distribution which is weakly informative and not in contradiction with the likelihood function. We obtain samples from the posterior distributions resulting from gamma GLM and the proposed model using Hamiltonian Monte Carlo \citep{Duane1987}. The posterior estimates are computed using Markov-chain samples of size 1,000,000, after discarding the first 10\% of the iterations. See \autoref{sec:supp_case_study} for the detailed expressions of the posterior densities and gradients. Note that the posterior density resulting from the proposed model has discontinuous derivatives, which surely has an impact on the performance of the numerical method. However, the discontinuity points have null measure and thus does not prevent the use of such a method.

Posterior medians and 95\% highest posterior density (HPD) credible intervals (CIs) are presented in \autoref{table:case_study}. They are computed for gamma GLM, based on the whole data set and without identified outliers, and for the proposed model. We observe significant differences between the estimates for gamma GLM based on the whole data set and the two other sets of estimates. The most significant difference in estimates is for $\nu$. Its point estimate for gamma GLM based on the whole data set is about half of that for the proposed model. Smaller estimated values for this parameter translate into larger CI lengths. We observe the impact of an inflated length in particular in the estimation of $\beta_7$ where the CI includes 0 in the case of gamma GLM (based on the whole data), while it does not in the case of the proposed model. The former CI suggests that the explanatory variable is not significantly related to the dependent variable, while the latter suggests otherwise. To formally conduct hypothesis testing or variable selection, one can implement a reversible jump algorithm as in \cite{gagnon2019RJ} or a lifted version  (which may be beneficial for variable selection) as in \cite{gagnon2019NRJ} or \cite{gagnon2020asymptotic} to sample from the joint posterior distribution of the models induced by the hypothesis testing or variable selection and their parameters.

While there are differences between the estimates for the proposed model and those for gamma GLM based on the data set without identified outliers, the conclusions suggested by the CIs are the same (if we consider that a CI with an endpoint equal to $0.00$, to two decimal places, includes the value $0$). Those differences should not come as a surprise as the estimation of the proposed model does not correspond to that of gamma GLM based on the data set without identified outliers, but rather to an estimation where erroneous and extreme data points are automatically assigned a weight which decreases with the uncertainty regarding whether they really are outliers.

\begin{table}
 \centering
\begin{tabular}{l rrrrrr}
\toprule
 &  \multicolumn{2}{c}{\textbf{Gamma GLM}} & \multicolumn{2}{c}{\textbf{Proposed model}} & \multicolumn{2}{c}{\textbf{Gamma GLM w.o.}} \cr
 \cmidrule(l){2-3} \cmidrule(l){4-5} \cmidrule(l){6-7}
 & Median & 95\% HPD CI & Median & 95\% HPD CI & Median & 95\% HPD CI  \cr
\midrule
 $\beta_1$ & $9.00$ & $(8.84, 9.16)$ & $9.03$ & $(8.92, 9.15)$ &  $9.05$ & $(8.94, 9.15)$  \cr
 $\beta_2$ & $0.68$ & $(0.64, 0.73)$ & $0.71$ & $(0.67, 0.74)$ & $0.71$ & $(0.67, 0.74)$ \cr
 $\beta_3$ & $-0.01$ & $(-0.06, 0.04)$ & $-0.03$ & $(-0.07, 0.01)$ & $-0.03$ & $(-0.07, 0.00)$  \cr
 $\beta_4$ & $0.21$ & $(0.11, 0.32)$ & $0.22$ & $(0.14, 0.30)$ & $0.21$ & $(0.14, 0.28)$  \cr
 $\beta_5$ & $0.09$ & $(-0.06, 0.26)$ & $0.00$ & $(-0.13, 0.12)$ & $-0.02$ & $(-0.14, 0.09)$ \cr
 $\beta_6$ & $0.09$ & $(-0.01, 0.19)$ & $0.08$ & $(0.00, 0.15)$ & $0.05$ & $(-0.02, 0.12)$  \cr
$\beta_7$ & $-0.10$ & $(-0.25, 0.04)$ & $-0.13$ & $(-0.23, -0.03)$ & $-0.13$ & $(-0.22, -0.03)$ \cr
 $\nu$ & $18.84$ & $(14.50, 23.50)$ & $36.32$ & $(28.20, 45.22)$ & $42.82$ & $(33.98, 52.34)$  \cr
 \bottomrule
\end{tabular}
 \caption{ \label{table:case_study} \small Posterior medians and 95\% HPD CIs under gamma GLM, based on the whole data set and without identified outliers, and under the proposed model.}
\end{table}

In regression analyses, residuals are often used to detect outliers. With GLMs, the Pearson residual is computed for each data point. Viewed as a function of the parameters $r_i(\bbeta, \nu)$ that is then estimated in a plug-in fashion in the original definition, it can be estimated in a Bayesian way by the posterior median, for instance, using the Markov-chain samples. The Bayesian estimates based on posterior medians can be found against the posterior medians of $\mu_i$ (the Bayesian analogue of the fitted values) in \autoref{fig:res_vs_fitted}. We observe that the residuals based on the estimation of the proposed model are overall more dispersed than those based on the estimation of gamma GLM, mainly due to a smaller estimated value of $\nu$ in the latter case. Data points are flagged as outliers and investigated if their residuals are extreme. With the residuals based on the estimation of gamma GLM, it is less evident which data points should be flagged as outliers, a consequence of the masking effect. Outlier detection based on the estimation of the proposed model is more effective.

  \begin{figure}[ht]
  \centering\small
  $\begin{array}{cc}
 \includegraphics[width=0.45\textwidth]{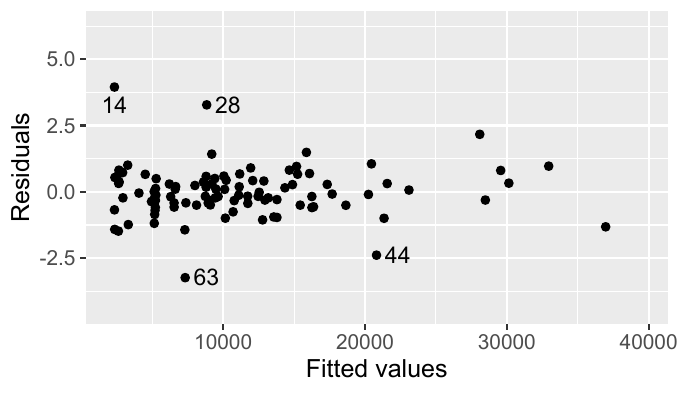} &  \includegraphics[width=0.45\textwidth]{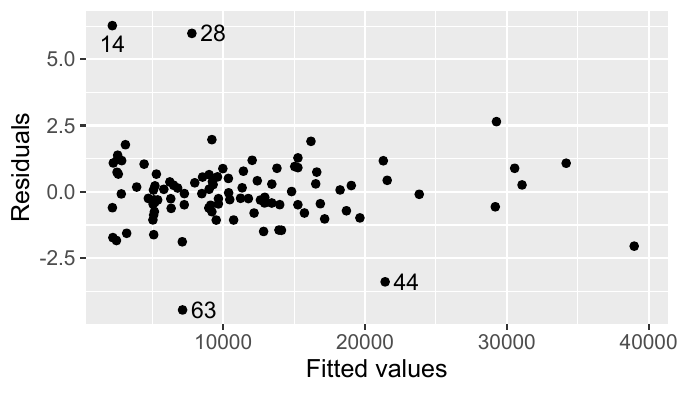}  \cr
   \hspace{-0mm} \textbf{(a) gamma GLM} & \hspace{-0mm} \textbf{(b) proposed model}
  \end{array}$\vspace{-2mm}
  \caption{\small Bayesian Pearson residuals against Bayesian fitted values under gamma GLM and the proposed model.}\label{fig:res_vs_fitted}
 \end{figure}
\normalsize

\section{Discussion}\label{sec:discussion}

In this paper, we highlighted that (non-)robustness against outliers is an aspect to bear in mind when conducting statistical analyses using GLMs. We also highlighted that there are few robust alternatives, especially when one wants to conduct Bayesian statistical analyses. While focusing on gamma GLM, which is an ubiquitous tool in actuarial science, we proposed an effective robust approach that is modelling-based and can thus be used for both frequentist and Bayesian analyses. The proposed model is easy to interpret and understand, and can be straightforwardly estimated (at least on small-to-moderate-size data sets). The resulting MLE (which can be viewed as an M-estimator for gamma GLM) is similar in flavour to the most commonly used robust frequentist estimator of \cite{cantoni2001robust}. The theoretical and empirical comparison made in the paper shows that the proposed approach is an appealing alternative.

 In future work, it would be interesting to study the computational aspect in depth, that is the scalability with the sample size and with the number of covariates. Also, there may be ways to improve the proposed approach. A weakness that we identified (which the approach of \cite{cantoni2001robust} also has) is that the area where the original gamma PDF is replaced by a robustified function (in the case of the method of \cite{cantoni2001robust}, it is instead the area where the derivative of the log-likelihood is replaced) is based on the Pearson residual, in the same spirit as in robust linear regression. But, contrarily to the standardized residual in linear regression, the Pearson residual is not symmetrically distributed. We believe it would be interesting to explore if using a Wilson--Hilferty transformation of the Pearson residual instead would yield a significant improvement (see, e.g., \cite{terrell2003wilson} for a discussion about Wilson--Hilferty transformations). We expect the resulting approach not to perform significantly better than that proposed here when the shape parameter is moderate to large like in our numerical experiments in \autoref{sec:numerical_experiments} as the asymmetry of the gamma is not too severe in this case. Finally, it would be interesting to work on model adaptation of GLMs whose response-variable distributions do not have tails, like logistic regression. For such a model, it is not clear to us that there even exists a precise definition of outliers. There are thus interesting fundamental questions to be answered regarding robustness when one wants to use such a GLM.

 \section*{Acknowledgements}

The authors thank two anonymous referees and an associate editor for constructive comments that led to an improved manuscript. Philippe Gagnon acknowledges support from NSERC (Natural Sciences and Engineering Research Council of Canada) and FRQNT (Fonds de recherche du Québec -- Nature et technologies).

\bibliographystyle{rss}
\bibliography{references}

\appendix

\section{Proofs}\label{sec:proofs}

We now present the proofs of all theoretical results in the same order as the results appeared in the paper.

\begin{proof}[Proof of \autoref{prop:asymptotic_lambdas}]
  We prove the result for $\lambdar$. The result for $\lambdal$ is proved analogously. We have that
 \[
 \lambdar= 1+\frac{f_{\text{mid}}(\zr)\log(\zr) \, \zr}{\P[Z_{\nu} > \zr]}.
 \]
 To prove the result, we will prove that the numerator of the fraction converges towards
 \[
  \frac{c}{\sqrt{2 \pi} \, \e^{c^2 / 2}},
 \]
 and that the denominator converges towards $1 - \Phi(c)$. We start with the analysis of the numerator:
 \begin{align*}
   f_{\text{mid}}(\zr)\log(\zr) \, \zr  &= \frac{\exp\{-\nu \zr\}(\zr \nu)^\nu \log(\zr)}{\Gamma(\nu)} \cr
   &= \frac{\exp\left\{-\nu (1+c/\sqrt{\nu})\right\}((1+c/\sqrt{\nu})\nu)^\nu \log(1+c/\sqrt{\nu})}{\sqrt{2\pi}(\nu/\e)^\nu / \sqrt{\nu}}\frac{\sqrt{2\pi}(\nu/\e)^\nu / \sqrt{\nu}}{\Gamma(\nu)}.
 \end{align*}
 By Stirling's formula, we know that
 \[
  \frac{\sqrt{2\pi}(\nu/\e)^\nu / \sqrt{\nu}}{\Gamma(\nu)} \rightarrow 1.
 \]
 We can thus focus on the other term. After simplification, it is equal to
 \begin{align*}
\frac{\exp\left\{-\nu (1+c/\sqrt{\nu})\right\}((1+c/\sqrt{\nu})\e)^\nu \log[(1+c/\sqrt{\nu})^{\sqrt{\nu}}]}{\sqrt{2\pi}}.
 \end{align*}
 We have that
 \[
  \log[(1+c/\sqrt{\nu})^{\sqrt{\nu}}] \rightarrow c,
 \]
 and
 \begin{align*}
  \exp\left\{-\nu (1+c/\sqrt{\nu})\right\}((1+c/\sqrt{\nu})\e)^\nu &= \exp\left\{-\nu (1+c/\sqrt{\nu})\right\}\exp\left\{\nu(1 + \log(1+c/\sqrt{\nu}))\right\} \cr
  &= \exp\left\{-c\sqrt{\nu}\right\}\exp\left\{\sqrt{\nu}\log[(1+c/\sqrt{\nu})^{\sqrt{\nu}}]\right\} \rightarrow \frac{1}{\e^{c^2 / 2}}.
 \end{align*}
 This concludes the proof that the numerator in $\lambda_{\text{r}}$ converges towards
 \[
  \frac{c}{\sqrt{2 \pi} \, \e^{c^2 / 2}}.
 \]

 We now turn to the proof that the denominator, $\P[Z_{\nu} > \zr]$, converges towards $1 - \Phi(c)$. We have that $Z_{\nu}$ follows a gamma distribution whose mean and shape parameters are given by 1 and $\nu$, respectively. This implies that the scale parameter is given by $1 / \nu$. Therefore,
 \begin{align*}
  \P[Z_{\nu} > \zr] = \P[X_{\nu} > \nu \zr] &= \P[X_{\nu} > \nu + c \sqrt{\nu}] \cr
  &= \P\left[\frac{X_{\nu} - \nu}{\sqrt{\nu}} > c \right],
 \end{align*}
 where $X_{\nu}$ follows a gamma distribution whose shape and scale parameters are $\nu$ and 1, respectively.

 We have that
 \begin{align*}
  \frac{X_{\nu} - \nu}{\sqrt{\nu}} = \sqrt{\frac{\lfloor \nu \rfloor}{\nu}} \sqrt{\lfloor \nu \rfloor} \left(\frac{1}{\lfloor \nu \rfloor}\sum_{k = 1}^{\lfloor \nu \rfloor} X_k - 1\right) + \frac{\tilde{X} - (\nu - \lfloor \nu \rfloor)}{\sqrt{\nu }},
 \end{align*}
 where $X_1, \ldots, X_{\lfloor \nu \rfloor}$ are independent random variables, each having an exponential distribution with a scale parameter of 1, and $\tilde{X}$ follows a gamma distribution whose shape and scale parameters are $\nu - \lfloor \nu \rfloor$ and 1, respectively. By the central limit theorem,
 \[
  \sqrt{\lfloor \nu \rfloor} \left(\frac{1}{\lfloor \nu \rfloor}\sum_{k = 1}^{\lfloor \nu \rfloor} X_k - 1\right)
 \]
 converges in distribution towards a standard normal distribution. Also, $\tilde{X} / \sqrt{\nu }$ converges towards 0 with probability 1. Therefore, by Slutsky’s theorem, we have that
 \[
 \frac{X_{\nu} - \nu}{\sqrt{\nu}}
 \]
 converges in distribution towards a standard normal distribution, which concludes the proof.
\end{proof}

We now present and prove two lemmas that will be used in the proof of \autoref{prop:proper}.

\begin{Lemma}\label{lemma1}
Viewed as a function of $\mu$, $f_{\nu,c}(y/\mu)/\mu$ is strictly increasing on $(0, y)$ and then strictly decreasing on $(y, \infty)$, for all $\nu$, $c$ and $y$. It is thus unimodal with a mode at $\mu = y$, and in particular, it is bounded above by $(\e^{-1}\nu)^\nu/(y\Gamma(\nu))$.
\end{Lemma}

\begin{proof}
Based on \eqref{eq:proposed},
\begin{align*}
        f_{\nu,c}(y/\mu)/\mu=\left\{
    \begin{array}{ll}
        f_{\text{mid}}(y/\mu)/\mu = \exp\left\{-\nu y/\mu\right\}\mu^{-\nu}y^{\nu-1}\nu^\nu/\Gamma(\nu)& \mbox{if } \zl\leq y/\mu\leq \zr, \\
        f_{\text{right}}(y/\mu)/\mu= f_{\text{mid}}(\zr)\frac{\zr}{y}\left(\frac{\log(\zr)}{\log(y/\mu)}\right)^{\lambdar} & \mbox{if } y/\mu>\zr,\\
        f_{\text{left}}(y/\mu)/\mu = f_{\text{mid}}(\zl) \frac{\zl}{y}\left(\frac{\log(\zl)}{\log(y/\mu)}\right)^{\lambdal} & \mbox{if } 0<y/\mu<\zl,\\
    \end{array}
    \right.
\end{align*}

We analyse the three parts of the function (of $\mu$) separately. We consider that all three parts exist; otherwise, one part (with $f_{\text{left}}$) has to be skipped.

We first consider that $\mu \in (0,y/\zr)$. In this case, $f_{\nu,c}(y / \mu)/\mu = f_{\text{right}}(y / \mu) / \mu$. We have that
\begin{align*}
  f_{\text{right}}(y/\mu)/\mu = f_{\text{mid}}(\zr) \frac{\zr}{y}\left(\frac{\log(\zr)}{\log(y/\mu)}\right)^{\lambdar}\propto \left(\frac{1}{\log(y/\mu)}\right)^{\lambdar}.
\end{align*}
This function (of $\mu$) is strictly increasing because  $\lambda_{\text{r}}>0$. Thus, for $\mu \in (0, y/z_{\text{r}})$,
\begin{align*}
   f_{\text{right}}(y / \mu) / \mu \leq f_{\text{mid}}(z_{\text{r}}) z_{\text{r}} / y.
\end{align*}

Analogously, we consider that $\mu \in (y / \zl, \infty)$. In this case, $f_{\nu,c}(y / \mu)/\mu = f_{\text{left}}(y / \mu) / \mu$. We have that
\begin{align*}
  f_{\text{left}}(y/\mu)/\mu = f_{\text{mid}}(\zl)\frac{\zl}{y}\left(\frac{\log(\zl)}{\log(y/\mu)}\right)^{\lambdal}\propto \left(\frac{1}{\log(\mu/y)}\right)^{\lambdal}.
\end{align*}
This function (of $\mu$) is strictly decreasing because $\lambdal>0$. Thus, for $\mu \in (y / \zl, \infty)$,
\begin{align*}
   f_{\text{left}}(y / \mu) / \mu \leq f_{\text{mid}}(\zl) \zl / y.
\end{align*}

Finally, consider that $\mu \in [y/\zr, y/\zl]$. In this case, $f_{\nu,c}(y/\mu)/\mu=f_{\text{mid}}(y/\mu)/\mu$. We have
$$f_{\text{mid}}(y / \mu)/\mu= f_{\text{mid}}(y / \mu) (y / \mu) / y. $$
If we evaluate the function at the boundaries, we obtain
\[
 f_{\text{mid}}(\zr) \zr / y,
\]
and
\[
f_{\text{mid}}(\zl) \zl / y,
\]
highlighting that the function is continuous.

Now, we show that the function is strictly increasing on $[y/\zr, y)$ and then strictly decreasing on $(y, y/\zl]$, which implies that it is unimodal with a mode at $\mu = y$. The derivative of the log of $f_{\text{mid}}(y / \mu) (y / \mu) / y$ with respect to $\mu$ is given by
\begin{align*}
  \frac{\partial}{\partial \mu}\log\left[f_{\text{mid}}(y / \mu) (y / \mu) / y\right] &= \frac{\partial}{\partial \mu} \left[-\nu (y / \mu) + \nu\log(y / \mu)+\nu\log\nu-\log(\Gamma(\nu))-\log(y)\right]\\
  &= \frac{\nu}{\mu^2}(y - \mu).
\end{align*}
The root of this function is $\mu=y$. If $\mu <y$, the derivative is strictly positive, meaning that the function $f_{\text{mid}}(y / \mu) (y / \mu) / y$ is strictly increasing on that part of the domain. If $\mu >y$, the derivative is strictly negative, meaning that the function $f_{\text{mid}}(y / \mu) (y / \mu) / y$ is strictly decreasing on that part of the domain. This allows to conclude that the function (of $\mu$) $f_{\nu,c}(y/\mu)/\mu$ is strictly increasing on $[y/z_{\text{r}}, y)$ and then strictly decreasing on $(y, y/z_{\text{l}}]$. Thus, $f_{\text{mid}}(1)/y$ is the maximum of $f_{\nu,c}(y/\mu)/\mu$. Therefore, the upper bound of $f_{\nu,c}(y/\mu)/\mu$ is given by $f_{\nu,c}(1)/y = (\e^{-1}\nu)^\nu/(y\Gamma(\nu))$.
\end{proof}

\begin{Lemma}\label{lemma2}
If $\pi(\, \cdot \,)$ is a proper PDF such that $\int_0^\infty \nu^{(n - p)/ 2} \, \pi(\nu) \, \d\nu < \infty$, then $\int_0^\infty \left[\frac{(\e^{-1}\nu)^\nu}{\Gamma(\nu)}\right]^{n-p}\, \pi(\nu)\,\d\nu<\infty$.
\end{Lemma}

\begin{proof}
We will separate the integral into two parts: from 0 to a large positive constant $\nu^*$, and from $\nu^*$ to $\infty$. The function ${(\e^{-1}\nu)^\nu}/{\Gamma(\nu)}$ is strictly increasing, thus this function is bounded on $0 < \nu \leq \nu^*$ by $(\e^{-1}\nu^*)^{\nu^*}/\Gamma(\nu^*)$. As $\nu$ gets large, $\Gamma(\nu)$ can be approximated by \textit{Stirling's formula}, given by
$$\Gamma(\nu)\approx S(\nu)=\frac{\sqrt{2\pi}(\nu/\e)^\nu}{\sqrt{\nu}}.$$
We thus have ${(\e^{-1}\nu)^\nu}/{\Gamma(\nu)}\approx{(\e^{-1}\nu)^\nu}/{S(\nu)}$. More precisely,
\begin{align*}
    \frac{S(\nu)}{\Gamma(\nu)}\to 1 &\Leftrightarrow \frac{(\e^{-1}\nu)^\nu}{(\e^{-1}\nu)^\nu}\frac{S(\nu)}{\Gamma(\nu)}\to 1 \Leftrightarrow \frac{\sqrt{2\pi}}{\sqrt{\nu}}\frac{(\e^{-1}\nu)^\nu}{\Gamma(\nu)}\to 1.
\end{align*}
Therefore, for all $\delta>0$, we can find a $\nu^*$ such that for all $\nu\geq\nu^*$,
$$\frac{(\e^{-1}\nu)^\nu}{\Gamma(\nu)} = \frac{(\e^{-1}\nu)^\nu}{\Gamma(\nu)}\frac{\sqrt{2\pi}}{\sqrt{\nu}}\frac{\sqrt{\nu}}{\sqrt{2\pi}}\leq(1+\delta)\frac{\sqrt{\nu}}{\sqrt{2\pi}}.$$
Thus,
\begin{align*}
    \int_0^\infty\pi(\nu)\left[\frac{(\e^{-1}\nu)^\nu}{\Gamma(\nu)}\right]^{n-p}\,\d\nu=&\int_0^{\nu^*}\pi(\nu) \left[\frac{(\e^{-1}\nu)^\nu}{\Gamma(\nu)}\right]^{n-p}\,\d\nu+\int_{\nu^*}^\infty\pi(\nu) \left[\frac{(\e^{-1}\nu)^\nu}{\Gamma(\nu)}\right]^{n-p}\,\d\nu\\
    \leq&\left[\frac{(\e^{-1}\nu^*)^{\nu^*}}{\Gamma(\nu^*)}\right]^{n-p}\int_0^{\nu^*}\pi(\nu)\,\d\nu + \int_{\nu^*}^\infty\pi(\nu)\left[(1 + \delta)\frac{\sqrt{\nu}}{\sqrt{2\pi}}\right]^{n-p}\,\d\nu\\
    =&\left[\frac{(\e^{-1}\nu^*)^{\nu^*}}{\Gamma(\nu^*)}\right]^{n-p}\int_0^{\nu^*}\pi(\nu)\,\d\nu + \frac{\int_{\nu^*}^\infty\pi(\nu)(1+\delta)\nu^{(n-p)/2}\,\d\nu}{(2\pi)^{(n-p)/2}}\\
    \leq&\left[\frac{(\e^{-1}\nu^*)^{\nu^*}}{\Gamma(\nu^*)}\right]^{n-p}\int_0^{\nu^*}\pi(\nu)\,\d\nu + \frac{\int_{0}^\infty\pi(\nu)(1+\delta)\nu^{(n-p)/2}\,\d\nu}{(2\pi)^{(n-p)/2}}\\
    <&\infty.
\end{align*}
The last step is due to the condition $\int_{0}^\infty \pi(\nu) \, \nu^{(n-p)/2} \d\nu < \infty$, and because $\pi(\, \cdot \,)$ is a proper PDF.
\end{proof}

\begin{proof}[Proof of \autoref{prop:proper}]
To prove this proposition, it suffices to show that the marginal $m(\by)$ is finite, that is
\begin{align*}
\iint\pi(\boldsymbol{\beta},\nu) \prod_{i=1}^n \frac{1}{\mu_i}f_{\nu,c}\left(\frac{y_i}{\mu_i}\right)\,\d\bbeta\,\d\nu<\infty.
\end{align*}

To prove this, we first split the data points into two parts. The first part contains $p$ data points, which will be used to perform a change of variables from $\boldsymbol{\beta}$ to $z_i = y_i/(\bx_i^T\boldsymbol{\beta})$ for $i=1,\ldots,p$. Without loss of generality, we choose the first $p$ data points. For the rest of the $n-p$ data points, we bound $\prod_{i=1}^{n-p} f_{\nu,c}(y_i/\mu_i)/\mu_i$ by a function depending on $\nu$ and then we show that this bound, multiplied by $\pi(\nu)$, is integrable with respect to $\nu$. We thus use the condition that $n \geq p$. When $n = p$, the proof is seen to be more simple, because the part with the rest of the $n - p$ data points does not actually exist. We have
\begin{align*}
    m(\textbf{y})&=\int_0^\infty\int_{\R^p} \pi(\bbeta, \nu) \prod_{i=1}^n \frac{f_{\nu,c}(y_i/\exp(\bx_i^T \bbeta))}{\exp(\bx_i^T \bbeta)} \, \d\bbeta \,\d\nu\\
    &=\int_0^\infty\int_{\R^p}\pi(\bbeta \mid \nu)\, \pi(\nu) \prod_{i=1}^{p} \frac{f_{\nu,c}(y_i/\exp(\bx_i^T \bbeta))}{\exp(\bx_i^T \bbeta)} \prod_{i=p+1}^n \frac{f_{\nu,c}(y_i/\exp(\bx_i^T \bbeta))}{\exp(\bx_i^T \bbeta)} \, \d\bbeta \, \d\nu\\
    &\overset{a}{\leq}B \int_0^\infty \left(\int_{\R^p} \prod_{i=1}^p \frac{f_{\nu,c}(y_i/\exp(\bx_i^T\bbeta))}{\exp(\bx_i^T \bbeta)} \, \d\bbeta\right) \pi(\nu) \prod_{i=p+1}^n\frac{(\e^{-1}\nu)^\nu}{y_i\Gamma(\nu)}\,\d\nu\\
    &\overset{b}{=}B \int_0^\infty\left(\Bigg| \text{det}\left(
    \begin{matrix}
           \bx_{1}^T \\
           \vdots \\
           \bx_{p}^T
         \end{matrix}\right)\Bigg|^{-1}\prod_{i=1}^p\int_{0}^\infty\frac{f_{\nu,c}(z_i)}{y_i} \d z_i\right) \pi(\nu) \prod_{i=p+1}^n\frac{(\e^{-1}\nu)^\nu}{y_i\Gamma(\nu)}\,\d\nu\\
    &\overset{c}{=}B\,\Bigg|\text{det}\left(
    \begin{matrix}
           \bx_{1}^T \\
           \vdots \\
           \bx_{p}^T
         \end{matrix}\right)\Bigg|^{-1}\prod_{i=1}^n\frac{1}{y_i}\int_0^\infty\displaystyle\pi(\nu)\left[\frac{(e^{-1}\nu)^\nu}{\Gamma(\nu)}\right]^{n-p}\,\d\nu\\
         &\overset{d}{<}\infty.
\quad
\end{align*}
In step $a$, we split the data points into two parts as we explained previously, and we use that $\pi(\boldsymbol{\beta}\mid\nu) \leq B $ with $B$ a positive constant. We also bound the product of $f(y_i/\mu_i)/\mu_i$ for $i=p+1,\ldots,n$ by using \autoref{lemma1}. In step $b$, we perform a change of variables from $\boldsymbol{\beta}$ to $z_i=y_i/\exp(\bx_i^T \bbeta)$, for $i=1, \dots, p$. For each $i$, we have $\partial z_i/ \partial\bbeta = y_i\bx_i^T/\exp(\bx_i^T\bbeta)$. The determinant is non-null because all explanatory variables are continuous. Indeed, consider the case $p = 2$ for instance; the determinant is
different from 0 provided that $x_{12}\neq x_{22}$, which happens with probability 1. When any type of explanatory variables is considered, we need to be able to select $p$ observations, say those with $\bx_{i_1},\ldots, \bx_{i_p}$, such that the matrix with rows $\bx_{i_1}^T,\ldots, \bx_{i_p}^T$ has a non-null
determinant. In step $c$, we used that $f_{\nu,c}$ is a PDF. In step $d$, we used \autoref{lemma2}.
\end{proof}

\begin{proof}[Proof of \autoref{prop:limit_PDF}]
We first consider the case where $i$ is such that $l_i = 1$, and in order to prove the result we show that
\[
 \lim_{y\rightarrow\infty} \frac{f_{\nu,c}(y/\mu)/\mu}{f_{\nu,c}(y)} = 1.
\]
In the denominator $f_{\nu,c}(y)$, $f_{\text{right}}$ is activated, and we have
\[
 f_{\text{right}}(y)= f_{\text{mid}}(\zr) \frac{\zr}{y} \left(\frac{\log(\zr)}{\log(y)}\right)^{\lambdar},
\]
  where $\zr$, $f_{\text{mid}}(\zr)$, and $\lambdar$ depend only on $\nu$ and $c$. As $y/\mu \rightarrow \infty$, $f_{\text{right}}$ is also activated in $f_{\nu,c}(y/\mu)/\mu$. Thus,
\begin{align*}
    \frac{f_{\text{right}}(y/\mu)/\mu}{f_{\text{right}}(y)}&=f_{\text{mid}}(\zr) \frac{\zr}{y}\left.\left(\frac{\log(\zr)}{\log(y/\mu)}\right)^{\lambdar} \right/ f_{\text{mid}}(\zr) \frac{\zr}{y}\left(\frac{\log(\zr)}{\log(y)}\right)^{\lambdar}\\
    &=\left(\frac{\log(y)}{\log(y)-\log(\mu)}\right)^{\lambdar}\rightarrow 1.
\end{align*}

We consider now that $y\rightarrow0$, under the condition that $c<\sqrt{\nu}$ and $\nu>1$. Under this condition, $f_{\text{left}}$ exists and is activated in both $f_{\nu,c}(y)$ and $f_{\nu,c}(y/\mu)/\mu$. We have
\begin{align*}
    \frac{f_{\text{left}}(y/\mu)/\mu}{f_{\text{left}}(y)}&=f_{\text{mid}}(\zl)\frac{\zl}{y}\left.\left(\frac{\log(\zl)} {\log(y/\mu)}\right)^{\lambdal} \right/ f_{\text{mid}}(\zl)\displaystyle\frac{\zl}{y}\left(\frac{\log(\zl)}{\log(y)}\right)^{\lambdal}\\
    &=\left(\frac{\log(y)}{\log(y)-\log(\mu)}\right)^{\lambda_{\text{l}}}\to 1.
\end{align*}
\end{proof}

\begin{proof}[Proof of \autoref{prop:estimating}]
 Let us consider that $\nu>1$ and $c<\sqrt{\nu}$; otherwise, the third case in the partial derivative below is omitted. We have that
\[
 \frac{\partial}{\partial \bbeta} \log f_{\nu, c}\left(\frac{y_i}{\mu_i}\right) - \log \mu_i = \begin{cases}
    \nu \left(\frac{y_i}{\mu_i} - 1\right) \, \bx_i \quad \text{if} \quad \zl \leq \frac{y_i}{\mu_i} \leq \zr, \cr
    \frac{\lambdar}{\log(y_i / \mu_i)} \, \bx_i \quad \text{if} \quad \frac{y_i}{\mu_i} > \zr, \cr
    \frac{\lambdal}{\log(y_i / \mu_i)} \, \bx_i \quad \text{if} \quad \frac{y_i}{\mu_i} < \zl.
 \end{cases}
\]
This can be rewritten as
\[
 \frac{\partial}{\partial \bbeta} \log f_{\nu, c}\left(\frac{y_i}{\mu_i}\right) - \log \mu_i = \begin{cases}
    \sqrt{\nu} \sqrt{\nu} \left(\frac{y_i}{\mu_i} - 1\right) \, \bx_i \quad \text{if} \quad -c \leq \sqrt{\nu}\left(\frac{y_i}{\mu_i} -1\right) \leq c, \cr
    \sqrt{\nu} \, \frac{\lambdar / \sqrt{\nu}}{\log(\sqrt{\nu}(y_i / \mu_i - 1) / \sqrt{\nu} + 1)} \, \bx_i \quad \text{if} \quad \sqrt{\nu}\left(\frac{y_i}{\mu_i} -1\right) > c, \cr
    \sqrt{\nu}\, \frac{\lambdal / \sqrt{\nu}}{\log(\sqrt{\nu}(y_i / \mu_i - 1) / \sqrt{\nu} + 1)} \, \bx_i \quad \text{if} \quad \sqrt{\nu}\left(\frac{y_i}{\mu_i} -1\right) < -c.
 \end{cases}
\]
The proof is concluded by recalling that $r_i(\bbeta, \nu) = \sqrt{\nu}(y_i / \mu_i - 1)$.
\end{proof}

\begin{proof}[Proof of \autoref{thm:robustness}]
 We start with the proof of Result (a), which is quite lengthy. We next turn to the proofs of Results (b) and (c) which are shorter.

 Let us assume for now that $m(\by) < \infty$ for all $\omega$, and $m(\by_k)< \infty$. This is proved below. We first observe that
\begin{align*}
\frac{m(\by)}{m(\by_k)\prod_{i=1}^n[f_{\nu,c}(y_i)]^{s_i + l_i}} & = \frac{m(\by)}{m(\by_k)\prod_{i=1}^n[f_{\nu,c}(y_i)]^{s_i+l_i}}\int_{\R^p}\pi(\bbeta \mid \by) \,\d\bbeta\\
& =\int_{\R^p}\frac{\pi(\bbeta)\prod_{i=1}^n[f_{\nu,c}(y_i/\mu_i)/\mu_i]}{m(\by_k)\prod_{i=1}^n[f_{\nu,c}(y_i)]^{s_i+l_i}} \,\d\bbeta\\
&= \int_{\R^p}\pi(\bbeta \mid \by_k)\prod_{i=1}^n\left[\frac{ f_{\nu,c}(y_i/\mu_i)/\mu_i}{f_{\nu,c}(y_i)}\right]^{s_i+l_i} \d\bbeta.
\end{align*}

We show that the last integral converges to 1 as $\omega\rightarrow\infty$. Assuming that we can interchange the limit and the integral, we obtain that
\begin{align*}
\lim_{\omega\rightarrow\infty}\int_{\R^p}\pi(\bbeta\mid\by_k)\prod_{i=1}^n\left[\frac{ f_{\nu,c}(y_i/\mu_i)/\mu_i}{f_{\nu,c}(y_i)}\right]^{s_i+l_i} \d\bbeta &= \int_{\R^p}\lim_{\omega\rightarrow\infty}\pi(\bbeta\mid\by_k)\prod_{i=1}^n\left[\frac{f_{\nu,c}(y_i/\mu_i)/\mu_i} {f_{\nu,c}(y_i)}\right]^{s_i+l_i}\d\bbeta\\
&\overset{a}{=}\int_{\R^p}\pi(\bbeta\mid\by_k)\times1\ \d\bbeta\overset{b}{=}1.
\end{align*}
In step $a$, we use \autoref{prop:limit_PDF}. In step $b$, we use that $\pi(\bbeta \mid \by_k)$ is proper. Indeed, we notice in the proof of \autoref{prop:proper} that, if $\nu$ is fixed, the posterior distribution (of $\bbeta$) is proper if the prior is bounded and if $k \geq p$. These conditions are satisfied because we assume that $\pi$ is bounded, and that $k\geq \lceil\lambdal / \lambdar\rceil(l + s) + 2p-1 \geq p$. Note that this implies that $m(\by_k) < \infty$ and $m(\by) < \infty$ for all $\omega$.

To prove that we can interchange the limit and the integral, we use Lebesgue's dominated convergence theorem. We thus need to prove that the integrand is bounded by an integrable function of $\bbeta$ that does not depend on $\omega$. Therefore, we need to show that
\begin{align*}
\pi(\bbeta \mid \by_k) \prod_{i=1}^n\left[\frac{f_{\nu,c}(y_i/\mu_i)/\mu_i}{f_{\nu,c}(y_i)}\right]^{s_i + l_i}
& \propto\frac{\pi(\bbeta)\prod_{i=1}^n[f_{\nu,c}(y_i/\mu_i)/\mu_i]}{\prod_{i=1}^nf_{\nu,c}(y_i)^{s_i+l_i}}\\
&= \pi(\bbeta)\prod_{i=1}^n\left[\frac{ f_{\nu,c}(y_i/\mu_i)/\mu_i}{f_{\nu,c}(y_i)}\right]^{s_i} \prod_{i=1}^n\left[\frac{f_{\nu,c}(y_i/\mu_i)/\mu_i}{f_{\nu,c}(y_i)}\right]^{l_i} \prod_{i=1}^n\left[f_{\nu,c}(y_i/\mu_i)/\mu_i\right]^{k_i}\\
&\leq B \prod_{i=1}^n \left[\frac{f_{\nu,c}(y_i/\mu_i)/\mu_i}{f_{\nu,c}(y_i)}\right]^{s_i} \prod_{i=1}^n\left[\frac{f_{\nu,c}(y_i/\mu_i)/\mu_i}{f_{\nu,c}(y_i)}\right]^{l_i} \prod_{i=1}^n\left[f_{\nu,c}(y_i/\mu_i)/\mu_i\right]^{k_i}\\
&= g(\bbeta) \, h(\bbeta, \omega),
\end{align*}
with $g$ an integrable function and $h$ a bounded function, where we used that $\pi(\bbeta) \leq B$ for all $\bbeta$ with $B$ a positive constant. The functions $g$ and $h$ are defined below.

Under the assumptions of \autoref{thm:robustness}, we know that there are at least $\lceil\lambdal / \lambdar\rceil(l + s) + 2p-1$ non-outliers in the data set. Without loss of generality, assume that the first $\lceil\lambdal / \lambdar\rceil(l + s) + 2p-1$ points are non-outliers, that is $k_1, \ldots, k_{\lceil\lambdal / \lambdar\rceil(l + s) + 2p-1} =1$.

\textbf{Step 1.} We first choose $p$ points among the non-outliers. Without loss of generality, we choose $(\bx_1,y_1),\ldots,(\bx_p,y_p)$. We show that $g$ defined as $g(\bbeta): = B \prod_{i=1}^p f_{\nu,c}(y_i/\mu_i)/\mu_i$ is integrable. We have, similarly as in the proof of \autoref{prop:proper},
\begin{align*}
\int_{\R^p}B\prod_{i=1}^p\frac{f_{\nu,c}(y_i/\exp(\bx_i^T\bbeta))}{\exp(\bx_i^T\bbeta)}\,\d\bbeta &= B\Bigg|\det\left(\begin{matrix}
\bx_1^T \\
\vdots \\
\bx_p^T
\end{matrix}\right)\Bigg|^{-1}\prod_{i=1}^p\frac{1}{y_i}\prod_{i=1}^p\int_{0}^{\infty}f_{\nu,c}(z_i)\,\d z_i\\
&=B\Bigg|\det\left(\begin{matrix}
\bx_1^T \\
\vdots \\
\bx_p^T
\end{matrix}\right)\Bigg|^{-1}\prod_{i=1}^p\frac{1}{y_i}<\infty,
\end{align*}
where we use the change of variables $z_i=y_i/\exp(\bx_i^T\bbeta)$, for $i=1, \ldots, p$. The determinant term is different from 0 because $\bx_1, \ldots, \bx_p$ are linearly independent (because the covariates are continuous). Given that these $p$ observations are non-outlying, $\prod_{i=1}^p\frac{1}{y_i}$ is bounded and independent of $\omega$.

\textbf{Step 2.} We show that the rest of the product, i.e.
\begin{align*}
  h(\bbeta, \omega):=\prod_{i=(\lambdal / \lambdar)(l + s) + 2p}^n\left[\frac{ f_{\nu,c}(y_i/\mu_i)/\mu_i}{f_{\nu,c}(y_i)}\right]^{s_i}\prod_{i=(\lambdal / \lambdar)(l + s) + 2p}^n\left[\frac{ f_{\nu,c}(y_i/\mu_i)/\mu_i}{f_{\nu,c}(y_i)}\right]^{l_i}\prod_{i=p+1}^n\left[f_{\nu,c}(y_i/\mu_i)/\mu_i\right]^{k_i}
\end{align*}
is bounded, and that the bound does not depend on neither $\bbeta$ nor $\omega$.

In order to show this, we split the domain of $\bbeta$. Doing so will allow for technical arguments yielding bounds that do not depend on neither $\bbeta$ nor $\omega$. Before presenting the precise split of the domain, we provide intuition of why it is useful to proceed like that. Consider $i$ with $l_i = 1$. The main difficulty in bounding
\[
\frac{f_{\nu,c}(y_i/\mu_i)/\mu_i}{f_{\nu,c}(y_i)}
\]
essentially resides in dealing with the term $f_{\nu,c}(y_i)$ in the denominator because it is small. The case $i$ with $l_i = 1$ is the easiest to provide intuition. The case $i$ with $s_i = 1$ is analogous but it is more difficult to provide intuition. The goal is thus essentially to get rid of $f_{\nu,c}(y_i)$. When $\mathbf{x}_i^T\boldsymbol\beta \leq \log(\omega) / 2$, say, $y_i/\mu_i =b_i\omega/\exp(\mathbf{x}_i^T\boldsymbol\beta) \geq b_i \sqrt{\omega}$ is large and, for any fixed $c$ and $\nu$, a corollary of \autoref{prop:limit_PDF} can be used to bound
\[
\frac{f_{\nu,c}(y_i/\mu_i)/\mu_i}{f_{\nu,c}(y_i)}.
\]
 When $\mathbf{x}_i^T\boldsymbol\beta > \log(\omega) / 2$, we are not guaranteed that $y_i/\mu_i$ is large and thus cannot use the PDF term of the outlier to bound $1/f_{\nu,c}(y_i)$. We thus have to resort to non-outliers. With non-outliers, we consider $y_j$ as fixed, and it is only when $1/\exp(\mathbf{x}_j^T\boldsymbol\beta)$ is large that the PDF term $f_{\nu,c}(y_j/\mu_j)/\mu_j$ can be used to bound $1/f_{\nu,c}(y_i)$.
 The strategy used below to deal with the situation is to divide the parameter space  in mutually exclusive areas for which we know exactly in which case we are: either we can use the outlier PDF term to bound $1/f_{\nu,c}(y_i)$ or not; in the latter case, we know that we have sufficiently non-outliers that can be used to bound all terms $1/f_{\nu,c}(y_i)$. To have a precise control over the number of non-outliers that can be used, we prove that when we cannot use PDF terms of outliers to bound $1/f_{\nu,c}(y_i)$, we have a maximum of $p - 1$ non-outliers that cannot be used to that job either. Using that $k\geq \lceil\lambdal / \lambdar\rceil(l + s) + 2p-1$, we know that at least $\lceil\lambdal / \lambdar\rceil(l + s)$ non-outlying points can be used to bound the terms $1/f_{\nu,c}(y_i)$, which will be shown to be sufficient (recall that $p$ non-outlying points have already been used to obtain a integrable function).

Let us now continue with the formal proof and present how we split the domain of $\bbeta$:
\begin{align*}
    &\R^p=\left[\cap_i\mathcal{O}_i^\mathsf{c}\right]\cup\left[\cup_i(\mathcal{O}_i \cap(\cap_{i_1}\mathcal{F}_{i_1}^\mathsf{c}))\right] \cup\left[\cup_{i,i_1}(\mathcal{O}_i\cap\mathcal{F}_{i_1}\cap(\cap_{i_2\neq i_1}\mathcal{F}_{i_2}^\mathsf{c}))\right]\\
    &\cup\ldots\cup\left[\cup_{i,i_1,\ldots,i_{p-1}(i_j\neq i_s, \forall i_j, i_s\ \text{s.t.} j\neq s)}\left(\mathcal{O}_i\cap\mathcal{F}_{i_1}\cap\ldots\cap\mathcal{F}_{i_{p-1}}\cap\left(\cap_{i_p\neq i_1,\ldots, i_{p-1}}\mathcal{F}_{i_p}^\mathsf{c}\right)\right)\right]\\
    &\qquad\cup\left[\cup_{i,i_1,\ldots,i_p(i_j\neq i_s, \forall i_j, i_s\ \text{s.t.} j\neq s)}\left(\mathcal{O}_i\cap\mathcal{F}_{i_1}\cap\ldots\mathcal{F}_{i_p}\right)\right],
\end{align*}
where \begin{align*}
   &\mathcal{O}_i:=
\begin{cases}
    \left\{\bbeta:\log(b_i\omega)-\bx_i^T\bbeta<\log(\omega)/2\right\}& \text{if } i \in\mathcal{I_L},\\
    \left\{\boldsymbol{\beta}: \bx_i^T\bbeta-\log(1/b_i \omega) < \log(\omega) / 2\right\} & \text{if } i \in\mathcal{I_S},
\end{cases}\\
   &\mathcal{F}_i:=\left\{\boldsymbol{\beta}:|\bx_i^T \bbeta| < \log(\omega)/\gamma\right\}\ \text{if }\ i\in \mathcal{I_F},
\end{align*}
with $\mathcal{I_L}$, $\mathcal{I_S}$, and $\mathcal{I_F}$ defined as follows:
\begin{align*}
   \mathcal{I_L}&:=\{i:i\in\{\lceil\lambdal / \lambdar\rceil(l + s) + 2p,\ldots,n\}\ \text{and}\ l_i=1\},\\
   \mathcal{I_S}&:=\{i:i\in\{\lceil\lambdal / \lambdar\rceil(l + s) + 2p,\ldots,n\}\ \text{and}\ s_i=1\},\\
   \mathcal{I_F}&:=\{p+1,\ldots,\lceil\lambdal / \lambdar\rceil(l + s) + 2p-1\},
\end{align*}
$\gamma$ being a positive constant that will be defined.

Remember that the first $p$ points, which are non-outliers, have already been used for the purpose of integration in Step 1. Thus, the index of each remaining non-outliers is greater than or equal to $p+1$.

The set $\mathcal{O}_i$ represents the hyperplanes $\bx_i^T\bbeta$ characterized by the different values of $\bbeta$ satisfying $\log(b_i\omega)-\bx_i^T\bbeta<\log(\omega)/2$ for $i\in\mathcal{I_L}$, and $\log(b_i/\omega)-\bx_i^T\bbeta<\log(\omega)/2$ for $i\in\mathcal{I_S}$. The points $(\bx_i, \log(b_i\omega))$ and $(\bx_i, \log(b_i/\omega))$ can be seen as log transformations of large outliers and of small outliers, respectively, given that $\omega\rightarrow\infty$.

Now we claim that $\mathcal{O}_i\cap\mathcal{F}_{i_1}\dots\cap\mathcal{F}_{i_p}=\varnothing$ for all $i,i_1,\ldots,i_p$ with $i_j\neq i_s, \forall i_j,i_s$ such that $j\neq s$. To prove this, we use the fact that $\bx_i$ (a vector of dimension $p$) can be expressed as a linear combination of $\bx_{i_1}, \ldots, \bx_{i_p}$. This is true because all explanatory variables are continuous, therefore the space spanned by the vectors $\bx_{i_1}, \ldots, \bx_{i_p}$ has dimension $p$. As a result, if $\bbeta\in \mathcal{F}_{i_1}\cap\ldots\cap\mathcal{F}_{i_p}$ and $\bx_i=\sum_{s=1}^p a_s\bx_{i_s}$, for some $a_1, \ldots, a_p\in\R$, and
\begin{itemize}
    \item if $i\in\mathcal{I_L}$,
    \begin{align*}
        \log(b_i\omega)-\bx_i^T\bbeta = \log(b_i\omega)-\left(\sum_{s=1}^p a_s\bx_{i_s}\right)^T\bbeta & \overset{a}{\geq}\log(\omega)-\sum_{s=1}^pa_s\bx_{i_s}^T\bbeta\\
        &\overset{b}{>} \log(\omega)-\frac{\log(\omega)}{\gamma}\sum_{s=1}^p a_s \cr
        & \overset{c}{\geq} \log(\omega)/2;\\
    \end{align*}
    \item if $i\in\mathcal{I_S}$,
    \begin{align*}
        \log(1/b_i\omega)-\bx_i^T\bbeta&=\log(1/b_i\omega)-\left(\sum_{s=1}^p a_s\bx_{i_s}\right)^T\bbeta \cr
        &<-\log(\omega)-\sum_{s=1}^pa_s\bx_{i_s}^T\bbeta\\
        &\overset{d}{<}-\log(\omega)+\left(\frac{\log(\omega)}{\gamma}\sum_{s=1}^p a_s\right)\overset{e}{\leq}-\log(\omega)/2.
    \end{align*}
\end{itemize}
In Step $a$, we use that $b_i\geq 1$ and we simplify the form of the linear combination. In Step $b$, because $\bbeta\in\mathcal{F}_{i_1}\cap\ldots\cap\mathcal{F}_{i_p}$, we have $\bx_i^T\bbeta<\log(\omega)/\gamma$ for all $i\in\{i_1,\ldots, i_p\}$. Thus, $-\sum_{s=1}^p a_s\bx^T_{i_s}\bbeta > -(\log(\omega)/\gamma)\sum_{s=1}^p a_s$. In Step $c$, we define the constant $\gamma$ such that $\gamma\geq 2\sum_{s=1}^p a_s$ (we define $\gamma$ such that it satisfies this inequality for any combination of $i$ and $i_1, \ldots i_p$; without loss of generality, we consider that $\gamma \geq 1$). The proof is analogous in the case where $i\in\mathcal{I_S}$. In Step $d$, we use the fact that $\bx_{i}^T\bbeta>-\log(\omega)/\gamma$, thus $-\sum_{s=1}^p a_s \bx_{i_s}^T\bbeta < \log(\omega)/\gamma$. In Step $e$, we use that $\gamma$ is such that $\gamma\geq 2\sum_{s=1}^p a_s$.

Therefore, we have that if $\bbeta\in\mathcal{F}_{i_1}\cap\ldots\cap\mathcal{F}_{i_p}$, then $\bbeta\notin\mathcal{O}_i$. This proves that $\mathcal{O}_i\cap\mathcal{F}_{i_1}\cap\ldots\cap\mathcal{F}_{i_p}=\varnothing$ for all $i,i_1,\ldots, i_p$ with $i_j\neq i_s, \forall i_j,i_s$ such that $j\neq s$. This result in turn implies that the domain of $\bbeta$ can be written as
\begin{align*}
    &\R^p=\left[\cap_i\mathcal{O}_i^\mathsf{c}\right]\cup\left[\cup_i(\mathcal{O}_i\cap(\cap_{i_1}\mathcal{F}_{i_1}^\mathsf{c}))\right] \cup\left[\cup_{i,i_1}(\mathcal{O}_i\cap\mathcal{F}_{i_1}\cap(\cap_{i_2\neq i_1}\mathcal{F}_{i_2}^\mathsf{c}))\right]\\
    &\cup\ldots\cup\left[\cup_{i,i_1,\ldots,i_{p-1}(i_j\neq i_s \text{s.t.} j\neq s)}\left(\mathcal{O}_i\cap\mathcal{F}_{i_1}\cap\ldots\cap\mathcal{F}_{i_{p-1}}\cap\left(\cap_{i_p\neq i_1,\ldots, i_{p-1}}\mathcal{F}_{i_p}^\mathsf{c}\right)\right)\right].
\end{align*}
This decomposition of $\R^p$ consists of $1+\sum_{i=0}^{p-1}\binom{\lceil\lambdal / \lambdar\rceil(l + s)+p-1}{i}$ mutually exclusive sets given by $\cap_i\mathcal{O}_i^\mathsf{c}$, $\cup_i(\mathcal{O}_i\cap(\cap_{i_1}\mathcal{F}_{i_1}^\mathsf{c}))$, $\cup_i(\mathcal{O}_i\cap\mathcal{F}_{i_1}\cap(\cap_{i_2\neq i_1}\mathcal{F}_{i_2}^\mathsf{c}))$ for $i_1\in\mathcal{I_F}$, and so on.

We find an upper bound on each of these subsets. Because there is a finite number of subsets, we will be able to bound $h$ by the maximal bound. Recall that
\begin{align*}
   h(\bbeta, \omega)&=\prod_{i=\lceil\lambdal / \lambdar\rceil(l + s) + 2p}^n\left[\frac{ f_{\nu,c}(y_i/\mu_i)/\mu_i}{f_{\nu,c}(y_i)}\right]^{l_i}\prod_{i=\lceil\lambdal / \lambdar\rceil(l + s) + 2p}^n\left[\frac{ f_{\nu,c}(y_i/\mu_i)/\mu_i}{f_{\nu,c}(y_i)}\right]^{s_i}\prod_{i=p+1}^n\left[f_{\nu,c}(y_i/\mu_i)/\mu_i\right]^{k_i}\\
   &:=A\times B \times C,
\end{align*}
where $A$, $B$ and $C$ represent the product on the left, the product in the middle and the product on the right, respectively.

\textbf{Situation 1.} If $\bbeta\in\cap_i\mathcal{O}_i^\mathsf{c}$, we have
\begin{align*}
A  \overset{a}{=} \prod_{i=\lceil\lambdal / \lambdar\rceil(l + s) + 2p}^n\left(\frac{\log(y_i)}{\log(y_i)-\log(\mu_i)}\right)^{l_i\lambdar} &\overset{b}{\leq} \prod_{i=\lceil\lambdal / \lambdar\rceil(l + s) + 2p}^n\left(\frac{\log(b_i)+\log(\omega)}{\log(\omega)/2}\right)^{l_i\lambdar} \\
& = \prod_{i=\lceil\lambdal / \lambdar\rceil(l + s) + 2p}^n\left(2+2\frac{\log(b_i)}{\log(\omega)}\right)^{l_i\lambdar} \overset{c}{\leq}3^{l\lambdar}\overset{d}{<} \infty.
\end{align*}
In Step $a$, if $\bbeta\in\cap_i\mathcal{O}_i^\mathsf{c}$, it means that  $y_i/\mu_i = b_i\omega/\mu_i\geq\sqrt{\omega}$ for $i \in \mathcal{I_L}$. We are thus sure that $b_i\omega/\mu_i$ and $b_i\omega$ are both on the right tail of $f_{\nu,c}$, that is $f_{\text{right}}$. In Step $b$, we use that $\log(b_i\omega)-\bx_i^T\bbeta\geq\log(\omega)/2$. In Step $c$, we have $2\log(b_i)/\log(\omega)\leq1$ for large enough $ \omega$. In Step $d$, for any fixed $\nu$, $3^{l \lambda_r}$ is finite given that $\lambda_r$, which depends only on $c$ and $\nu$, is finite.

For $B$, we have
\begin{align*}
B &\overset{a}{=} \prod_{i=\lceil\lambdal / \lambdar\rceil(l + s) + 2p}^n\left(\frac{\log(y_i)}{\log(y_i)-\log(\mu_i)}\right)^{s_i\lambdal} \cr
&\overset{b}{\leq} \prod_{i=\lceil\lambdal / \lambdar\rceil(l + s) + 2p}^n\left(\frac{-\log(b_i)-\log(\omega)}{-\log(\omega)/2}\right)^{s_i\lambdal}\\
& = \prod_{i=\lceil\lambdal / \lambdar\rceil(l + s) + 2p}^n\left(2+2\frac{\log(b_i)}{\log(\omega)}\right)^{s_i\lambdal} < 3^{s\lambdal}\overset{c}{<} \infty.
\end{align*}
The proof is analogous. In Step $a$, if $\bbeta\in\cap_i\mathcal{O}_i^\mathsf{c}$, it means that  $y_i/\mu_i = (1/b_i \omega)/\mu_i\leq1/\sqrt{\omega}$ for $i \in \mathcal{I_S}$. We are thus sure that $(1/b_i \omega)/\mu_i$ and $1/ b_i \omega$ are both on the left tail of $f_{\nu,c}$, that is $f_{\text{left}}$. In Step $b$, we use that $\log(y_i)-\bx_i^T\bbeta\leq-\log(\omega)/2$. In Step $c$, $3^{s\lambdal}$ is finite given that $\lambdal$, depending only on $c$ and $\nu$, is finite.

For $C$, we have
\begin{align*}
   \prod_{i=p+1}^n\left[f_{\nu,c}(y_i/\mu_i)/\mu_i\right]^{k_i}\overset{a}{\leq}\prod_{i=p+1}^n\left[\frac{(\e^{-1}\nu)^\nu}{y_i\Gamma(\nu)}\right]^{k_i}<\infty.
\end{align*}
In Step $a$, according to \autoref{lemma1}, $f_{\nu,c}(y/\mu)/\mu$ is upper bounded by $(\e^{-1}\nu)^\nu/(y\Gamma(\nu))$ for any value of $\mu$, when $y, \nu$ and $c$ are considered fixed.

We conclude that in this situation, $A\times B\times C$ is bounded.

\textbf{Situation 2.} Consider now that $\bbeta$ belongs to one of the $\sum_{i=1}^{p-1}\binom{\lceil\lambdal / \lambdar\rceil(l + s)+p-1}{i}$ mutually exclusive sets $\cup_i(\mathcal{O}_i\cap(\cap_{i_1}\mathcal{F}_{i_1}^\mathsf{c}))$, $\cup_{i}(\mathcal{O}_i\cap\mathcal{F}_{i_1}\cap(\cap_{i_2\neq i_1}\mathcal{F}_{i_2}^\mathsf{c}))$ for $i_1\in \mathcal{I_F}$, etc.

We analyse $A$, $B$ and $C$ separately. We have
\begin{align}
\label{partA}
A &=  \prod_{i=\lceil\lambdal / \lambdar\rceil(l + s) + 2p}^n\left[\frac{1}{\mu_i} f_{\nu,c}(y_i/\mu_i)\right]^{l_i}\prod_{i=\lceil\lambdal / \lambdar\rceil(l + s) + 2p}^n\left[\frac{1}{f_{\nu,c}(y_i)}\right]^{l_i}\nonumber\\
&\propto\prod_{i=\lceil\lambdal / \lambdar\rceil(l + s) + 2p}^n\left[\frac{1}{\mu_i} f_{\nu,c}(y_i/\mu_i)\right]^{l_i}\prod_{i=\lceil\lambdal / \lambdar\rceil(l + s) + 2p}^n\left[y_i\left(\log(y_i)\right)^{\lambdar}\right]^{l_i}\nonumber\\
& = \prod_{i=\lceil\lambdal / \lambdar\rceil(l + s) + 2p}^n\left[\frac{y_i}{\mu_i} f_{\nu,c}(y_i/\mu_i)\right]^{l_i}\prod_{i=\lceil\lambdal / \lambdar\rceil(l + s)  + 2p}^n\left(\log(b_i\omega)\right)^{l_i\lambdar}\nonumber\\
& \overset{a}{\leq} \left(\frac{(\e^{-1}\nu)^\nu}{\Gamma(\nu)}\right)^l\prod_{i=\lceil\lambdal / \lambdar\rceil(l + s) + 2p}^n\left(\log(b_i\omega)\right)^{l_i\lambdar}.
\end{align}
In Step $a$, we can deduce from \autoref{lemma1} that, viewed as a function of $\mu$, $(y /\mu) f_{\nu,c}(y / \mu)$ is bounded by $(\e^{-1}\nu)^\nu/\Gamma(\nu)$, for all $\nu, c$, and $y$. Notice that the only part that depends on $\omega$ in \eqref{partA} is a product of terms $\log(b_i\omega)^{\lambdar}$, one for each large outlier.

Analogously, for $B$, we have
\begin{align}
\label{partB}
 B &=  \prod_{i=\lceil\lambdal / \lambdar\rceil(l + s) + 2p}^n\left[\frac{1}{\mu_i} f_{\nu,c}(y_i/\mu_i)\right]^{s_i}\prod_{i=\lceil\lambdal / \lambdar\rceil(l + s) + 2p}^n\left[\frac{1}{f_{\nu,c}(y_i)}\right]^{s_i}\nonumber\\
&=\prod_{i=\lceil\lambdal / \lambdar\rceil(l + s) + 2p}^n\left[\frac{1}{\mu_i} f_{\nu,c}(y_i/\mu_i)\right]^{s_i}\prod_{i=\lceil\lambdal / \lambdar\rceil(l + s) + 2p}^n\left[\frac{y_i}{\zl}\left(\frac{\log(y_i)}{\log(\zl)}\right)^{\lambdal}\right]^{s_i}\nonumber\\
& \overset{a}{\propto} \prod_{i=\lceil\lambdal / \lambdar\rceil(l + s) + 2p}^n\left[\frac{y_i}{\mu_i} f_{\nu,c}(y_i/\mu_i)\right]^{s_i}\prod_{i=\lceil\lambdal / \lambdar\rceil(l + s) + 2p}^n\left(-\log(1/b_i \omega)\right)^{s_i\lambdar}\nonumber\\
&\overset{b}{\leq} \left(\frac{(\e^{-1}\nu)^\nu}{\Gamma(\nu)}\right)^s\prod_{i=\lceil\lambdal / \lambdar\rceil(l + s)+2p}^n\left(\log(b_i \omega)\right)^{s_i\lambdal}.
\end{align}
In Step $a$, we change the sign of $\log(1/b_i \omega)$ because $\log(\zl)<0$. In Step $b$, we bound $(y /\mu) f_{\nu,c}(y / \mu)$ by $(\e^{-1}\nu)^\nu/\Gamma(\nu)$. Notice that the only part that depends on $\omega$ in \eqref{partB} is a product of terms $\log(b_i\omega)^{\lambdal}$, one for each small outlier.

We now turn to $C$. We have shown previously that in any of the sets to which $\bbeta$ can belong, there are at most $p-1$ non-outlying points such that $|\bx_i^T\bbeta|< \log(\omega)/\gamma$. The case where that upper bound is attained is that where $\bbeta\in\cup_i(\mathcal{O}_i\cap\mathcal{F}_{i_1}\cap\ldots\cap\mathcal{F}_{i_{p-1}}\cap(\cap_{i_p\neq i_1,\ldots,i_{p-1}}\mathcal{F}_{i_p}^\mathsf{c}))$. Without loss of generality, suppose that all non-outlying points such that $|\bx_i^T \bbeta| < \log(\omega)/\gamma$ have index $i$ belonging to $\{p+1, ..., 2p-1\}$. There are thus at least $\lceil\lambdal / \lambdar\rceil(l + s)+p-1-(p-1)=\lceil\lambdal / \lambdar\rceil(l + s)$ remaining non-outliers such that $|\bx_i^T\bbeta|\geq \log(\omega)/\gamma$, with $i=2p,\ldots, \lceil\lambdal / \lambdar\rceil(l + s)+2p-1$. Therefore, for these points, we are sure that $y_i/\mu_i$ is on the extremities of $f_{\nu,c}$, as either $y_i/\mu_i = \exp(\log(a_i) - \bx_i^T\bbeta)) \geq a_i\omega^{1/\gamma} > \zr$ or $y_i/\mu_i  \leq a_i/ \omega^{1/\gamma} < \zl$, for large enough $\omega$.

In the situation where $y_i/\mu_i$ is on right tail, that is $y_i/\mu_i \geq a_i\omega^{1/\gamma}$, we have
\begin{align}
\label{eq:boundr}
    \frac{1}{\mu_i} f_{\nu,c}(y_i/\mu_i)\propto\frac{1}{a_i}\left(\frac{\log(\zr)}{\log(a_i)-\log(\mu_i)}\right)^{\lambdar}&\propto\left(\frac{1}{\log(a_i)-\log(\mu_i)}\right)^{\lambda_{\text{r}}} \nonumber \\
   &\overset{a}{\leq}\left(\frac{1}{\log(a_i)+\log(\omega)/\gamma}\right)^{\lambdar}
   \overset{b}{\leq}\left(\frac{2\gamma}{\log(\omega)}\right)^{\lambdar}.
\end{align}
In Step $a$, we have that $\mu_i\leq \omega^{-1/\gamma}$ given that $y_i/\mu_i \geq a_i\omega^{1/\gamma}$, thus $-\log(\mu_i)\geq \log(\omega)/\gamma$. In Step $b$, we use the fact that $\log(a_i)\geq-\log(\omega)/(2\gamma)$ for large enough $\omega$. We thus have $\log(a_i)+\log(\omega)/\gamma\geq \log(\omega)/(2\gamma)$.

In the situation where $y_i/\mu_i$ is on the left tail, that is $y_i/\mu_i \leq a_i/\omega^{1/\gamma}$, we have
\begin{align}
\label{eq:boundl}
\frac{1}{\mu_i} f_{\nu,c}(y_i/\mu_i)  \propto\frac{1}{a_i}\left(\frac{\log(\zl)}{\log(a_i)-\log(\mu_i)}\right)^{\lambdal}&\overset{a}{\propto}\left(\frac{1}{\log(\mu_i)-\log(a_i)}\right)^{\lambdal} \nonumber \\
   &\leq\left(\frac{1}{\log(\omega)/\gamma-\log(a_i)}\right)^{\lambdal}\overset{b}{\leq}\left(\frac{2\gamma}{\log(\omega)}\right)^{\lambdal}.
\end{align}
In Step $a$, we change the sign because $\log(\zl)<0$. In Step $b$, we use the fact that $\log(a_i)\leq\log(\omega)/(2\gamma)$ for large enough $\omega$.

The reason we consider these two cases is that we want to use densities of ``extreme non-outliers'', that is $f_{\nu,c}(y_j/\mu_j)/\mu_j$ with $j$ such that $\bbeta \in \mathcal{F}_j^\mathsf{c}$, to cancel each $\log(b_i \omega)$ at some power for $i\in\mathcal{I_L}\cup\mathcal{I_S}$ that appears in the bounds of $A$ and $B$ (recall \eqref{partA} and \eqref{partB}). As explained, there are at least $\lceil\lambdal / \lambdar\rceil(l + s)$ extreme non-outliers that can be used. However, the major problem here is that we do not know how many of those $\lceil\lambdal / \lambdar\rceil(l + s)$ extreme non-outliers are such that $y_j/\mu_j$ is on the right tail, and how many are such that $y_j/\mu_j$ is on the left tail, which depends on the value of $\bbeta$. We thus have to consider all possible scenarios, including the worst-case scenario. We now present clearly how we bound each $\log(b_i \omega)$ at some power for $i\in\mathcal{I_R}\cup\mathcal{I_S}$ by using the densities of extreme non-outliers, in all scenarios.

Let us first consider that $y_i$ is a large outlying observation, that is $i\in\mathcal{I_L}$. We take $\lceil\lambdal / \lambdar\rceil$ non-outliers among the $\lceil\lambdal / \lambdar\rceil(l+s)$ extreme non-outliers that are thus such that $\bbeta\in\mathcal{F}_j^\mathsf{c}$ for all of these. In other words, all these points are such that $|\bx_j^T\bbeta|\geq \log(\omega)/\gamma$. There are two possible cases.

\textbf{Case 1.} There is at least one point among the $\lceil\lambdal / \lambdar\rceil$ points such that $y_j/\mu_j\geq a_j\omega^{1/\gamma}$, implying that the density is evaluated on the right tail. In this case, we have
\begin{align*}
    \log(b_i\omega)^{\lambdar}\prod_{j=i_1}^{i_{\lceil\lambdal / \lambdar\rceil}}\frac{1}{\mu_j}f_{\nu,c}(y_j/\mu_j) &\overset{a}{\leq} \log(b_i\omega)^{\lambdar}\left(\frac{2\gamma}{\log(\omega)}\right)^{\lambdar}\\
    &=\left(\frac{2\gamma\log(b_i\omega)}{\log(\omega)}\right)^{\lambda_{\text{r}}}\overset{b}{\leq}\left(4\gamma\right)^{\lambda_{\text{r}}}\overset{c}{<}\infty.
\end{align*}
In Step $a$, we take one point such that $y_j/\mu_j\geq a_j\omega^{1/\gamma}$. We bound $f_{\nu,c}(y_j/\mu_j)/\mu_j$ by the bound presented in \eqref{eq:boundr}, and we bound the rest of the points by 1 using the bounds in \eqref{eq:boundr} and \eqref{eq:boundl}, given that $2 \gamma /\log(\omega) \leq 1$ for large enough $\omega$. In Step $b$, we have that $\log(b_i\omega)/\log(\omega) = (\log(b_i)+\log(\omega))/\log(\omega)\leq2$, as $\log(b_i)/\log(\omega)\leq 1$ for large enough $\omega$. In Step $c$, every term is a well-defined constant, the result is thus finite.

\textbf{Case 2.} No point among the $\lceil\lambdal / \lambdar\rceil$ points is such that $y_j/\mu_j\geq a_j\omega^{1/\gamma}$, implying that the density of every point is evaluated on the left tail. In this case, we have
\begin{align*}
    \log(b_i\omega)^{\lambdar}\prod_{j=i_1}^{i_{\lceil\lambdal / \lambdar\rceil}}\frac{1}{\mu_j}f_{\nu,c}(y_j/\mu_j)
    &\overset{a}{\leq}\log(b_i\omega)^{\lambdar}\left(\frac{2\gamma}{\log(\omega)}\right)^{\lceil\lambdal / \lambdar\rceil\lambdal}\\
    &=\left(\frac{2\gamma\log(b_i\omega)}{\log(\omega)}\right)^{\lambdar}\left(\frac{2\gamma}{\log(\omega)}\right)^{\lceil\lambdal / \lambdar\rceil\lambdal-\lambdar}\\
    &\overset{b}{\leq}\left(4\gamma\right)^{\lambdar}<\infty.
\end{align*}
In Step $a$, we bound every term $f_{\nu,c}(y_j/\mu_j)/\mu_j$ by the bound in \eqref{eq:boundl}. In Step $b$, given that $\lambda_{\text{l}}/\lambda_{\text{r}}\geq 1$ (recall \autoref{fig:lambdas}), we know that $\lceil\lambdal / \lambdar\rceil \lambdal \geq \lambdar$. We also use that $2\gamma/\log(\omega)\leq 1$, and $\log(b_i\omega)/\log(\omega) = (\log(b_i)+\log(\omega))/\log(\omega)\leq2$, as $\log(b_i)/\log(\omega)\leq 1$ for large enough $\omega$.

We showed that we can use the product of the densities of $\lceil\lambdal / \lambdar\rceil$ extreme non-outliers to offset $\log(b_i \omega)^{\lambdar}$ for $i\in \mathcal{I_L}$, so that the product
\[
 \log(b_i\omega)^{\lambdar}\prod_{j=i_1}^{i_{\lceil\lambdal / \lambdar\rceil}}\frac{1}{\mu_j}f_{\nu,c}(y_j/\mu_j)
\]
is bounded. The approach is analogous for small outliers, that is for $i\in\mathcal{I_S}$. A difference is that, in Case 1, we instead consider that there is at least one point among the $\lceil\lambdal / \lambdar\rceil$ points such that $y_i/\mu_i \leq a_i/\omega^{1/\gamma}$, implying that the density is evaluated on the left tail. Also, in Case 2, we instead consider that no point among the $\lceil\lambdal / \lambdar\rceil$ points is such that $y_j/\mu_j\leq a_j\omega^{1/\gamma}$, implying that the density of every point is evaluated on the right tail. In this case, we have a product
\[
\left(\frac{2\gamma\log(b_i\omega)}{\log(\omega)}\right)^{\lambdal}\left(\frac{2\gamma}{\log(\omega)}\right)^{\lceil\lambdal / \lambdar\rceil\lambdar-\lambdal}
\]
that appears and we bound the right term by 1 as above using that $\lceil\lambdal / \lambdar\rceil\lambdar-\lambdal \geq (\lambdal / \lambdar)\lambdar-\lambdal = 0$. We therefore know that we can offset $\log(b_i \omega)$ at some power for $i\in \mathcal{I_L}\cup \mathcal{I_S}$ using the product of the densities of $\lceil\lambdal / \lambdar\rceil$ extreme non-outliers, in all scenarios.

Therefore, if we multiply now $A$, $B$ and $C$, we obtain that
\begin{align*}
    A\times B\times C&\overset{a}{\leq} \left[\left(\frac{(\e^{-1}\nu)^\nu}{\Gamma(\nu)}\right)^{s+l}\prod_{i=\lceil\lambdal / \lambdar\rceil(l + s)+2p}^n\left(\log(b_i\omega)\right)^{l_i\lambdar+s_i\lambdal}\right]\prod_{i=p+1}^n\left(\frac{f_{\nu,c}(y_i/\mu_i)}{\mu_i}\right)^{k_i}\\
    & \overset{b}{\leq}\left[ \left(\frac{(\e^{-1}\nu)^\nu}{\Gamma(\nu)}\right)^{s+l}\prod_{i=\lceil\lambdal / \lambdar\rceil(l + s)+2p}^n\left(\log(b_i\omega)\right)^{l_i\lambdar+s_i\lambdal}\right]\prod_{i=p+1}^{2p-1}\left(\frac{(\e^{-1}\nu)^\nu}{y_i\Gamma(\nu)}\right)\prod_{i=2p}^n\left(\frac{f_{\nu,c}(y_i/\mu_i)}{\mu_i}\right)^{k_i}\\
    & \overset{c}{\leq}\left[ \left(\frac{(\e^{-1}\nu)^\nu}{\Gamma(\nu)}\right)^{s+l}\right]\left[\prod_{i=p+1}^{2p-1}\left(\frac{(\e^{-1}\nu)^\nu}{y_i\Gamma(\nu)}\right)\right]\left[\prod_{i=\lceil\lambdal / \lambdar\rceil(l + s)+2p}^n\left(\frac{f_{\nu,c}(y_i/\mu_i)}{\mu_i}\right)^{k_i}\right]\left[\left(4\gamma\right)^{l\lambdar+s\lambdal}\right]\\
    &\overset{d}{\leq}\left[\left(\frac{(\e^{-1}\nu)^\nu}{\Gamma(\nu)}\right)^{n-p-\lceil\lambdal / \lambdar\rceil(l + s)}\left(\prod_{i=p+1}^{2p-1}\frac{1}{y_i}\right)\prod_{i=\lceil\lambdal / \lambdar\rceil(l + s)+2p}^{n}\left(\frac{1}{y_i}\right)^{k_i}\right]\left[\left(4\gamma\right)^{l\lambdar+s\lambdal}\right]\\
    &\overset{e}{=}\left[\left(\frac{(\e^{-1}\nu)^\nu}{\Gamma(\nu)}\right)^{n-p-\lceil\lambdal / \lambdar\rceil(l + s)}\left(\prod_{i=p+1}^{2p-1}\frac{1}{a_i}\right)\prod_{i=\lceil\lambdal / \lambdar\rceil(l + s)+2p}^{n}\left(\frac{1}{a_i}\right)^{k_i}\right]\left[\left(4\gamma\right)^{l\lambdar+s\lambdal}\right]\\
    &\overset{f}{<}\infty.
\end{align*}

In Step $a$, we bound $A$ and $B$ by expressions that are previously shown (see \eqref{partA} and \eqref{partB}). In Step $b$, we bound $f_{\nu,c}(y_i/\mu_i)/\mu_i$, for $i=p+1,\ldots,2p-1$, by $(\e^{-1}\nu)^\nu/(y_i\Gamma(\nu))$ (see \autoref{lemma1}). Recall that these are non-outliers such that $|\bx_i^T\bbeta|<\log(\omega)/\gamma$. In Step $c$, we simplify each $\log(b_i\omega)^{\lambdar}$ and $\log(b_i\omega)^{\lambdal}$ by multiplying each of these terms by the product of densities of $\lceil\lambdal / \lambdar\rceil$ extreme non-outliers and by bounding the resulting product by $\left(4\gamma\right)^{\lambdar}$ or $\left(4\gamma\right)^{\lambdal}$, as we have explained earlier. In Step $d$, we bound the rest of the non-outlier terms $f_{\nu,c}(y_i/\mu_i)/\mu_i$ by $(\e^{-1}\nu)^\nu/(y_i\Gamma(\nu))$. If we consider all the $k$ non-outliers, $p$ non-outliers were used for the change of variables and to integrate over $\bbeta$ at the beginning, we bounded $p - 1$ terms of non-extreme non-outliers, and $\lceil\lambdal / \lambdar\rceil(l + s)$ were used to offset the outliers. After Steps (a)-(c), there are thus possibly still $k-p-(p-1)-\lceil\lambdal / \lambdar\rceil(l+s)=k-2p-\lceil\lambdal / \lambdar\rceil(l+s)+1$ non-outliers left, that need to be considered, which is what was done in Step d. The condition of this theorem $k\geq \lceil\lambdal / \lambdar\rceil(l+s)+2p-1$ is to make sure that we have enough non-outlying points to bound the whole product. The proof is simpler and still valid if there is no outlier left after Steps (a)-(c); Step (d) is simply skipped.  In Step $e$, every $y_i$ in the expression is a non-outlying observation, and is thus equal to $a_i$. Finally, in Step $f$, the whole expression is finite given that all terms are constants.

Therefore, $h(\bbeta, \omega)=A\times B\times C$ is bounded. This completes the proof of Result (a).

We now turn to the proof of Result (b). We have that
\begin{align*}
    \pi(\bbeta\mid\by) = \pi(\bbeta\mid\by_k) \, \frac{m(\by_k)\prod_{i=1}^n f_{\nu,c}(y_i)^{s_i+l_i}}{m(\by)}\prod_{i=1}^n\left[\frac{f_{\nu,c}(y_i/\mu_i)/\mu_i}{f_{\nu,c}(y_i)}\right]^{s_i+l_i},
\end{align*}
and
\begin{align*}
\frac{m(\by_k)\prod_{i=1}^n f_{\nu,c}(y_i)^{s_i+l_i}}{m(\by)}\prod_{i=1}^n\left[\frac{f_{\nu,c}(y_i/\mu_i)/\mu_i}{f_{\nu,c}(y_i)}\right]^{s_i+l_i} \rightarrow 1,
\end{align*}
as $\omega \rightarrow\infty$, for any $\bbeta\in\R^p$, using Result (a) and \autoref{prop:limit_PDF}. We also showed that $\pi(\beta\mid \by_k)$ is proper. This concludes the proof of Result (b).

We finish with the proof of Result (c). This result is a direct consequence of Result (b) using Scheffé's theorem \citep{scheffe1947useful}.
\end{proof}

\section{Supplementary material for \autoref{sec:case_study}}\label{sec:supp_case_study}

With gamma GLM, the posterior density $\pi(\, \cdot \,, \cdot \mid \by)$ is such that
\[
 \pi(\bbeta, \nu \mid \by) \propto \pi(\bbeta, \nu) \prod_{i = 1}^n \frac{1}{\mu_i} f_\nu\left(\frac{y_i}{\mu_i}\right),
\]
with
\[
 f_\nu(z) := \frac{\exp(-\nu z) \, z^{\nu - 1} \nu^\nu}{\Gamma(\nu)}, \quad z > 0.
\]
The posterior estimates are computed using Hamiltonian Monte Carlo. We apply a change of variables $\eta := \log \nu$ so that the resulting posterior distribution has $\R^{p+1}$ for support. The resulting posterior density is such that
\[
 \pi(\bbeta, \eta \mid \by) \propto \e^{\eta} \, \pi(\e^{\eta}) \prod_{i = 1}^n \frac{1}{\mu_i} f_{\e^{\eta}}\left(\frac{y_i}{\mu_i}\right);
\]
recall the prior on the parameters. The log density is such that (if we omit the normalization constant)
\[
 \log \pi(\bbeta, \eta \mid \by) = \eta + \log \pi( \e^{\eta}) + \sum_{i = 1}^n \log f_{\e^{\eta}}\left(\frac{y_i}{\mu_i}\right) - \log \mu_i.
\]
We have that
\[
 \frac{\partial}{\partial \eta} \log \pi( \e^{\eta}) = (\alpha - 1) - \frac{\e^\eta}{\theta},
\]
where $\alpha > 0$ and $\theta > 0$ are the shape and scale parameters of the prior distribution. Also,
\begin{align*}
 &\frac{\partial}{\partial \bbeta} \log f_{\e^{\eta}}\left(\frac{y_i}{\mu_i}\right) - \log \mu_i = \e^\eta \left(\frac{y_i}{\mu_i} - 1\right) \, \bx_i, \cr
 &\frac{\partial}{\partial \eta} \log f_{\e^{\eta}}\left(\frac{y_i}{\mu_i}\right) - \log \mu_i = \e^\eta \left(-\frac{y_i}{\mu_i} + \log\left(\frac{y_i}{\mu_i}\right) + \eta + 1 - \frac{\Gamma'(\e^\eta)}{\Gamma(\e^\eta)}\right),
\end{align*}
where $\Gamma'$ is the derivative of $\Gamma$.

With the proposed robust GLM, the posterior density $\pi(\, \cdot \,, \cdot \mid \by)$ (with the same prior as before) is such that
\[
 \pi(\bbeta, \nu \mid \by) \propto \pi(\nu) \prod_{i = 1}^n \frac{1}{\mu_i} f_{\nu, c}\left(\frac{y_i}{\mu_i}\right).
\]
We also use Hamiltonian Monte Carlo to compute the posterior estimates.

The posterior density resulting from the change of variables $\eta := \log \nu$ is such that
\[
 \pi(\bbeta, \eta \mid \by) \propto \e^{\eta} \, \pi(\e^{\eta}) \prod_{i = 1}^n \frac{1}{\mu_i} f_{\e^{\eta}, c}\left(\frac{y_i}{\mu_i}\right).
\]
The log density is such that (if we omit the normalization constant)
\[
 \log \pi(\bbeta, \nu \mid \by) = \eta + \log \pi( \e^{\eta}) + \sum_{i = 1}^n \log f_{\e^{\eta}, c}\left(\frac{y_i}{\mu_i}\right) - \log \mu_i.
\]

We have that
\[
 \frac{\partial}{\partial \bbeta} \log f_{\e^{\eta}, c}\left(\frac{y_i}{\mu_i}\right) - \log \mu_i = \begin{cases}
    \e^\eta \left(\frac{y_i}{\mu_i} - 1\right) \, \bx_i \quad \text{if} \quad \zl \leq \frac{y_i}{\mu_i} \leq \zr, \cr
    \frac{\lambdar}{\log(y_i / \mu_i)} \, \bx_i \quad \text{if} \quad \frac{y_i}{\mu_i} > \zr, \cr
    \frac{\lambdal}{\log(y_i / \mu_i)} \, \bx_i \quad \text{if} \quad \frac{y_i}{\mu_i} < \zl,
 \end{cases}
\]
and
\[
 \frac{\partial}{\partial \eta} \log f_{\e^{\eta}, c}\left(\frac{y_i}{\mu_i}\right) - \log \mu_i = \begin{cases}
   \e^\eta \left(-\frac{y_i}{\mu_i} + \log\left(\frac{y_i}{\mu_i}\right) + \eta + 1 - \frac{\Gamma'(\e^\eta)}{\Gamma(\e^\eta)}\right) \quad \text{if} \quad \zl \leq \frac{y_i}{\mu_i} \leq \zr, \cr
    - \frac{c}{2} \, \e^{\eta / 2}\left(1 + \frac{1}{1 + c \, \e^{-\eta / 2}} \right) + \e^{\eta}\left(\log(1 + c \, \e^{-\eta / 2}) +\eta - \frac{\Gamma'(\e^\eta)}{\Gamma(\e^\eta)}\right) \cr
      \quad + \lambdar' \left(\log(\log(\zr)) - \log\left(\log\left(\frac{y_i}{\mu_i}\right)\right)\right) - \frac{\lambdar}{\log(1 + c \, \e^{-\eta / 2})} \frac{c \, \e^{-\eta / 2}}{2(1 + c \, \e^{-\eta / 2})} \quad \text{if} \quad \frac{y_i}{\mu_i} > \zr, \cr
     \frac{c}{2} \, \e^{\eta / 2}\left(1 + \frac{1}{1 - c \, \e^{-\eta / 2}} \right) + \e^{\eta}\left(\log(1 - c \, \e^{-\eta / 2}) +\eta - \frac{\Gamma'(\e^\eta)}{\Gamma(\e^\eta)}\right) \cr
      \quad + \lambdal' \left(\log(-\log(\zl)) - \log\left(-\log\left(\frac{y_i}{\mu_i}\right)\right)\right) + \frac{\lambdal}{\log(1 - c \, \e^{-\eta / 2})} \frac{c \, \e^{-\eta / 2}}{2(1 - c \, \e^{-\eta / 2})} \quad \text{if} \quad \frac{y_i}{\mu_i} < \zl,
 \end{cases}
\]
where $\lambdar' := \frac{\partial}{\partial \eta} \lambdar$ and $\lambdal' := \frac{\partial}{\partial \eta} \lambdal$ which are evaluated numerically.

\section{Robust inverse Gaussian GLM}\label{sec:robust_inv_GLM}

In this section, we propose a heavy-tailed version of inverse Gaussian GLM. We proceed in a similar fashion as in \autoref{sec:proposed}: we present the model in \autoref{sec:model_inv} and theoretical properties in \autoref{sec:properties_inv}; the proofs of theoretical results are deferred to \autoref{sec:proofs_inv}.

\subsection{Model definition}\label{sec:model_inv}

The inverse Gaussian PDF is such that
\[
 f_{\mu, \nu}(y) = \sqrt{\frac{\nu}{2 \pi y^3}} \exp\left(-\frac{\nu (y - \mu)^2}{2 \mu^2 y}\right), \quad y > 0,
\]
where $\mu > 0$ is the mean parameter and $\nu > 0$ is the shape parameter, the latter corresponding to the inverse of the dispersion parameter. In a context of GLM, the link function that is the most commonly used is the log function; it will be used in the rest of the section.

With the inverse Gaussian PDF, $\mu$ is not a scale parameter and neither does $\nu$. This implies that we cannot define a heavy-tailed version the same way we did with the gamma PDF. That being said, by considering the same transformation $Z := Y / \mu \sim g_{\mu, \nu}$, we will be able to proceed similarly. The difference is that the PDF of the transformed random variable still depends on $\mu$ and is as follows:
\[
  g_{\mu, \nu}(z) = \sqrt{\frac{\nu / \mu}{2\pi z^3}} \exp\left(-\frac{\nu / \mu}{2}\left(z - 1\right)^2 \frac{1}{z}\right), \quad z > 0.
\]
Interestingly, in both the polynomial and the exponential terms, we see a function of $z$ that is multiplied by $\nu / \mu$. Note that this is the PDF of an inverse Gaussian distribution with a mean parameter of $1$ and a shape parameter of $\nu / \mu$.

We thus base our proposal on this PDF, in the same spirit as the gamma variant. We assume that $Z_i := Y_i / \mu_i \sim f_{\mu_i, \nu, c}$, where the proposed PDF  $f_{\mu_i, \nu, c}$ is defined as
\begin{align*}
 f_{\mu_i, \nu, c}(z_i) := \begin{cases}
    f_{\text{mid}}(z_i) := \sqrt{\frac{\nu / \mu_i}{2\pi z_i^3}} \exp\left(-\frac{\nu / \mu_i}{2}\left(z_i - 1\right)^2 \frac{1}{z_i}\right) \quad \text{if} \quad \zl \leq z_i \leq \zr, \cr
    f_{\text{right}}(z_i) := f_{\text{mid}}(\zr) \frac{\zr}{z_i} \left(\frac{\log \zr}{\log z_i}\right)^{\lambdar} \quad \text{if} \quad z_i > \zr, \cr
    f_{\text{left}}(z_i) := f_{\text{mid}}(\zl) \frac{\zl}{z_i} \left(\frac{\log \zl}{\log z_i}\right)^{\lambdal} \quad \text{if} \quad 0 < z_i < \zl,
 \end{cases}
\end{align*}
where $\zr, \lambdar, \zl$ and $\lambdal$ are functions of $\mu_i > 0$, $\nu > 0$ and $c > 0$ given by
\begin{align*}
 & \zr := 1 +\sqrt{\frac{\mu_i}{\nu}} \, c, \quad \zl := \max\left\{0, 1 - \sqrt{\frac{\mu_i}{\nu}} \, c\right\}, \cr
 & \lambdar := 1 + \frac{f_{\text{mid}}(\zr) \log(\zr) \, \zr}{\P(Z_{\nu/\mu_i} > \zr)}, \quad \text{and} \quad \lambdal := 1 - \frac{f_{\text{mid}}(\zl) \log(\zl) \, \zl}{\P(Z_{\nu/\mu_i} < \zl)} = 1 + \frac{f_{\text{mid}}(\zl) \log(1 / \zl) \, \zl}{\P(Z_{\nu/\mu_i} < \zl)},
\end{align*}
with $Z_{\nu/\mu_i}$ being a random variable following an inverse Gaussian distribution whose mean and shape parameters are 1 and $\nu / \mu_i$, respectively. We observe a lot of similarities with the model defined in \autoref{sec:model}. In particular, we can obtain a result analogous to \autoref{prop:asymptotic_lambdas} about the behaviour of $\lambdar$ and $\lambdal$ (see \autoref{prop:asymptotic_lambdas_inv} and \autoref{fig:lambdas_inv}). We use the same notation to simplify.

\begin{Proposition}\label{prop:asymptotic_lambdas_inv}
 Viewed as functions of $\nu / \mu_i$, both $\lambdal$ and $\lambdar$ converge, as $\nu/\mu_i \rightarrow \infty$ for any fixed $c$, towards
\[
 1 + \frac{c \, \e^{-c^2 / 2}}{\sqrt{2 \pi} \, (1 - \Phi(c))}.
\]
\end{Proposition}

  \begin{figure}[ht]
  \centering
\includegraphics[width=0.45\textwidth]{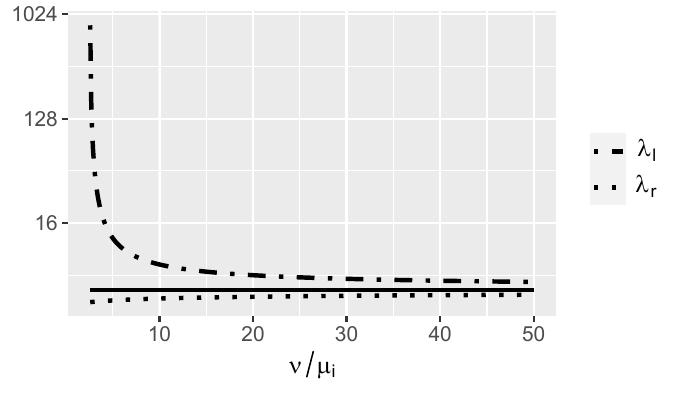}
\vspace{-2mm}
  \caption{\small $\lambdar$ and $ \lambdal$ as a function of $\nu / \mu_i$ when $c = 1.6$; the scale on the $y$-axis is
logarithmic; the black horizontal line represents the asymptotic value as $\nu / \mu_i \rightarrow \infty$.}\label{fig:lambdas_inv}
 \end{figure}
\normalsize

A difference between the model defined here and that in \autoref{sec:model} is that $\zl$ is defined differently; here, it is equal to $1 - \sqrt{\frac{\mu_i}{\nu}} \, c$ as soon as this is positive. This is because the inverse Gaussian PDF always has a left tail, in the sense that $f_{\mu, \nu}(y) \rightarrow 0$ as $y \rightarrow 0$, regardless of the value of $\mu$ and $\nu$.

Comparisons between inverse Gaussian PDFs (with mean and shape parameters of 1 and $\nu /\mu_i$, respectively) and $f_{\mu_i, \nu, c}$ with $c = 1.6$ are shown for different values of $\nu/\mu_i$ in \autoref{fig:inv_prop}. The observations are the same as in \autoref{sec:model}: both PDFs are globally quite similar, but beyond the threshold at which they start to be defined differently, $f_{\mu_i, \nu, c}$ first decreases slightly faster for a short interval (a consequence of the continuity of the function with a constraint of integrating to 1), after which $f_{\mu_i, \nu, c}$ goes above the inverse Gaussian PDF.

  \begin{figure}[ht]
  \centering\small
  $\begin{array}{ccc}
 \vspace{-2mm}\hspace{-2mm}\includegraphics[width=0.34\textwidth]{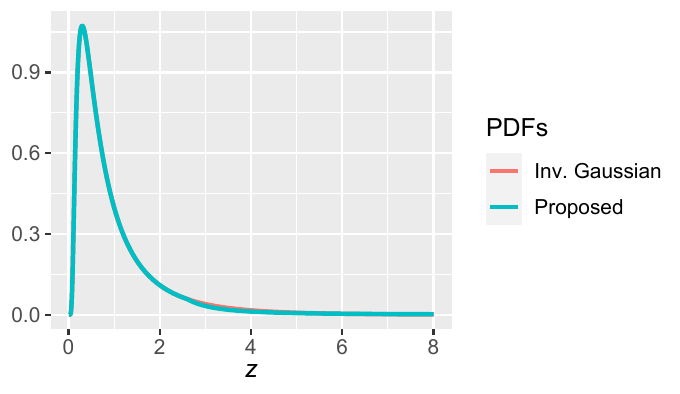} &  \hspace{-5mm} \includegraphics[width=0.34\textwidth]{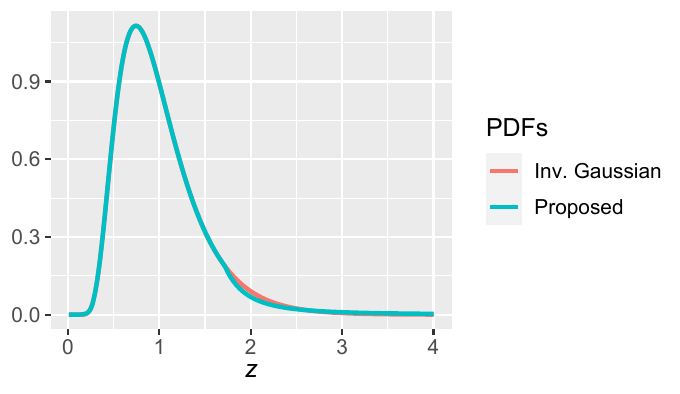} &  \hspace{-5mm} \includegraphics[width=0.34\textwidth]{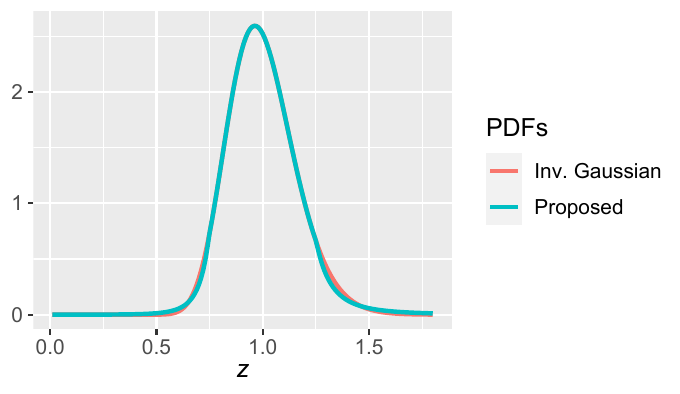} \cr
   \hspace{-5mm} \textbf{(a) $\nu / \mu_i = 1$} & \hspace{-7mm} \textbf{(b) $\nu / \mu_i = 5$} & \hspace{-5mm} \textbf{(c) $\nu / \mu_i = 40$} \cr
   \vspace{-2mm}\hspace{-2mm}\includegraphics[width=0.34\textwidth]{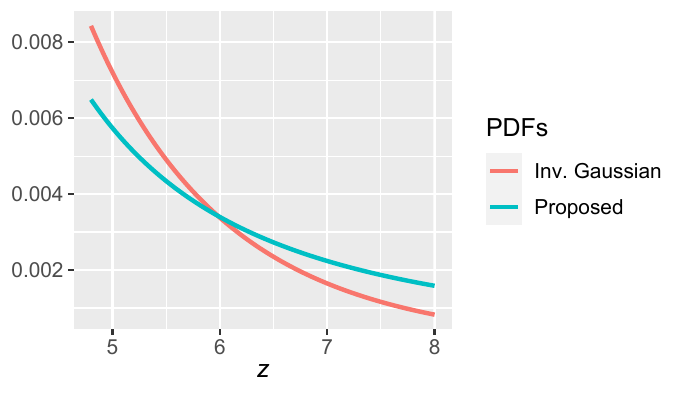} &  \hspace{-5mm} \includegraphics[width=0.34\textwidth]{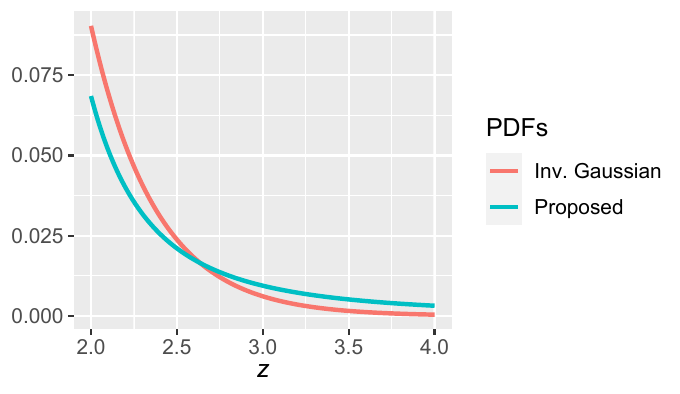} &  \hspace{-5mm} \includegraphics[width=0.34\textwidth]{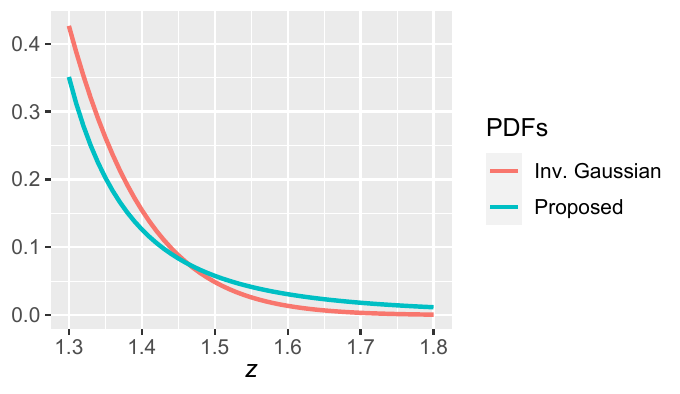} \cr
   \hspace{-5mm} \textbf{(d) $\nu / \mu_i = 1$ (zoom in right tail)} & \hspace{-7mm} \textbf{(e) $\nu / \mu_i = 5$ (zoom in right tail)} & \hspace{-5mm} \textbf{(f) $\nu / \mu_i = 40$ (zoom in right tail)} \cr
  \end{array}$\vspace{-2mm}
  \caption{\small Comparisons between inverse Gaussian PDFs and $f_{\mu_i,\nu, c}$ with $c = 1.6$, for different values of $\nu / \mu_i$.}\label{fig:inv_prop}
 \end{figure}
\normalsize

We finish this section by presenting results of a numerical experiment similar to that conducted for \autoref{fig:estimates_yn}. The simulation setting is the same, with the obvious difference that the original data set is simulated using a inverse Gaussian distribution. Next, we gradually increased the value of $y_n$ from 5 (a non-outlying value) to 15 (a clearly outlying value). For each data set associated with a different value of $y_n$, we estimated the parameters $\nu$, $\beta_1$ and $\beta_2$ of inverse Gaussian GLM and the proposed model based on maximum likelihood method. As expected, we observe a robustness problem with the estimation of inverse Gaussian GLM, a problem which is addressed with the proposed model. In \autoref{sec:properties_inv}, we present a theoretical result which allows to have a characterization of the robustness of the proposed model. Note that a robust estimation of inverse Gaussian GLM is not available through the \textsf{robustbase R} package.

  \begin{figure}[ht]
  \centering\small
  $\begin{array}{ccc}
 \vspace{-2mm}\hspace{-2mm}\includegraphics[width=0.34\textwidth]{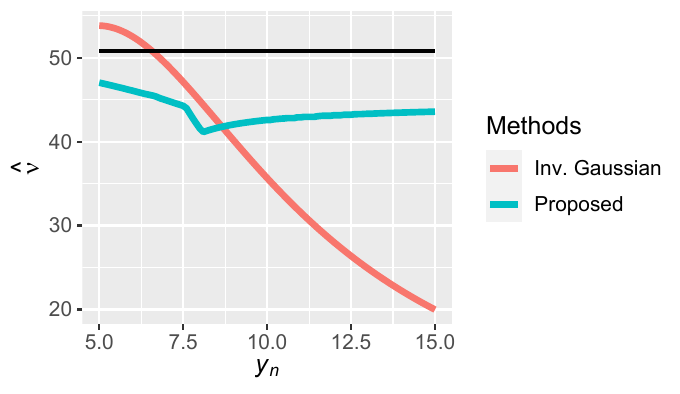} &  \hspace{-5mm} \includegraphics[width=0.34\textwidth]{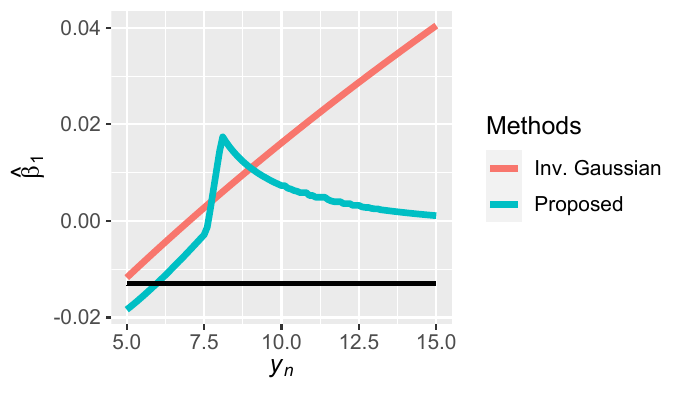} &  \hspace{-5mm} \includegraphics[width=0.34\textwidth]{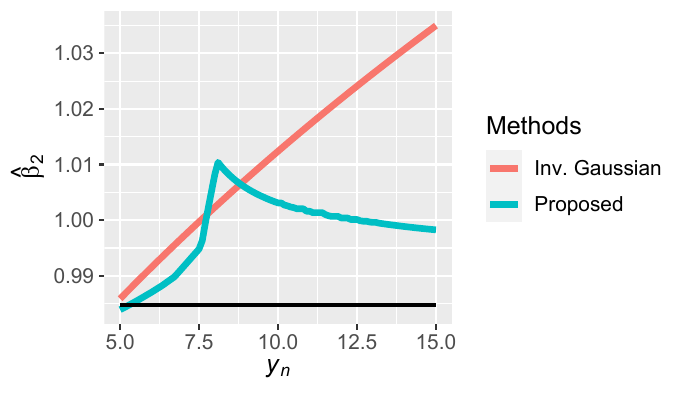} \cr
   \hspace{-5mm} \textbf{(a) $\hat{\nu}$ as a function of $y_n$} & \hspace{-7mm} \textbf{(b) $\hat{\beta}_1$ as a function of $y_n$} & \hspace{-5mm} \textbf{(c) $\hat{\beta}_2$ as a function of $y_n$}
  \end{array}$\vspace{-2mm}
  \caption{\small Estimates of $\nu$, $\beta_1$ and $\beta_2$ as a function of $y_n$ based on the estimation of inverse Gaussian GLM and the proposed method (with $c = 1.6$); the black horizontal lines represent the maximum likelihood estimates of inverse Gaussian GLM based on the data set excluding the outlier.}\label{fig:estimates_yn_inv}
 \end{figure}
\normalsize

\subsection{Theoretical properties}\label{sec:properties_inv}

In this section, we first consider a Bayesian framework and provide conditions under which the posterior distribution is proper (result analogous to \autoref{prop:proper}). We next provide a result about the limiting behaviour of the proposed PDF evaluated at an outlying data point (result analogous to \autoref{prop:limit_PDF}).

Let us start by setting our Bayesian framework. The framework is as in \autoref{sec:properties}. We consider that the explanatory-variable data points $\bx_1, \ldots, \bx_n$ are fixed and known, that is not realizations of random variables, contrarily to $y_1, \ldots, y_n$. The posterior distribution is thus conditional on the latter only. The prior distribution is denoted by $\pi(\, \cdot \,, \cdot \,)$. Let $\pi(\, \cdot \,, \cdot \mid \by)$ be the posterior distribution, where $\by := (y_1, \ldots, y_n)^T$. It is such that
\[
 \pi(\bbeta, \nu \mid \by) = \pi(\bbeta, \nu) \left[\prod_{i = 1}^n \frac{1}{\mu_i} f_{\mu_i, \nu, c}\left(\frac{y_i}{\mu_i}\right) \right] \Bigg/ m(\by), \quad \bbeta \in \R^p, \nu > 0,
\]
where
\[
 m(\by) := \int_{\R^p}\int_0^\infty \pi(\bbeta, \nu) \left[\prod_{i = 1}^n \frac{1}{\mu_i} f_{\mu_i, \nu, c}\left(\frac{y_i}{\mu_i}\right) \right] \d\nu \, \d\bbeta,
\]
if $m(\by) < \infty$, a situation where the posterior distribution is proper and thus well defined.

\begin{Proposition}\label{prop:proper_inv}
 Assume that $\pi(\, \cdot \,, \cdot \,)$ is a proper PDF such that
 \[
  \int_{\R^p}\int_0^\infty \pi(\bbeta, \nu) \, \nu^{n / 2} \prod_{i=1}^n (1 + \e^{-\bx_i^T \bbeta / 2}) \,  \d\nu \, \d\bbeta < \infty.
 \]
 Then, the posterior distribution is proper.
\end{Proposition}

The assumption in \autoref{prop:proper_inv} is satisfied when, for instance: i) $\bbeta$ and $\nu$ are \textit{a priori} independent, ii) the distribution of $\nu$ is a gamma with any shape and scale parameters, iii) the distribution of $\bbeta$ is a normal. When the prior distribution of $\bbeta$ is a normal,
\[
 \int_{\R^p} \prod_{i=1}^n (1 + \e^{-\bx_i^T \bbeta / 2}) \, \pi(\bbeta) \, \d\bbeta
\]
corresponds to a sum of moment generating functions of univariate normal distributions, which is finite.

We now provide a characterization of the robustness of the proposed model by considering the same asymptotic framework as in \autoref{sec:properties}. We provide in \autoref{prop:limit_PDF_inv} a result analogous to \autoref{prop:limit_PDF}.

\begin{Proposition}\label{prop:limit_PDF_inv}
For any $i$ with $l_i = 1$, and $c$, $\nu$ and $\mu_i$ fixed, we have that
\[
 \lim_{\omega \rightarrow \infty} \frac{ f_{\mu_i,\nu,c}(y_i/\mu_i)/\mu_i}{f_{\mu_i, \nu,c}(y_i)} = 1.
 \]
 If $\sqrt{\nu / \mu_i} > c$ (the condition under which $f_{\text{left}}$ exists), the same result holds for any $i$ with $s_i = 1$. Consequently, the likelihood function $\prod_{i = 1}^n f_{\mu_i, \nu,c}(y_i/\mu_i)/\mu_i$, when evaluated at $(\bbeta, \nu)$ such that $\sqrt{\nu / \mu_i} > c$ for all $i$ with $s_i =1$, asymptotically behave like
 \begin{align}\label{eq:limit_likelihood_inv}
  \prod_{i = 1}^n \left[\frac{1}{\mu_i} f_{\mu_i, \nu, c}\left(\frac{y_i}{\mu_i}\right)\right]^{k_i} \left[f_{\mu_i, \nu, c}(y_i)\right]^{s_i + l_i},
 \end{align}
 as $\omega \rightarrow \infty$, implying that, if the MLE belongs to a compact set such that $\sqrt{\nu / \mu_i} > c$ for all $i$ with $s_i =1$, then it corresponds asymptotically to the mode of \eqref{eq:limit_likelihood_inv}, provided that the latter belongs to a compact set such that $\sqrt{\nu / \mu_i} > c$ for all $i$ with $s_i =1$ as well.
\end{Proposition}

The difference with \autoref{prop:limit_PDF} is that limiting PDF depends on $\bbeta$ as well as $\nu$. This implies that we cannot state that the proposed model is partially robust, even though we empirically observe (as in \autoref{fig:estimates_yn_inv}) a certain degree of robustness. The asymptotic behaviour of $f_{\mu_i,\nu,c}(y_i/\mu_i)/\mu_i$ also indicates that we cannot prove a result like \autoref{thm:robustness} for the model proposed here (at least using the proof technique employed for \autoref{thm:robustness}). This is because we cannot write $f_{\mu_i,\nu,c}(y_i)$ as a product of two terms: one depending on $(\bbeta, \nu)$ but not on $y_i$, and the other one depending on $y_i$ but not on $(\bbeta, \nu)$.

To push further the analysis for a characterization of the robustness, we can analyse the behaviour of $f_{\mu_i, \nu,c}(y_i)$ when $\nu / \mu_i$ is large. This allows to understand the asymptotic behaviour of $f_{\mu_i,\nu,c}(y_i/\mu_i)/\mu_i$ as $\omega \rightarrow \infty$, that is the limiting behaviour of the proposed PDF evaluated at an outlying data point, when $\nu / \mu_i$ is large. Let us look at the case $l_i = 1$ (the case $s_i = 1$ can be analysed analogously):
\begin{align*}
 f_{\mu_i, \nu,c}(y_i) &= f_{\text{mid}}(\zr) \frac{\zr}{y_i} \left(\frac{\log \zr}{\log y_i}\right)^{\lambdar} \cr
 &= f_{\text{mid}}(\zr) \, \zr \log \zr \frac{1}{y_i} \frac{(\log \zr)^{\lambdar - 1}}{(\log y_i)^{\lambdar}}.
\end{align*}
With the proof of \autoref{prop:asymptotic_lambdas_inv}, we know that
\[
 f_{\text{mid}}(\zr) \, \zr \log \zr
 \]
 behaves like
  \[
  \frac{c}{\sqrt{2 \pi} \, \e^{c^2 / 2}},
 \]
 when $\nu / \mu_i$ is large and that $\lambdar - 1$ behaves like
 \[
  \frac{c \, \e^{-c^2 / 2}}{\sqrt{2 \pi} \, (1 - \Phi(c))} =: \overline{\lambdar}.
 \]
 Therefore, when $\nu / \mu_i$ is large, $f_{\mu_i, \nu,c}(y_i)$ behaves like
   \[
  \frac{c}{\sqrt{2 \pi} \, \e^{c^2 / 2}} \frac{1}{y_i} \frac{(\log \zr)^{\overline{\lambdar}}}{(\log y_i)^{\overline{\lambdar} + 1}},
 \]
 which depends on the parameters $\bbeta$ and $\nu$ only through
 \[
  \log \zr = \log\left(1 +\sqrt{\frac{\mu_i}{\nu}} \, c\right),
 \]
 which is a (slowly) decreasing function of  $\nu / \mu_i$. The fact the variation is slow explains, in our opinion, empirical results such as those shown in \autoref{fig:estimates_yn_inv}. Note that we can obtain an analogous result about the behaviour of $f_{\nu, c}(y_i)$ analysed in \autoref{sec:properties}; the result is consistent with what is observed in \autoref{fig:f_nu}.

\subsection{Proofs}\label{sec:proofs_inv}

In this section, we first present the proof of \autoref{prop:asymptotic_lambdas_inv}. Next, we present a lemma which is useful for the proof of \autoref{prop:proper_inv}, and next we present the proof of this latter result. The proof of \autoref{prop:limit_PDF_inv} is analogous to that of \autoref{prop:limit_PDF}.

\begin{proof}[Proof of \autoref{prop:asymptotic_lambdas_inv}]
  We prove the result for $\lambdar$. The result for $\lambdal$ is proved analogously. We have that
 \[
 \lambdar= 1+\frac{f_{\text{mid}}(\zr)\log(\zr) \, \zr}{\P[Z_{\phi} > \zr]},
 \]
 where $\phi := \nu/\mu_i$. We omit the dependence in $i$ as it is not important. We study the limit of $\lambdar$ as $\phi \rightarrow \infty$, for fixed $c$. The proof is similar as that of \autoref{prop:asymptotic_lambdas}.

 To prove the result, we will prove that the numerator of the fraction converges towards
 \[
  \frac{c}{\sqrt{2 \pi} \, \e^{c^2 / 2}},
 \]
 and that the denominator converges towards $1 - \Phi(c)$. We start with the analysis of the numerator:
 \begin{align*}
   f_{\text{mid}}(\zr)\log(\zr) \, \zr  &= \sqrt{\frac{\phi}{2\pi \zr^3}} \exp\left(-\frac{\phi}{2}\left(\zr - 1\right)^2 \frac{1}{\zr}\right) \, \log(\zr) \, \zr \cr
   &= \frac{\sqrt{\phi} \sqrt{1 + c / \sqrt{\phi}}}{2\pi} \exp\left(-\frac{c^2}{2} \frac{1}{1 + c / \sqrt{\phi}}\right) \, \log(1 + c / \sqrt{\phi}) \cr
   &= \frac{\sqrt{1 + c / \sqrt{\phi}}}{2\pi} \exp\left(-\frac{c^2}{2} \frac{1}{1 + c / \sqrt{\phi}}\right) \, \log\left[(1 + c / \sqrt{\phi})^{\sqrt{\phi}}\right] \cr
   &\rightarrow \frac{c}{2\pi} \exp\left(-\frac{c^2}{2}\right),
 \end{align*}
 using that $1 + c / \sqrt{\phi} \rightarrow 1$ and that $(1 + c / \sqrt{\phi})^{\sqrt{\phi}} \rightarrow \e^c$.  This concludes the proof that the numerator in $\lambda_{\text{r}}$ converges towards
 \[
  \frac{c}{\sqrt{2 \pi} \, \e^{c^2 / 2}}.
 \]

 We now turn to the proof that the denominator, $\P[Z_{\phi} > \zr]$, converges towards $1 - \Phi(c)$. We have that
 \begin{align*}
  \P[Z_{\phi} > \zr] = \P[X_{\phi, \phi^2} > \phi \zr] &= \P[X_{\phi, \phi^2} > \phi + c \sqrt{\phi}] \cr
  &= \P\left[\frac{X_{\phi, \phi^2} - \phi}{\sqrt{\phi}} > c \right],
 \end{align*}
 where $X_{\phi, \phi^2}$ follows an inverse Gaussian distribution whose mean and shape parameters are $\phi$ and $\phi^2$, respectively.

 We have the following equality in distribution:
 \begin{align*}
  \frac{X_{\phi, \phi^2} - \phi}{\sqrt{\phi}} = \sqrt{\frac{\lfloor \phi \rfloor}{\phi}} \sqrt{\lfloor \phi \rfloor} \left(\frac{1}{\lfloor \phi \rfloor}\sum_{k = 1}^{\lfloor \phi \rfloor} X_k - 1\right) + \frac{\tilde{X} - (\phi - \lfloor \phi \rfloor)}{\sqrt{\phi }},
 \end{align*}
 where $X_1, \ldots, X_{\lfloor \phi \rfloor}$ are independent random variables, each having an inverse Gaussian distribution with a mean and shape of 1, and $\tilde{X}$ follows an inverse Gaussian distribution whose mean and shape parameters are $\nu - \lfloor \nu \rfloor$ and $(\nu - \lfloor \nu \rfloor)^2$, respectively. By the central limit theorem,
 \[
  \sqrt{\lfloor \phi \rfloor} \left(\frac{1}{\lfloor \phi \rfloor}\sum_{k = 1}^{\lfloor \phi \rfloor} X_k - 1\right)
 \]
 converges in distribution towards a standard normal distribution. Also, $\tilde{X} / \sqrt{\phi }$ converges towards 0 with probability 1. Therefore, by Slutsky’s theorem, we have that
 \[
 \frac{X_{\phi, \phi^2} - \phi}{\sqrt{\phi}}
 \]
 converges in distribution towards a standard normal distribution, which concludes the proof.
\end{proof}

\begin{Lemma}\label{lemma1_inv}
Viewed as a function of $\bbeta$ and $\nu$, $(1 / \mu_i) f_{\mu_i, \nu, c}(y_i / \mu_i)$ is bounded by $B \sqrt{\nu} (1 + \e^{-\bx_i^T \bbeta / 2})$, where $B>0$ is a constant (with respect to $\bbeta$ and $\nu$).
\end{Lemma}

\begin{proof}
We have that
\begin{align}\label{eq:lemma1_inv}
 \frac{1}{\mu_i} \, f_{\mu_i, \nu, c}\left(\frac{y_i}{\mu_i}\right) = \begin{cases}
      \sqrt{\frac{\nu}{2\pi y_i^3}} \exp\left(-\frac{\nu (y_i - \mu_i)^2}{2 \mu_i^2 y_i}\right) \quad \text{if} \quad \zl \leq y_i / \mu_i \leq \zr, \cr
      f_{\text{mid}}(\zr) \frac{\zr}{y_i} \left(\frac{\log \zr}{\log y_i / \mu_i}\right)^{\lambdar} \quad \text{if} \quad y_i / \mu_i > \zr, \cr
      f_{\text{mid}}(\zl) \frac{\zl}{y_i} \left(\frac{\log \zl}{\log y_i / \mu_i}\right)^{\lambdal} \quad \text{if} \quad 0 < y_i / \mu_i < \zl.
 \end{cases}
\end{align}
We bound the function in the three cases.

Firstly, we choose $B$ such that
\[
 \sqrt{\frac{\nu}{2\pi y_i^3}} \exp\left(-\frac{\nu (y_i - \mu_i)^2}{2 \mu_i^2 y_i}\right) \leq B \sqrt{\nu} \leq B \sqrt{\nu} (1 + \e^{-\bx_i^T \bbeta / 2}),
\]
using that $\exp(-x)\leq 1$ for all $x \geq 0$.

Secondly, using that $y_i / \mu_i > \zr$,
\begin{align*}
 f_{\text{mid}}(\zr) \frac{\zr}{y_i} \left(\frac{\log \zr}{\log y_i / \mu_i}\right)^{\lambdar} &\leq f_{\text{mid}}(\zr) \frac{\zr}{y_i} \cr
 &=\sqrt{\frac{\nu / \mu_i}{2\pi \zr^3}} \exp\left(-\frac{\nu / \mu_i}{2}\left(\zr - 1\right)^2 \frac{1}{\zr}\right)  \frac{\zr}{y_i} \cr
 &= \frac{1}{\sqrt{2 \pi} y_i} \frac{1}{\sqrt{\zr}}\sqrt{\frac{\nu}{\mu_i}}\exp\left(-\frac{c^2}{2}\frac{1}{\zr}\right) \cr
 &\leq B \sqrt{\frac{\nu}{\mu_i}} \leq B \sqrt{\nu} (1 + \e^{-\bx_i^T \bbeta / 2}),
 \end{align*}
because $B$ can be chosen such that
\[
 \frac{1}{\sqrt{2 \pi} y_i} \leq B,
\]
$\zr \geq 1$, and
\[
 \exp\left(-\frac{c^2}{2}\frac{1}{\zr}\right) \leq 1.
\]

Finally, let us consider that $\zl = 1 - \sqrt{\frac{\mu_i}{\nu}} \, c$, which is the situation where the third case in \eqref{eq:lemma1_inv} can be activated. Using that $(y_i / \mu_i)^{-1} > \zr^{-1}$,
\begin{align*}
 f_{\text{mid}}(\zl) \frac{\zl}{y_i} \left(\frac{\log \zl}{\log y_i / \mu_i}\right)^{\lambdal} &\leq f_{\text{mid}}(\zl) \frac{\zl}{y_i} \cr
 &= \sqrt{\frac{\nu / \mu_i}{2\pi \zl^3}} \exp\left(-\frac{\nu / \mu_i}{2}\left(\zl - 1\right)^2 \frac{1}{\zl}\right) \frac{\zl}{y_i} \cr
 &= \frac{1}{\sqrt{2 \pi} y_i} \sqrt{\frac{\nu}{\mu_i}}\frac{1}{\sqrt{\zl}} \exp\left(-\frac{c^2}{2}\frac{1}{\zl}\right) \cr
 &\leq B \sqrt{\frac{\nu}{\mu_i}} \leq B \sqrt{\nu} (1 + \e^{-\bx_i^T \bbeta / 2}),
\end{align*}
because $B$ can be chosen such that
\[
 \frac{1}{\sqrt{2 \pi} y_i} \frac{1}{\sqrt{\zl}} \exp\left(-\frac{c^2}{2}\frac{1}{\zl}\right) \leq B,
\]
given that
\[
 \frac{1}{\sqrt{\zl}} \exp\left(-\frac{c^2}{2}\frac{1}{\zl}\right)
\]
is bounded for any $0 < \zl < 1$.
\end{proof}

\begin{proof}[Proof of \autoref{prop:proper_inv}]
 Using \autoref{lemma1_inv},
 \begin{align*}
  m(\by) &= \int_{\R^p}\int_0^\infty \pi(\bbeta, \nu) \left[\prod_{i = 1}^n \frac{1}{\mu_i} f_{\mu_i, \nu, c}\left(\frac{y_i}{\mu_i}\right) \right] \d\nu \, \d\bbeta \cr
  &\leq \int_{\R^p}\int_0^\infty \pi(\bbeta, \nu) \left[\prod_{i = 1}^n B \sqrt{\nu} (1 + \e^{-\bx_i^T \bbeta / 2}) \right] \d\nu \, \d\bbeta \cr
  &= B^n \int_{\R^p}\int_0^\infty \pi(\bbeta, \nu) \, \nu^{n / 2} \prod_{i=1}^n (1 + \e^{-\bx_i^T \bbeta / 2}) \,  \d\nu \, \d\bbeta < \infty,
 \end{align*}
 by assumption.
\end{proof}

\end{document}